\documentclass[acmsmall,review=false]{acmart}
\usepackage{acm-ec-26} 

\PassOptionsToPackage{sort}{natbib}

\author{Shangyuan Yang}
\author{Kirthevasan Kandasamy}

\usepackage{amsfonts,amsmath, mathtools}
\usepackage{amsthm, thmtools, thm-restate}
\usepackage{bm,bbm}
\usepackage{enumitem}
\usepackage{wrapfig}

\usepackage{booktabs}
\usepackage{enumitem}
\usepackage{graphicx}
\usepackage[noend]{algpseudocode}
\usepackage{algorithm}
\usepackage{subcaption}
\usepackage{kk, defns_data_sharing}

\title[Pairwise Exchanges of Freely Replicable Goods with Negative Externalities]{Pairwise
Exchanges of Freely Replicable Goods with Negative Externalities}

\begin{abstract}
We study a setting where a set of agents engage in \emph{pairwise exchanges} of \emph{freely replicable goods} (\eg digital goods such as data), where two agents grant each other a copy of a good they possess in exchange for a good they lack. Such exchanges introduce a fundamental tension: while agents benefit from acquiring additional goods, they incur \emph{negative externalities} when others do the same. This dynamic typically arises in real-world scenarios where competing entities may benefit from selective collaboration.
For example, in a data sharing consortium, pharmaceutical companies might share (copies of) drug discovery data, when the value of accessing a competitor’s data outweighs the risk of revealing their own.

In our model, an altruistic central planner wishes to design an \emph{exchange protocol} (without money), to structure such exchanges between agents. The protocol operates over multiple rounds, proposing sets of pairwise exchanges in each round, which agents may accept or reject.
We formulate three key desiderata for such a protocol:  
\emph{(i) individual rationality:} agents should not be worse off by participating in the protocol;  
\emph{(ii) incentive-compatibility:} agents should be incentivized to share as much as possible by accepting all exchange proposals by the planner;  
\emph{(iii) stability:} there should be no further mutually beneficial exchanges upon termination.  
We design an exchange protocol for the planner that satisfies all three desiderata.

While the above desiderata are inspired by classical models for exchange, free-replicability and negative externalities necessitate novel and nontrivial reformalizations of these goals. We also argue that achieving Pareto-efficient agent utilities---often a central goal in exchange models without externalities---may be ill-suited in this setting.

\end{abstract}

\begin{document}

\maketitle

\setcounter{tocdepth}{1}
\tableofcontents


\newcommand{\insertFigNaiveNotNic}{

\begin{figure}[t]
    \centering
    \vspace{-0.2in}
    \begin{subfigure}[b]{0.40\textwidth}
        \includegraphics[width=\textwidth]{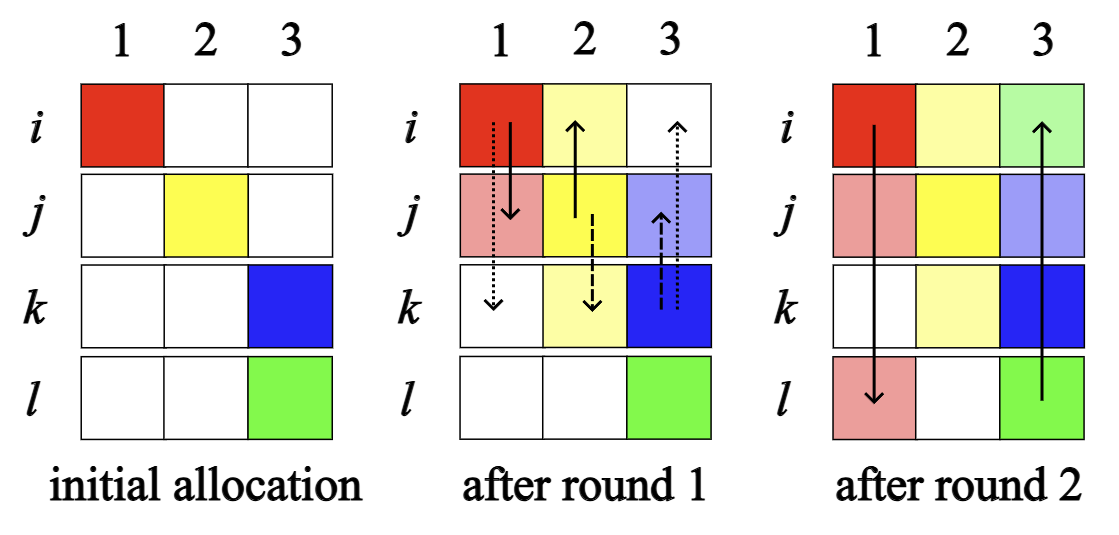} \\[-0.22in]
        \caption{\emph{An illustration of our initial attempt.}}
        \label{fig:naivesub}
    \end{subfigure}
    \hspace{0.3in}
    \begin{subfigure}[b]{0.50\textwidth}
        \includegraphics[width=\textwidth]{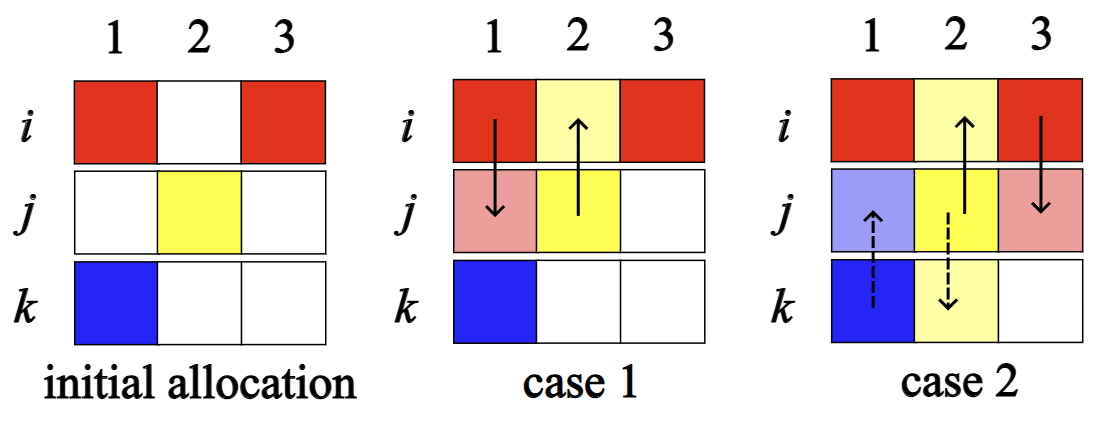}\\[-0.2in]
        \caption{\emph{Agents rejecting proposals to obtain rarer goods.}}
        \label{fig:notnicsub}
    \end{subfigure}
        \vspace{-0.12in}
    \caption{
        \emph{(\subref{fig:naivesub})}
        There are 4 agents and 3 goods with $\initdata_i = [1, 0, 0]$, $\initdata_j = [0, 1, 0]$, and $\initdata_k = \initdata_l = [0, 0, 1]$.
        In round 1, the protocol proposes three exchanges $P_1 = \{((i, 1), (j, 2)), ((i, 1), (k, 3)), ((j, 2), (k, 3))\}$; if all are accepted, it results in a stable allocation.
        However, suppose $((i, 1), (k, 3))$ was rejected by either agent.
        In round 2, the protocol proposes $P_2 = \{((i, 1), (l, 3))\}$. If accepted, this results in the stable allocation shown. 
        Recall, we disallow agents from onward-sharing goods they received from others in exchange for more goods.
        \\
        \emph{(\subref{fig:notnicsub})}
        Consider 3 agents $i, j, k$ and initial goods $\initdata_i = [1, 0, 1]$, $\initdata_j = [0, 1, 0]$, $\initdata_k = [1, 0, 0]$.
        Suppose the protocol proposes $P_1 = \{((i, 1), (j, 2))\}$, which results in a stable allocation if accepted. However, agent $j$ may reject it, as she would gain only one additional good in the final allocation. In particular, she may hope to get (the rarer) good 3 from $i$ and good 1 from $k$, for a total of two goods.
        However, had the protocol instead proposed $P_1 = \{((i, 3), (j, 2)), ((j, 2), (k, 1))\}$, then $j$ will accept both proposals since she gains two goods.  
        Similarly, $k$ will accept her proposal as she gains a good, and $i$ will also accept (provided that both $j$ and $k$ follow the accepting policy), since rejecting it would result in not receiving good 2.
        \vspace{-0.15in}
    }
    \label{fig:naivenotnic}
\end{figure}

}

\newcommand{\insertFigIR}{
    \begin{figure}[t]
    \centering
    \begin{subfigure}[b]{0.20\textwidth}
        \includegraphics[width=\textwidth]{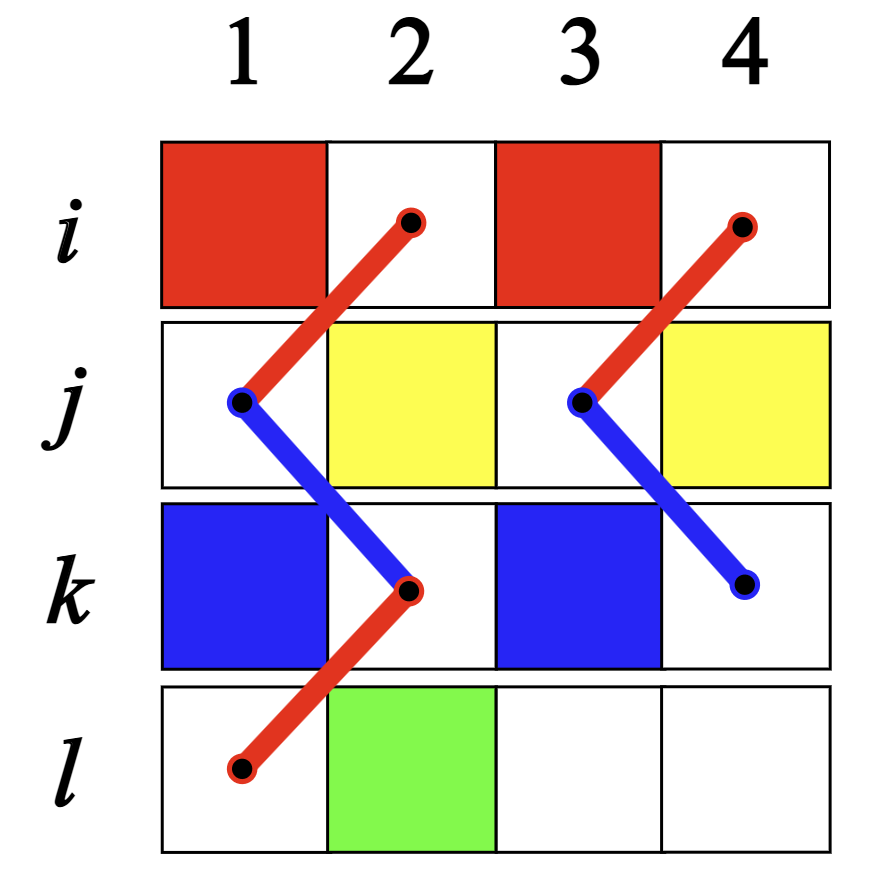}
        \caption{}
        \label{fig:IRsub1}
    \end{subfigure}
    \hspace{1.0in}
    \begin{subfigure}[b]{0.20\textwidth}
        \includegraphics[width=\textwidth]{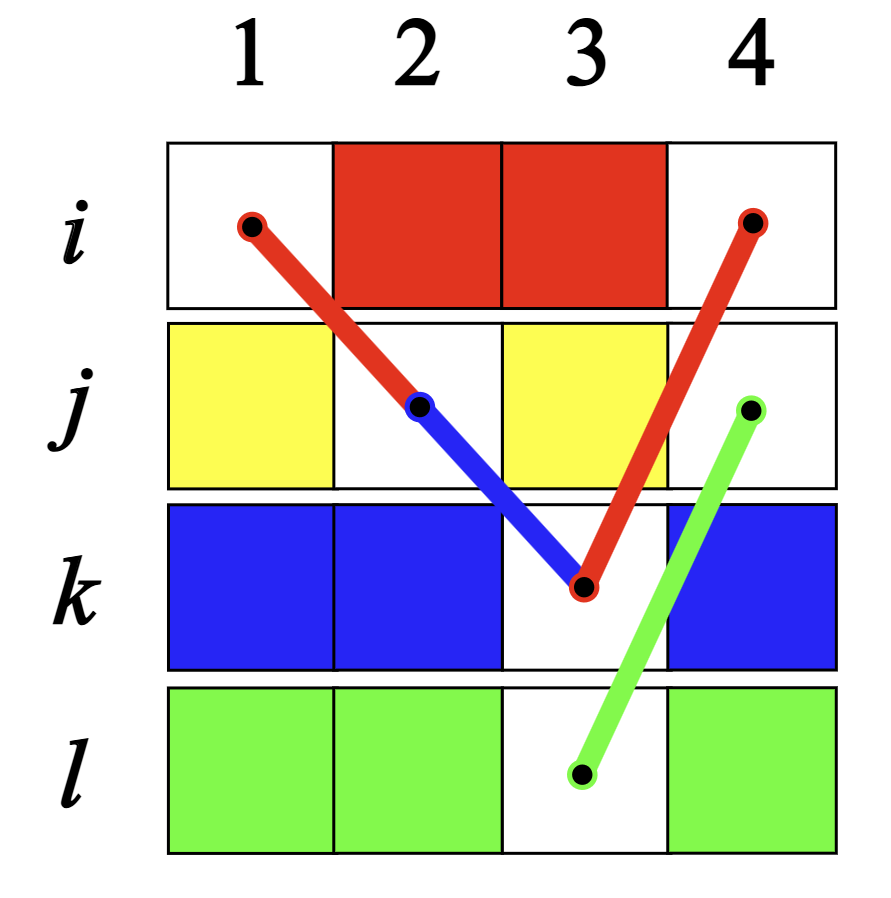}
        \caption{}
        \label{fig:IRsub2}
    \end{subfigure}
    \vspace{-0.1in}
    \caption{
        Illustration of \emph{tracked paths} for the IR proof via two separate instances (\subref{fig:IRsub1}) and (\subref{fig:IRsub2}).
        In both instances, we have 4 goods and 4 agents $\{i,j,k,l\}$ with $\beta_{ij} < \beta_{ik} < \beta_{il} < \beta_{jk} < \beta_{jl} < \beta_{kl}$.
        Edges only in $E$ are red, edges only in $E'$ are blue, and edges in both are green. 
        In instance (\subref{fig:IRsub1}),
        if $i$ participates, \algname{} would propose $E = \{\{(i, 2), (j, 1)\}, \{(i, 4), (j, 3)\}, \{(k, 2), (l, 1)\}\}$. If $i$ does not participate, \algname{} would propose $E' = \{\{(j, 1), (k, 2)\}, \{(j, 3), (k, 4)\}\}$. The path along goods 1 and 2 is a $\Tcal_2$ tracked path (see Step 3 of the IR proof), where $i$ gains $1-\beta_{il}$ from participating. The path along goods 3 and 4 is a $\Tcal_1$ tracked path, where $i$ gains $1 + \beta_{ik}$ from participating. 
        In (\subref{fig:IRsub2}),
        if $i$ participates, \algname{} would propose $E = \{\{(i, 1), (j, 2)\}, \{(i, 4), (k, 3)\}, \{(j, 4), (l, 3)\}\}$. If $i$ does not participate, \algname{} would propose $E' = \{\{(j, 2), (k, 3)\}, \{(j, 4), (l, 3)\}\}$. All edges except $\{(j, 4), (l, 3)\}$ form a $\Tcal_3$ tracked path, where $i$ gains 2 utility by participating. Finally, $\{(j, 4), (l, 3)\}$ is an exchange unaffected by $i$'s participation, and it is disjoint from any tracked paths.
    }
    \label{IR-proof-graph}
    \end{figure}
}

\newcommand{\insertFigIRIntro}{
    \begin{figure}[t]
    \centering
    \begin{subfigure}[b]{0.20\textwidth}
        \includegraphics[width=\textwidth]{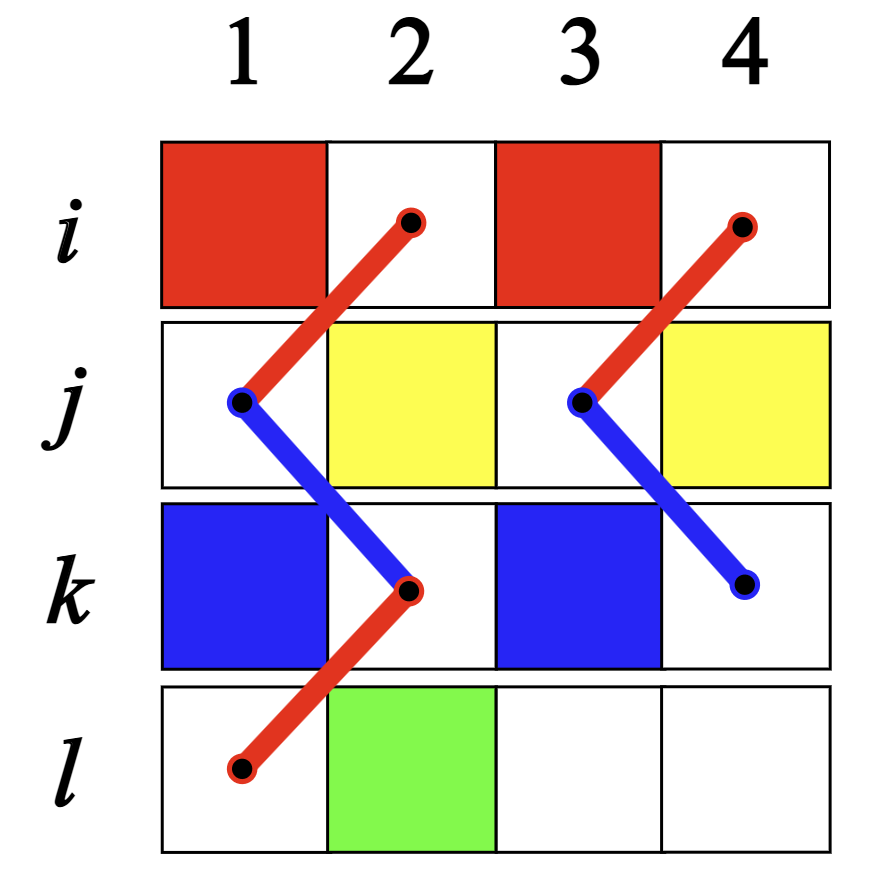}
        \caption{}
        \label{fig:IRsub1-intro}
    \end{subfigure}
    \hspace{1.0in}
    \begin{subfigure}[b]{0.20\textwidth}
        \includegraphics[width=\textwidth]{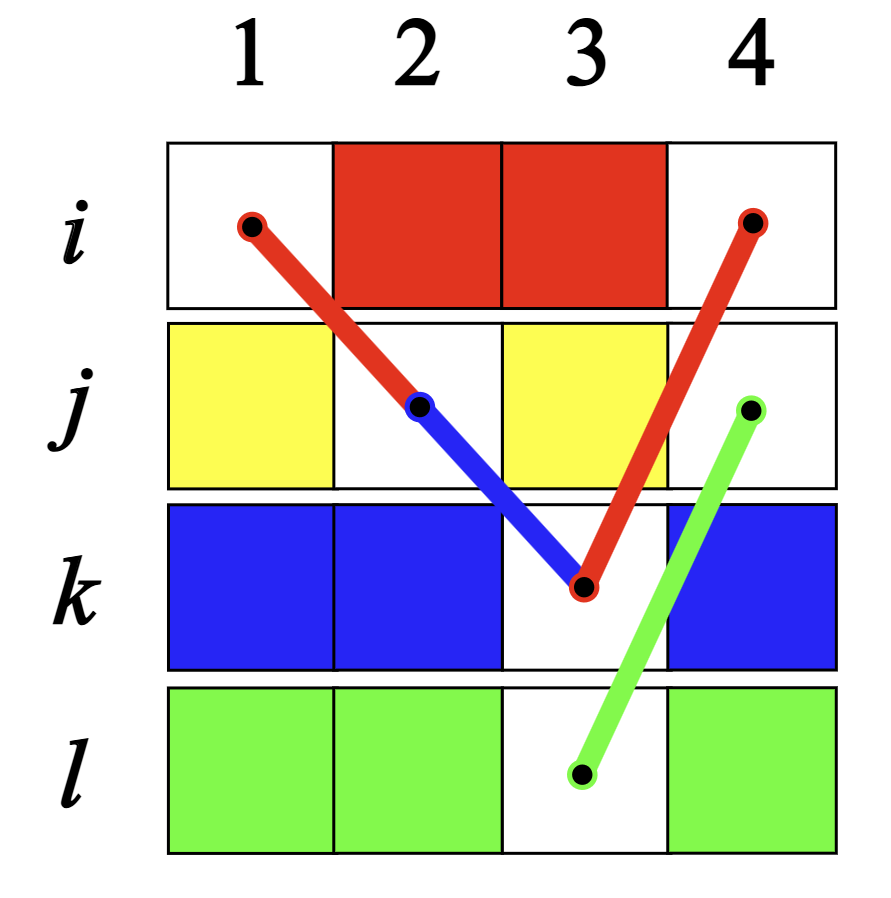}
        \caption{}
        \label{fig:IRsub2-intro}
    \end{subfigure}
    \vspace{-0.1in}
    \caption{
        Illustration of our graph construction in the IR proof on two instances (a) and (b).
        In both instances, we have 4 goods and 4 agents $\{i,j,k,l\}$ with $\beta_{ij} < \beta_{ik} < \beta_{il} < \beta_{jk} < \beta_{jl} < \beta_{kl}$.
            A shaded square means the agent originally has that good;
        for instance, in~(\subref{fig:IRsub1-intro}), agent $i$'s initial allocation is $\initdatai = [1, 0, 1, 0]$.
        Each agent-good pair is a vertex.
        Edges represent exchanges; for instance, the edge between $(i, 2)$ and $(j, 1)$ in~(\subref{fig:IRsub1-intro}) represents the exchange $((i, 1), (j, 2))$ where $i$ \emph{receives} a copy of good $2$ from $j$, while $j$ receives $1$.
        In instance~(\subref{fig:IRsub1-intro}),
        if $i$ participates, \algname{} would propose $E = \{\{(i, 2), (j, 1)\}, \{(i, 4), (j, 3)\}, \{(k, 2), (l, 1)\}\}$. If $i$ does not participate, \algname{} would propose $E' = \{\{(j, 1), (k, 2)\}, \{(j, 3), (k, 4)\}\}$.
        Edges only in $E$ are red, edges only in $E'$ are blue, and edges in both are green. 
        The paths which have alternating red and blue edges are tracked paths.
        \vspace{-0.1in}
    }
    \label{IR-proof-graph-intro}
    \end{figure}
}

\newcommand{\insertFigTree}{
    \begin{figure}
        \centering
        \begin{subfigure}[b]{0.48\textwidth}
            \includegraphics[width=\textwidth]{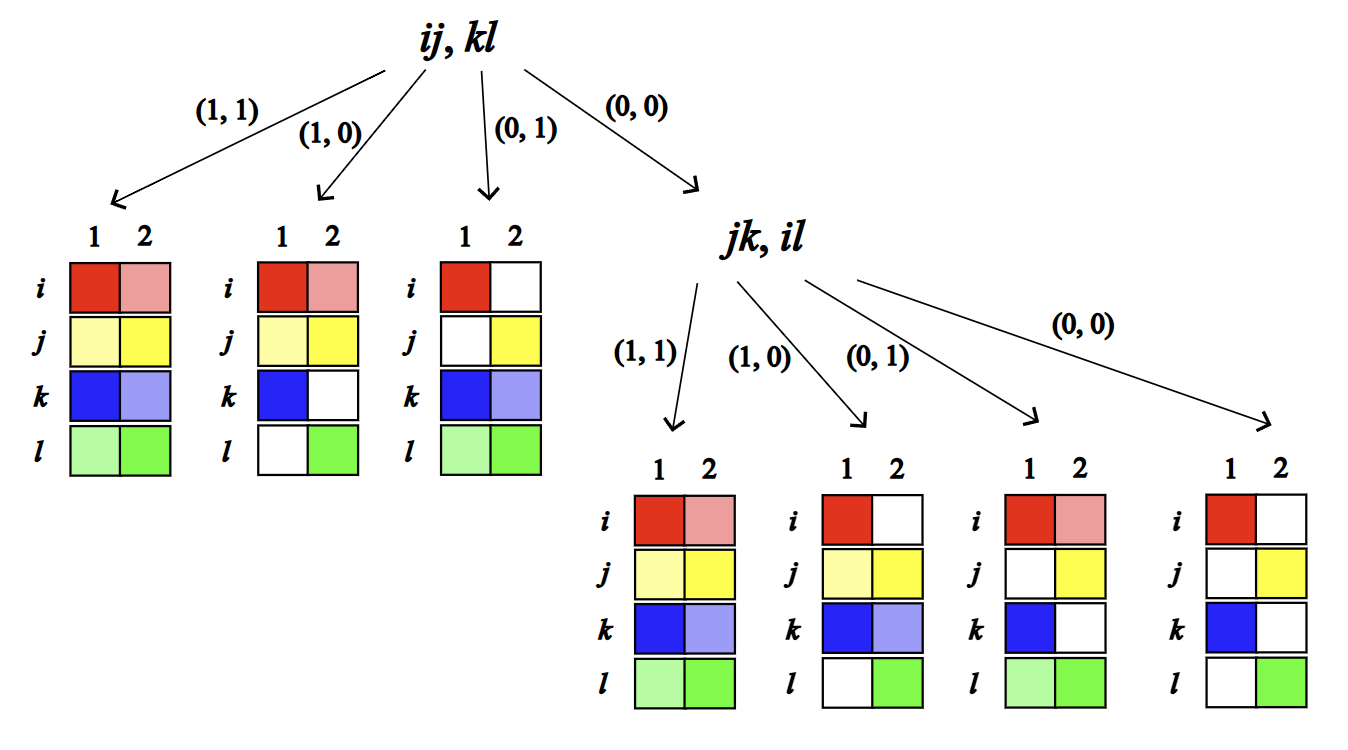}
            \caption{}
            \label{fig:tree-sub1}
        \end{subfigure}
        \begin{subfigure}[b]{0.50\textwidth}
            \includegraphics[width=\textwidth]{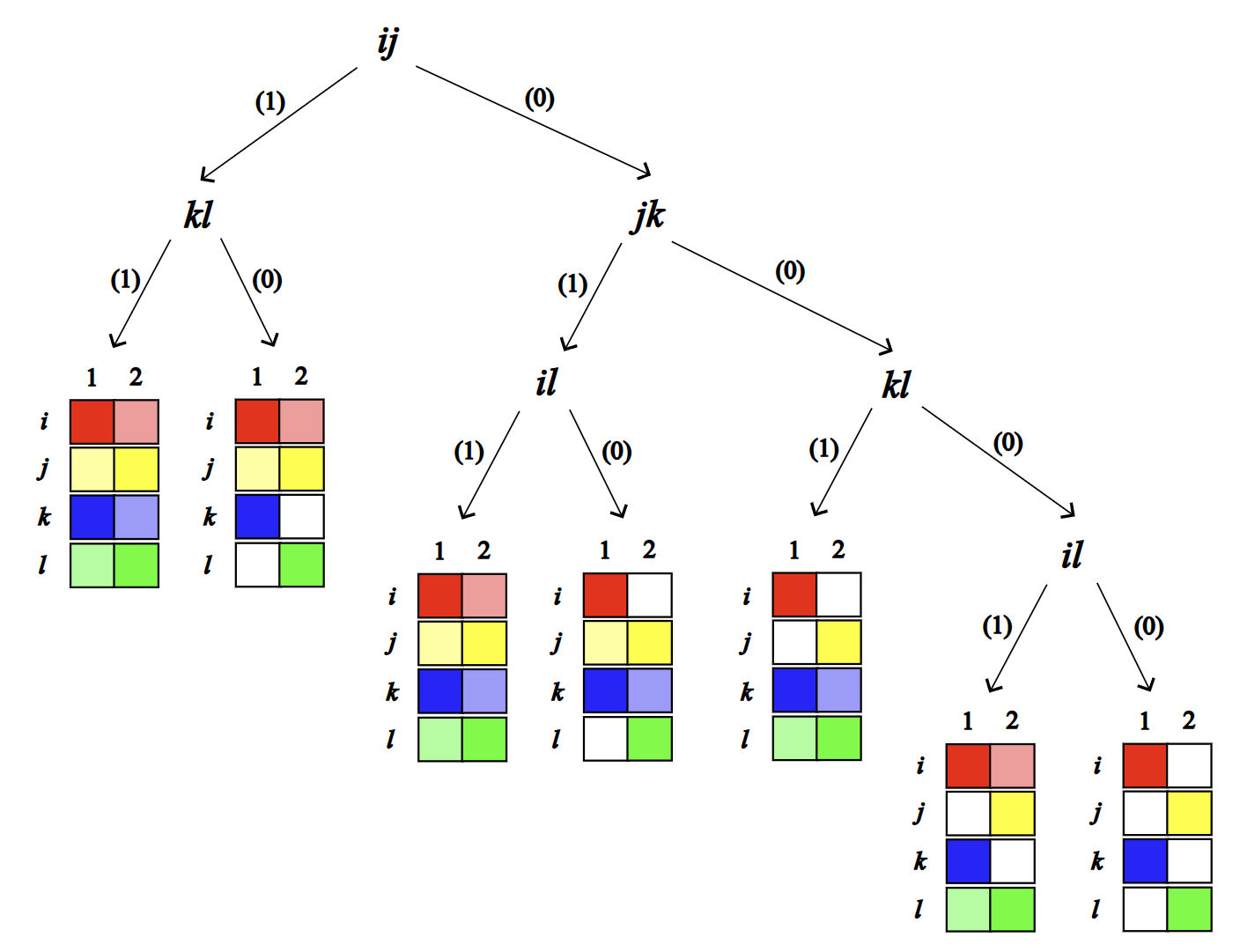}
            \caption{}
            \label{fig:tree-sub2}
        \end{subfigure}
        \vspace{-0.1in}
        \caption{
        An illustration of trees representing 
        stable protocols on the problem instance $(\initdata, \beta)$ in the proof of Theorem~\ref{thm:nodsic}.
         The ordered pair on each edge encodes acceptance/rejection of proposals. \emph{(\subref{fig:tree-sub1})} is the tree for \algname{}, which proposes exchanges $ij$ and $kl$ in round 1 and terminates unless both exchanges are rejected. If so, it proposes exchanges $jk$ and $il$ in round 2, and is guaranteed to terminate by then.
         \emph{(\subref{fig:tree-sub2})} is the tree for the SAC protocol,
         which proposes exchanges one at a time in ascending order of competition factors. If 2 exchanges are not simultaneously feasible ($ij$ and $jk$ for example), then the less prioritized exchange ($jk$) can only be scheduled once the more prioritized exchange ($ij$) is rejected.
        }
        \label{decision-tree}
    \end{figure}
}

\newcommand{\insertFigTreeIntro}{
    \begin{figure}
        \centering
        \begin{subfigure}[b]{0.48\textwidth}
            \includegraphics[width=\textwidth]{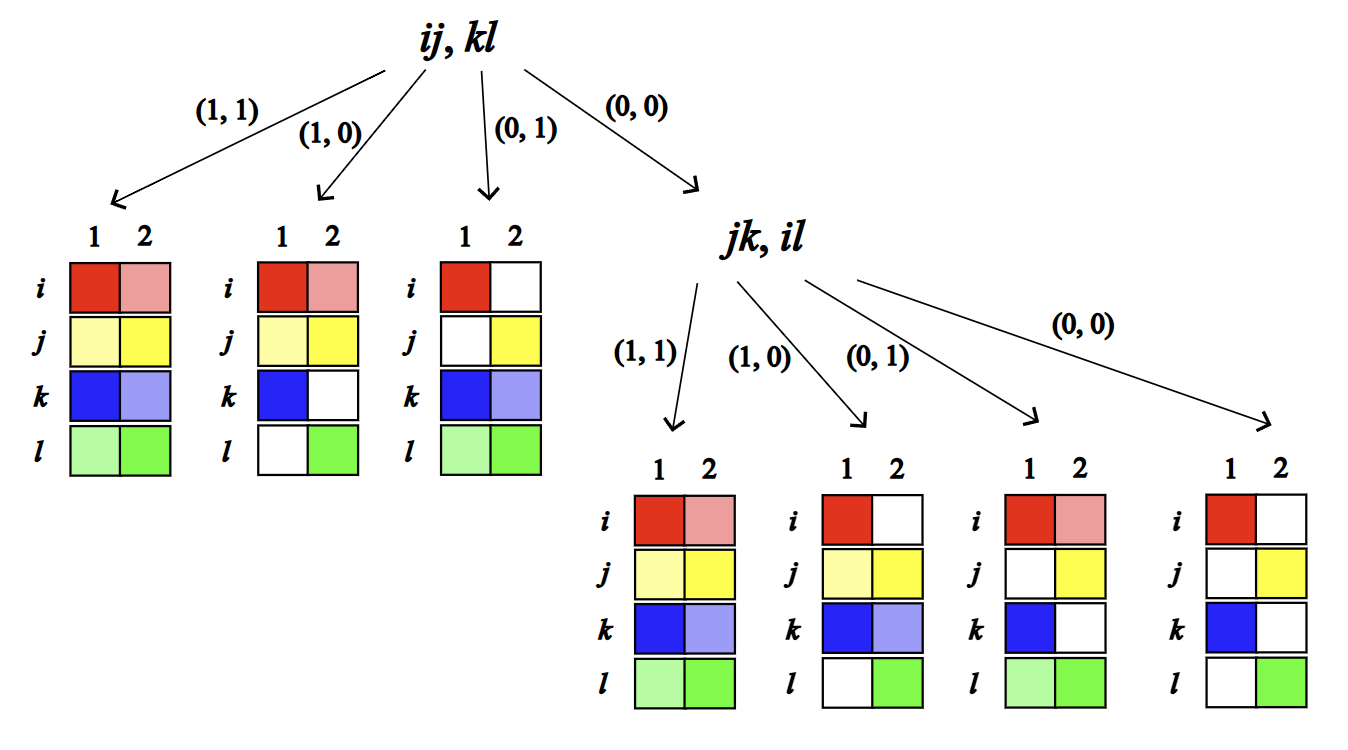}
            \caption{}
            \label{fig:tree-sub1-intro}
        \end{subfigure}
        \begin{subfigure}[b]{0.50\textwidth}
            \includegraphics[width=\textwidth]{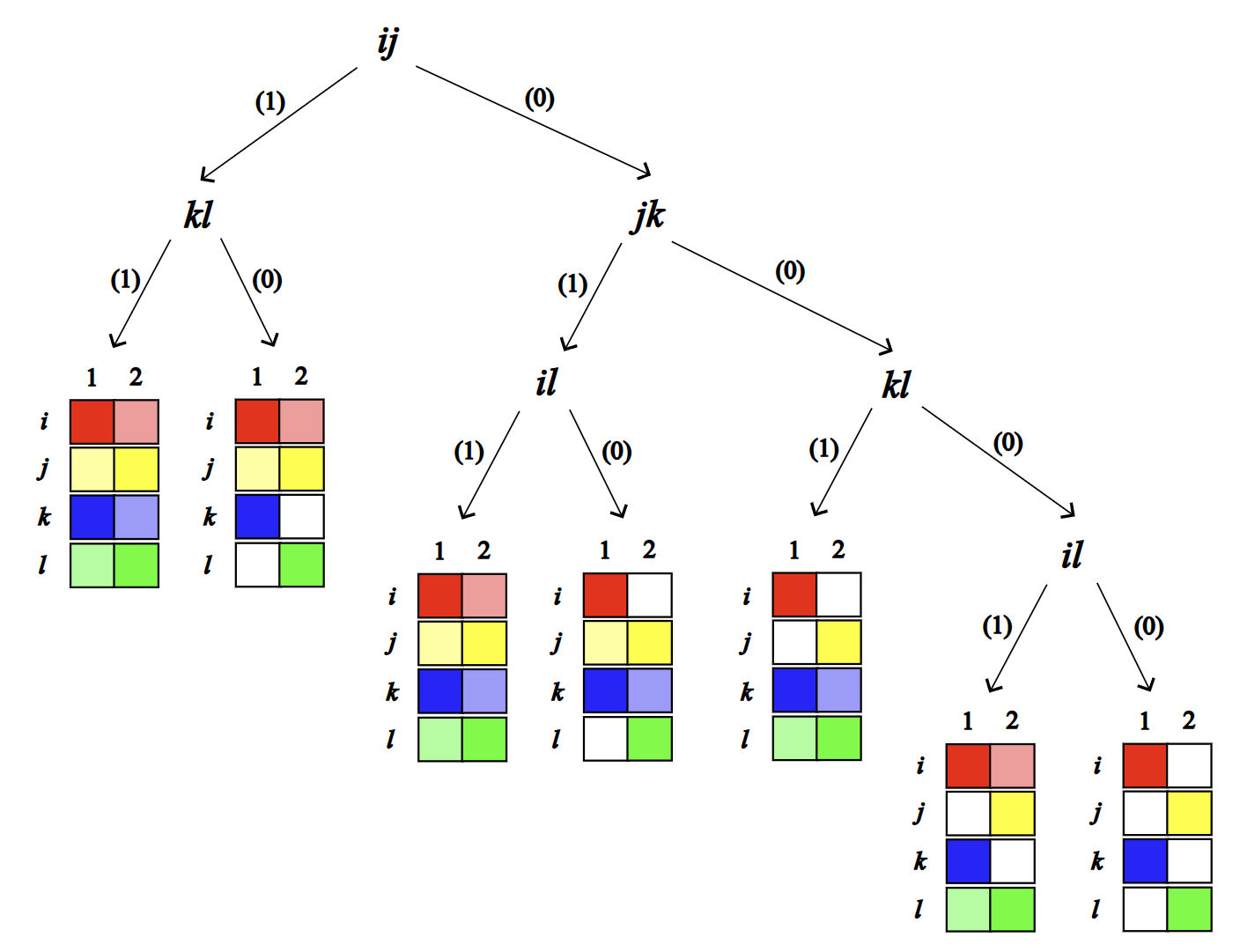}
            \caption{}
            \label{fig:tree-sub2-intro}
        \end{subfigure}
        \vspace{-0.1in}
        \caption{
        An illustration of trees representing stable protocols on the problem instance in the proof of Theorem~\ref{thm:nodsic}.
        The internal nodes represent proposals by the protocol.
        The ordered pair on each edge encodes acceptance/rejection of proposals.
        For example, the protocol in \emph{(\subref{fig:tree-sub1-intro})}  proposes exchanges $ij$ and $kl$ in round 1 and terminates unless both exchanges are rejected. If so, it proposes exchanges $jk$ and $il$ in round 2.
        The protocol in \emph{(\subref{fig:tree-sub2-intro})} is the only protocol on this problem instance which does not admit an inversion pair.
        \vspace{-0.25in}
        }
        \label{fig:decision-tree-intro}
    \end{figure}
}

\newcommand{\insertFigPO}{
    \begin{figure}[t]
        \centering
        \begin{subfigure}[b]{0.45\textwidth}
            \includegraphics[width=\textwidth]{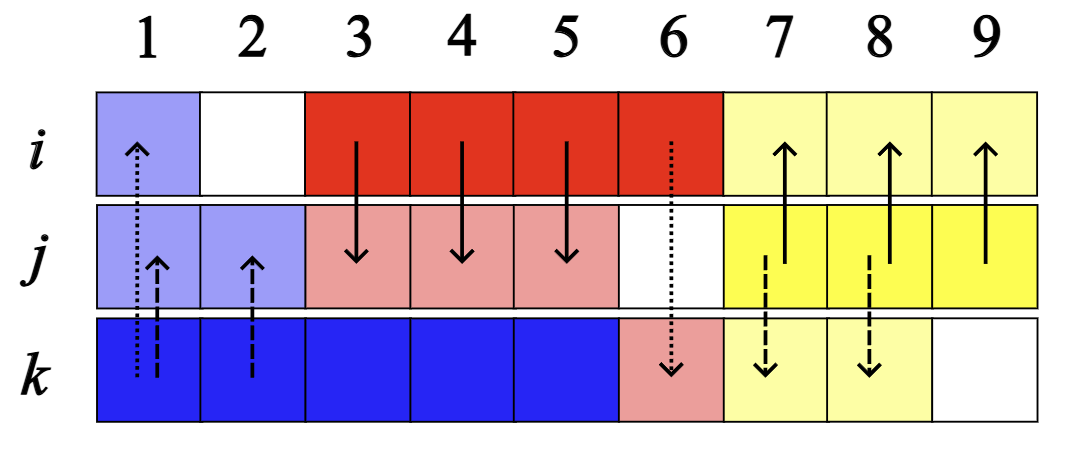}
            \caption{}
            \label{fig:PObad}
        \end{subfigure}
        \hspace{0.3in}
        \begin{subfigure}[b]{0.45\textwidth}
            \includegraphics[width=\textwidth]{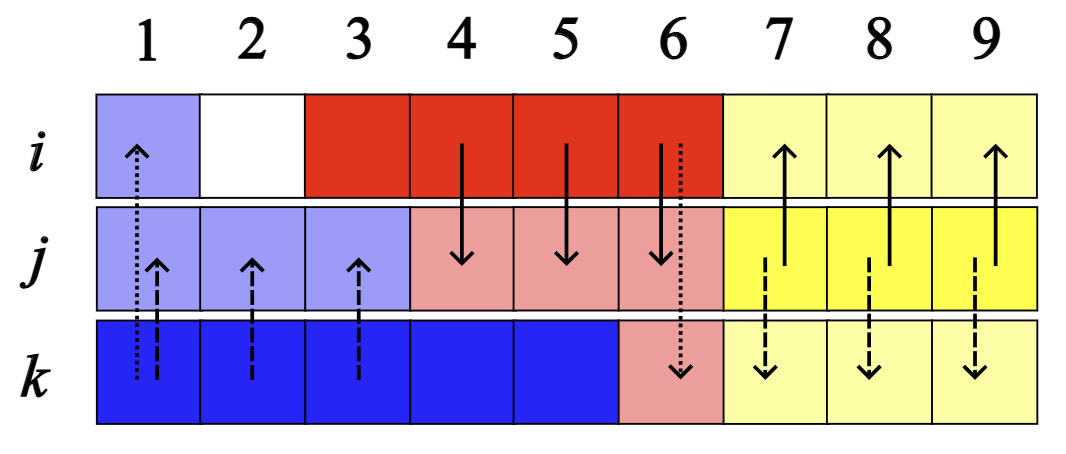}
            \caption{}
            \label{fig:POgood}
        \end{subfigure}
        \vspace{-0.10in}
        \caption{
        (\subref{fig:PObad}) An illustration that stability does not imply PE, even in LECs. Assume all agents follow the accepting policy. If a protocol proposes the illustrated exchanges, the resulting allocation is stable. However, it is not PE, as giving each agent one additional good would improve everyone's utility.
        (\subref{fig:POgood}) The \retrospect{} subroutine modifies the proposed exchanges as shown, resulting in a Pareto-efficient allocation, since at least one agent receives all goods (see Lemma~\ref{perfect-agent-PE}).
}
        \label{fig:POexample}
    \end{figure}
}

\newcommand{\insertPairwiseIllus}{
\begin{wrapfigure}{r}{0.42\textwidth}
    \vspace{-0.22in}
    \centering
    \includegraphics[width=0.41\textwidth]{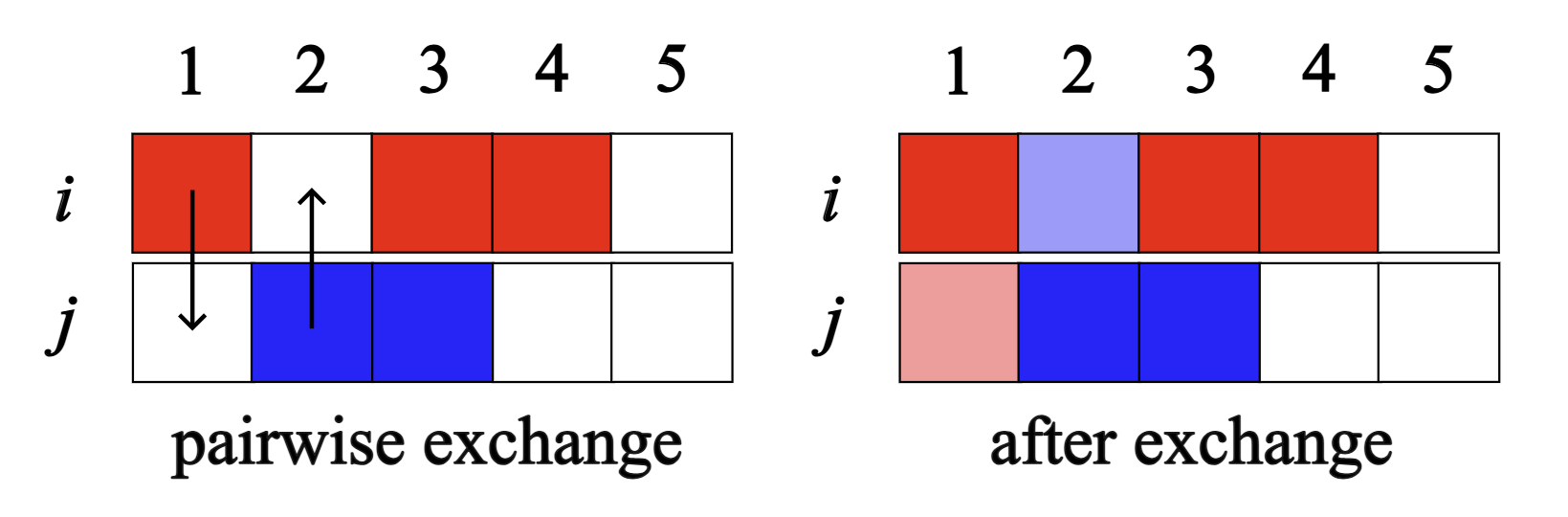}
    \vspace{-0.18in}
    \caption{\small Agent~$i$ holds goods 1,3, and 4, while $j$ holds goods 2 and 3.  
    In a pairwise exchange, $i$ could give good 1 to $j$ and receive good 2.  
    Both agents then retain all previously held goods and also acquire the newly received good.
    \vspace{-0.2in}}
    \label{fig:pairwiseillus}
\end{wrapfigure}
}

\newcommand{\insertFigIRIllus}{
\begin{wrapfigure}{r}{0.3\textwidth}
    \vspace{-0.35in}
    \centering
    \includegraphics[width=0.29\textwidth]{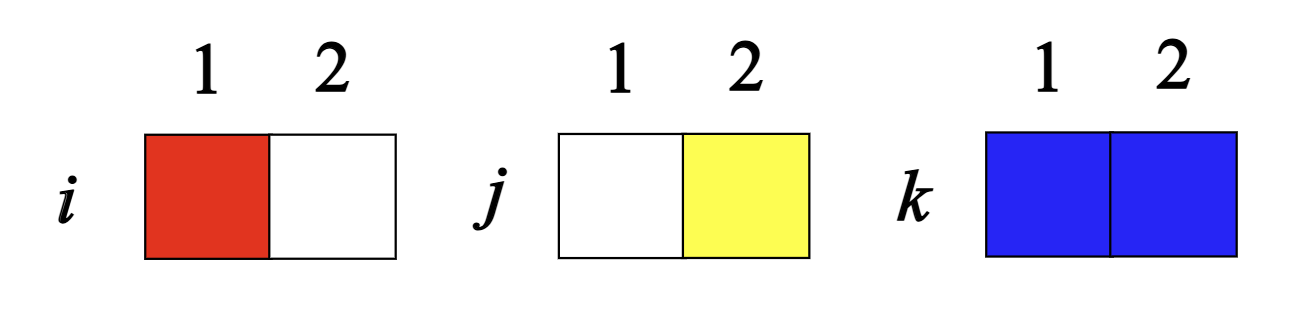}
    \vspace{-0.1in}
    \caption{\small Agent~$k$ holds both goods, while $i$ and $j$ hold only one each.
    \vspace{-0.3in}}
    \label{fig:irillus}
\end{wrapfigure}
}

\newcommand{\insertFigIRProofIllus}{
\begin{wrapfigure}{r}{0.25\textwidth}
    \vspace{-0.3in}
    \centering
    \includegraphics[width=0.24\textwidth]{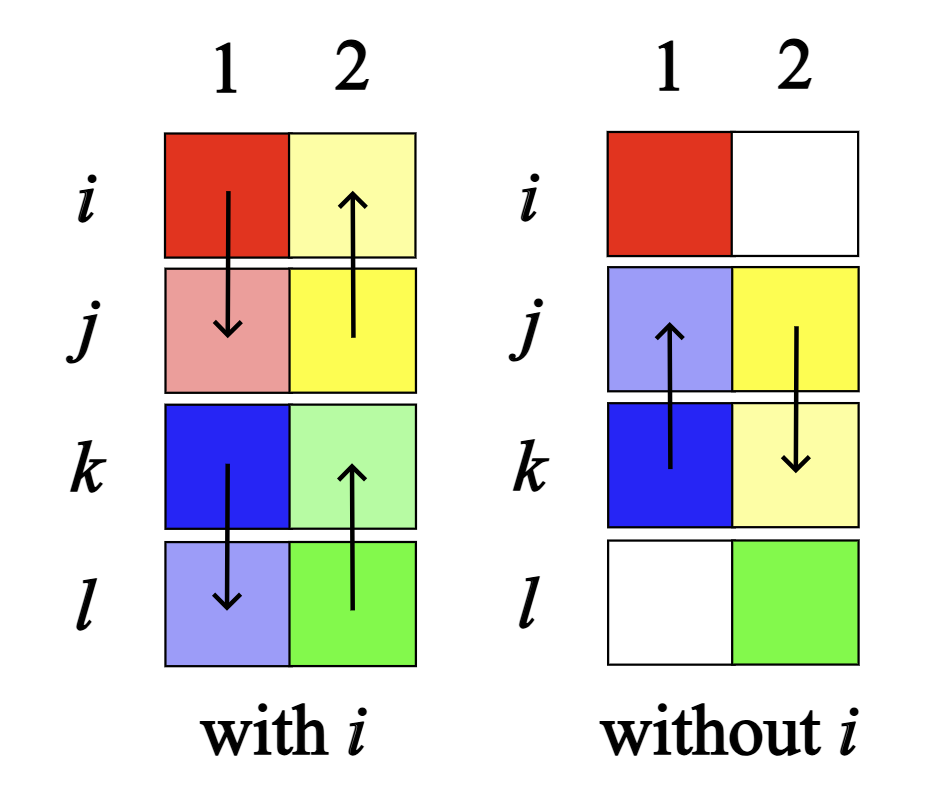}
    \vspace{-0.15in}
    \caption{An instance illustrating the challenges in IR.}
    \vspace{-0.3in}
    \label{fig:irproofillus}
\end{wrapfigure}
}


\newcommand{\insertAlgoMain}{
    \begin{algorithm}
    \caption{\algname}
    \label{alg:main}
    \begin{algorithmic}[1]
    \State \textbf{Input:} 
            participating agents $\partagents$,
            initial goods $\{\initdatai\}_{i\in\partagents}$,
            competition levels $\beta = \{\betaij\}_{i,j\in\partagents}$.
    \State $\rejprops \leftarrow \emptyset$, \Comment{Previously rejected proposals}
    \For{$t \leftarrow 1, 2, ...$}
        \State $\Pt \leftarrow \emptyset$ \Comment{$\Pt$ are
    the exchanges that will be proposed
                at the end of round $t$.}
        \State $\interdatat \leftarrow \datatmo$.
            \Comment{$\interdatat$ tracks the allocation if all proposals
                        in $\Pt$ are accepted.}
        \For{each pair of agents $(i,j)\in\partagents\times\partagents$ in
            increasing order of $\beta_{ij}$} \label{increasing-beta}
            \Statex \Comment{Pick the lexicographically smallest $(i, j)$ if $\betaij$ values are equal.} 
            \While{\textbf{true}} \label{exchange-loop}
                \State $\dmdjfromi \leftarrow \{r \in [m] ;\, \interdatat_{j,r} = 0,
                                                \initdata_{i,r} = 1$\}
                \Comment{Identify goods $i$ could give to $j$.}
               \label{lin:dmdjfromi}
                \State $\dmdifromj \leftarrow \{s \in [m] ;\, \interdatat_{i,s} = 0,
                                                \initdata_{j, s} = 1$\}
                \Comment{Identify goods $j$ could give to $i$.}
               \label{lin:dmdifromj}
                \If{$\dmdjfromi \neq \emptyset \text{ \textbf{and} } \dmdifromj \neq \emptyset$} \label{demand-nonempty}
    \label{lin:checkdmds}
                    \If{there exists $r \in \dmdjfromi$, $s \in \dmdifromj$ s.t. $((i, r),
(j, s)) \notin \rejprops$ } \label{find-exchange}
                        \Statex \Comment{Pick the lexicographically smallest $(r, s)$ if multiple are available.}
                        \State $\Pt \leftarrow \Pt \cup \{ ((i, r), (j, s)) \}$,
                        \Comment{Lazily add an exchange to $\Pt$.}
                    \label{lin:Ptadd}
                        \State $\interdatat_{i,s} \leftarrow 1, \interdatat_{j,r} \leftarrow 1$
                        \Comment{and update $\interdatat$.}
                    \Else
                        \Comment{If all such pairs have been rejected previously.}
                        \State \textbf{break} \label{no-exchange1}
                            \Comment{End while loop and move to the next agent pair.}
                    \EndIf
                \Else
                        \Comment{If either $\dmdjfromi$ or $\dmdifromj$ are empty.
                    Else from line \ref{lin:checkdmds}.}
                    \If{$\dmdjfromi = \emptyset$} $r \leftarrow \retrospectmath(j, i, \{i, j\})$ \label{IC1}
                    \EndIf
                    \If{$\dmdifromj = \emptyset$} $s \leftarrow \retrospectmath(i, j, \{i, j\})$ \label{IC2}
                    \EndIf
                    \If{$r = \retrofailmath$ \text{ \textbf{or} } $s=\retrofailmath$}
                    \Comment{End while loop and move to next agent pair.}
                        \State Undo changes to $\interdatat$ and $\Pt$ in last two calls
                                to \retrospect. \label{no-exchange2}
                            \textbf{break} 
                    \EndIf
                \EndIf
            \EndWhile
        \EndFor
        \State Propose $\Pt$ to agents and let them accept/reject proposals.
        \State $\datat \leftarrow \datatmo$
            \Comment{Compute allocation at end of round $t$. Initialized to $\datatmo$.}
        \label{lin:agentphasestart}
        \State $\datat_{j, r} \leftarrow 1$,  $\datat_{i, s} \leftarrow 1$,
        \; for every accepted proposal $((i, r), (j, s))$ in $\Pt$
        \State Add every rejected proposal in $\Pt$ to $\rejprops$.
        \label{lin:agentphaseend}
        \If{all proposals in $\Pt$ were accepted}
            \State $\Ptpo \leftarrow \emptyset$,
                       \textbf{break}.
            \Comment{Termination.}
        \label{lin:termination}
        \EndIf
    \EndFor

\Statex

\Function{\retrospect}{receiving agent $i$, giving agent $j$, examined agents $S$}
\Statex
    \Comment{Check if any of $i$'s exchanges in $\Pt$ can be adjusted to accommodate
            another good from $j$.}
    \For{each good $s \in \{1,\dots,m\}$ s.t. $\datatmo_{i,s} = 0$, $\initdata_{j, s} = 1$}
    \label{lin:retroforstart}
\Statex \Comment{Goods which $i$ has not received by round $t-1$ but could receive from $j$.}
        \label{all-possible}
        \If{$\interdatat_{i,s} = 0$}
            \Comment{If $i$ is not yet assigned this good in the current round.}
            \State \Return $s$ \label{base-case}
            \Comment{Base case. End recursion if we have found such a good $s$.}
        \Else
        \Comment{This means $i$ is scheduled to receive good $s$ in $\Pt$.}
            \State Let $((i, r), (j', s))\in\Pt$ be the exchange in $\Pt$ via
                which $i$ receives $s$.\label{query}
            \If{$j' \notin S$} \label{condition}
                \Comment{If $j'\in S$, we go back to line~\ref{lin:retroforstart}
            to check the next good.}
                \State $s' \leftarrow \retrospectmath(i, j', S \cup \{j'\})$ \label{recursive-call}
                        
                \If{$s' \neq \retrofailmath$ and $((i, r), (j', s')) \notin \rejprops$}
                    \State $\interdatat_{i, s'} \leftarrow 1,
                    \interdatat_{i, s} \leftarrow 0$, $\Pt \leftarrow \Pt \setminus \{((i, r), (j', s))\} \cup \{((i, r), (j', s'))\}$ \label{alter-x}
                    \State \Return $s$ \label{recursive-case}
                    \Comment{found a good from $j'$, change $\interdatat$ and $\Pt$.}
                \Else
                    \State Undo changes to $\interdatat$ and $\Pt$ in the above \retrospect \ call, if any.
                \EndIf
            \EndIf
        \EndIf
    \EndFor
    \State \Return \retrofail \label{all-examined}
\EndFunction
    \end{algorithmic}
    \end{algorithm}
}

\newcommand{\insertSimulation}{
    \begin{table}[t]
    \centering
    \small
    \resizebox{\linewidth}{!}{\footnotesize
    \begin{tabular}{|c|cc|cc|cc|}
    \hline
    \textbf{Dimension} & \multicolumn{2}{c|}{$p=0.1$} & \multicolumn{2}{c|}{$p=0.5$} & \multicolumn{2}{c|}{$p=0.9$} \\
    \textbf{$n \times m$} & \textbf{Trials} & \textbf{Avg Runtime} & \textbf{Trials} & \textbf{Avg Runtime} & \textbf{Trials} & \textbf{Avg Runtime} \\
    \hline
    $10 \times 10$   & $10^6$ & 0.017 & $10^6$ & 0.021 & $10^6$ & 0.011 \\
    $20 \times 20$   & $10^5$ & 0.12  & $10^5$ & 0.13  & $10^6$ & 0.048 \\
    $30 \times 30$   & $10^5$ & 0.42  & $10^5$ & 0.40  & $10^5$ & 0.20  \\
    $40 \times 40$   & $10^4$ & 1.96  & $10^4$ & 1.13  & $10^5$ & 0.23  \\
    $50 \times 50$   & $10^3$ & 21.2  & $10^4$ & 2.89  & $10^5$ & 0.37  \\
    $60 \times 60$   & $10^2$ & 95.3  & $10^4$ & 7.90  & $10^5$ & 0.57  \\
    $70 \times 70$   & N/A  & N/A   & $10^3$ & 19.3  & $10^5$ & 0.79  \\
    $80 \times 80$   & N/A  & N/A   & $10^3$ & 46.7  & $10^4$ & 1.28  \\
    $90 \times 90$   & N/A  & N/A   & $10^2$ & 223   & $10^4$ & 1.67  \\
    $100\times100$   & N/A  & N/A   & $10^2$ & 1242  & $10^4$ & 2.12  \\
    \hline
    \end{tabular}
    }
    \vspace{5pt}
    \caption{Simulation of \algname{} over various settings of parameters $n, m, p$. We run each setting for a reasonable number of trials based on the runtime. The Avg Runtime shown is an average over all trials, in milliseconds, for the given $(n,m,p)$ values. An N/A means that the setting takes too long to run. Notably, no counter example against Conjecture \ref{conj:clearperfagent} was found.}
    \label{tab:simulation}
    \end{table}
}

\clearpage
\pagenumbering{arabic}


\section{Introduction}
\label{sec:intro}

In many domains, competing agents can mutually benefit from collaborating on \emph{freely replicable goods}\footnote{A \emph{freely replicable} (a.k.a. nonrivalrous) good is one that can be used by multiple agents without reducing its availability to others; examples include digital goods such as datasets, and intellectual property such as patents or software.}.
For instance, pharmaceutical companies may share drug discovery data with competitors, when the benefits of accessing others' data outweigh the risks of revealing their own proprietary information. Such collaboration can strengthen the competitive standing of the firms relative to the broader market, while also accelerating drug development for the benefit of society.
Similar dynamics arise when technology firms cross-license patents or software, or when academic laboratories share data with each other.
We describe some real‑world instances of such collaboration, both to motivate this work, and to justify the design choices that follow.
\begin{enumerate}[leftmargin=0.22in]
    \item \emph{Data sharing between academic laboratories~\citep{clariti,aflow,hgp,genova2016rda}:}  
    Academic laboratories at universities compete to publish first and are therefore protective of their datasets, especially in fields such as medicine and materials science, where a single dataset can yield multiple publications. Nevertheless, labs often enter into data-sharing agreements to accelerate research. Trading datasets for money is often illegal and viewed as unethical in academic contexts. 
    Different labs have varying degrees of competition with each other, and these are often publicly known through overlaps in research focus.
    Through a lab's published papers, it is often publicly known which labs hold which datasets, even when the datasets themselves are undisclosed.
    Moreover, once a dataset is shared, the recipient is typically prohibited from re-sharing it, as onward sharing could undermine the original lab’s ownership; such restrictions are enforced through legal agreements and by the risk of exclusion from future collaborations.

    \item \emph{Cross‑licensing consortia in industry~\citep{iqconsortium,ccc,wifialliance}:}  
    Commercial firms, competing for the same customer base, may form consortia to cross-license intellectual property (IP)---such as patents or software---to foster innovation and reduce litigation risk.
    Firms are often more willing to enter cross-licensing agreements than purchase usage rights, particularly when the monetary value of an IP is uncertain.
    Firms' competition levels are publicly known 
    based on product and market overlap
    (\eg Pfizer and Merck compete more directly than Novartis; HBO and Netflix more directly than Spotify). 
    IP received through a consortium is restricted from onward licensing via legal agreements, as it could dilute the competitive advantage intended by the original cross-license agreement.
    Finally, patent ownership information is in public record.

    \item \emph{Data and model‑sharing consortia in industry~\citep{freightdat,cpad,crdsa,foshan}:}  
    Similarly, firms may collaborate by sharing datasets or trained models as it allows smaller organizations to compete against larger, data-rich players.
    While the data itself is private, information about the type of data each owns is known through press releases, industry reports, consortium registries, or regulatory filings.
\end{enumerate}

Motivated by these examples, we study a model that captures the competitive dynamics of
agents sharing freely replicable goods, such as digital goods (\eg datasets) and IP (\eg patents).
Each agent seeks to acquire more goods for
her own benefit, but experiences a \emph{negative externality} when others possess
her goods; for instance, a pharmaceutical company risks losing market share if its competitor develops a drug using its data or IP.
Naively granting all agents access to all goods maximizes sharing, but disincentivizes participation from agents who initially hold many goods, as they risk giving others access to their many goods, while receiving only a few in return.

\parahead{Pairwise exchanges}
An alternative approach, which we explore in this work, is to facilitate \emph{pairwise exchanges},
where, as illustrated in Figure~\ref{fig:pairwiseillus}, two agents grant each other access to a good they have in exchange for a good they lack (\eg exchanging datasets).  
We study the design of \emph{exchange protocols}---without money---to structure such exchanges.  
The protocol operates over multiple rounds, proposing a set of pairwise exchanges at each round, which agents may choose to accept or reject.  
The protocol should
ensure that agents benefit by participating,
incentivize agents to share as much as possible by accepting all proposed exchanges,
and produce a stable outcome such that no further mutually beneficial exchanges remain upon termination.

\insertPairwiseIllus
\emph{Why pairwise exchanges?}
One could consider more general frameworks where agents submit their goods to a central planner, who then reallocates access. 
However, we focus on pairwise exchanges for three reasons: 
\emph{(i)} To our knowledge, no prior work studies exchanging freely replicable goods under negative externalities; pairwise exchanges offer a nontrivial yet theoretically tractable starting point, yielding interesting technical insights and socially desirable outcomes. 
\emph{(ii)} Restricting to pairwise exchanges leads to computationally feasible algorithms, as considering all possible multilateral exchanges among agents quickly becomes expensive without additional assumptions. 
\emph{(iii)} Pairwise exchanges also give agents greater control over who gains access and under what conditions, making participation more appealing in practical settings.

\vspace{0.1in}

Next, in~\S\ref{sec:summarymodel} and~\S\ref{sec:summaryresults} we summarize our contributions.  
To simplify exposition, going forward, we describe our work in the context of a set of competitors sharing copies of digital goods.

\subsection{Summary of our Model}
\label{sec:summarymodel}

As the first to study exchanging freely replicable goods with negative externalities, we contribute a novel framework. In~\S\ref{sec:summaryresults}, we present  a set of impossibility results that justify this model.

\parahead{Environment and agent utilities}
There is a finite set of agents $\allagents$ and $m$ freely replicable and \emph{indivisible} goods (\eg each good is a data point).
Initially, each agent possesses a subset of these goods,
where multiple agents may hold the same good.
Agent $i$’s \emph{initial allocation} is encoded by $\initdatai \in \{0,1\}^m$, where $\initdatair = 1$ if $i$ has good $r$. 
Denote  $\initdata = \{\initdatai\}_{i\in\allagents} \in \{0,1\}^{|\Acal| \times m}$.
Competition between agents is described by symmetric parameters $\betaij = \betaji \in (0,1)$: the higher $\betaij$, the more strongly $i$ and $j$ compete.
An instance of the problem is thus given by $(\initdata, \beta)$.

After performing exchanges, let
$\datai \in \{0,1\}^m$ be the goods held by agent $i$. Let $\data = \{\datai\}_{i\in\allagents} \in \{0,1\}^{|\Acal| \times m}$ be the allocation to all agents.
An agent’s utility $\utilalloci$ under allocation $\data$ is,
\begin{align}
\utilalloci(\data) = \onev^\top \datai - \sum_{j \neq i} \betaij \onev^\top \dataj.
\label{eqn:utilsummary}
\end{align}
Here, the first term rewards $i$ for her own holdings, while the second penalizes her for the goods held by competitors.
In this work, we assume that all agents value all goods equally;
while admittedly a simplification, the problem is quite rich and nontrivial even under this assumption, and captures many of the core challenges of exchanging freely replicable goods under competition.

\parahead{Exchange protocol}
An altruistic \emph{planner} (\eg a consortium administrator) facilitates the exchange of goods through a multi-round protocol (a.k.a. mechanism or policy).
A subset of agents $\partagents \subset \allagents$ choose to participate. 
In each round, the planner proposes a set of \emph{pairwise exchanges},
where
a proposal $p = ((i,r),(j,s))$ recommends that $i$ gives a copy of good $r$ to $j$, and $j$ gives a copy of $s$ to $i$. 
Each agent can either accept or reject proposals involving them, and an exchange occurs only if both agents accept.
At the end of each round, the agents exchange goods,
retaining the goods they give away, while acquiring copies of goods received from others.
In the next round, the protocol  proposes a new set of exchanges.
This process continues for multiple rounds, with the planner updating future proposals based on past acceptances and rejections, until the protocol chooses to terminate.
For technical convenience, we do not allow
a protocol to propose exchanges which have been previously rejected by agents\footnote{\label{fn:cannotproposerejected}We invoke this condition only in one impossibility result, to rule out pathological protocols and simplify the proof.}.

We highlight two key features of our model: 
\emph{(i)} Onward sharing of goods received from one agent in exchange for additional goods from another is disallowed. Such constraints reflect real-world practices in data-sharing and cross-licensing consortia, as noted in the examples above. Moreover, allowing onward sharing would fundamentally change the problem characteristics, and could be an interesting direction for future work. We briefly discuss the challenges in Appendix~\ref{sec:discussion}.
\emph{(ii)} We assume that the initial allocations $\{\initdatai\}_{i\in\partagents}$ and competition levels $\{\betaij\}_{i,j\in\partagents}$ of participating agents are \emph{truthfully revealed} to the planner. As the examples above indicate, this assumption is often justified in practice, as they may already be publicly known.
Moreover, as we will show via Theorem~\ref{thm:truthfulreporting},  eliciting the initial allocation $\initdata$ truthfully from the agents is impossible in this setting.

\parahead{Strategic behavior}
An agent $i$’s strategy $\strati$ specifies, based on the history of previous exchanges, whether to accept or reject each proposal involving her.
To maximize sharing and deter strategic behavior, we seek to design protocols that incentivize agents to follow the \emph{accepting policy}, where an agent accepts all proposed exchanges.

\parahead{Desiderata} 
We identify three key desiderata for any exchange protocol, stated informally below.
\begin{enumerate}[leftmargin=0.22in]
    \item \emph{Nash Incentive compatibility (NIC).}  
    The best strategy for each agent should be to accept all exchanges proposed by the planner, when other agents are also doing so.
    Precisely, the accepting strategy profile should constitute a Nash equilibrium of the protocol.

    \item \emph{Individual rationality (IR).}  
    Agents should not be worse off by participating.
    Precisely, an agent who joins the protocol and follows the accepting strategy should obtain at least as much utility as she would, if she stayed out while the others participate with the accepting strategy. 

    \item \emph{Stability.}  
    We say that an \emph{allocation} $\data$ is stable if no two agents can further increase their utilities by exchanging goods.
An \emph{exchange protocol} is stable if, whenever a mutually beneficial exchange remains upon termination, it has already been rejected by the agents.
Hence, a protocol is stable only if it terminates with a stable allocation when all agents follow the accepting strategy.

\end{enumerate}

\subparahead{Key challenges in formulating desiderata} 
Similar desiderata appear in classical models for exchange of goods, such as stable matching~\citep{gale1962college}, one-sided matching~\citep{shapley1974cores}, and exchange economies~\citep{varian1973equity}, under irreplicable  goods and no externalities.
However, studying freely replicable goods and/or negative externalities
requires a careful reformulation of these desiderata.

For example, dominant-strategy incentive compatibility (DSIC), where it is the best strategy for an agent to accept all proposals regardless of others' strategies, is often achievable in classical exchange models~\citep{gale1962college,shapley1974cores}.
However, in Theorem~\ref{thm:nodsic}, we show that this is not always attainable in our setting, motivating our focus on the weaker NIC criterion.

\insertFigIRIllus
Similarly, standard IR definitions, which require that each agent be better off participating than in the initial state, is problematic under competition:
achieving this will prevent many exchanges and will not result in a stable exchange protocol.
To illustrate, consider Figure~\ref{fig:irillus}.
Agent $k$'s initial utility is $2-\betaik-\betajk$ (see~\eqref{eqn:utilsummary}). If $i$ and $j$ exchange goods 1 and 2, $k$ is worse off, with utility  $2-2\betaik-2\betajk$, as both her competitors have gained goods.
As $k$ cannot gain any goods herself, the only way to ensure that she is not worse off than the initial state would be to not propose the exchange between $i$ and $j$.
However, this is undesirable, as $k$'s participation should not hinder trades between the others.
Moreover, enforcing this requirement will not result in a stable protocol, as $i$ and $j$ can benefit by privately exchanging outside the protocol.
Hence, we define IR in our setting for an agent~$k$ by comparing her utility in two scenarios: 
\emph{(a)} when she does not participate, but a subset of other agents \( \partagents \) do, and  
\emph{(b)} when she joins \( \partagents \).  
IR requires that agent \( k \)'s utility in \emph{(b)} is no worse than in \emph{(a)} for all \( \partagents \subset \allagents \setminus \{k\} \).

Finally, defining stability using a core-like notion from coalitional games~\citep{shapley1971core,shapley1974cores} is also unsuitable under large negative externalities. A break-away coalition may obtain higher utility only because it reduces overall sharing, and its members need not remain committed to the coalition as they can benefit from trades outside the coalition.
We discuss this further in Appendix~\ref{sec:discussion}.

\parahead{Efficiency criteria}
In classical \emph{non-competitive} exchange models~\citep{gale1962college,varian1973equity,shapley1974cores}, an exchange benefits the two agents without hurting others.
Therefore, repeated exchanges improves everyone's utility.
In these settings, stability (no further mutually beneficial exchanges) coincides with Pareto efficiency (no agent can be made better without making another worse off).
In contrast, in our setting, while two agents $i$ and $j$ each gain utility $1 - \beta_{ij}$ from exchanging a good, every other agent $k$ suffers a utility loss of $-\beta_{ik} - \beta_{jk}$ (see~\eqref{eqn:utilsummary}).
As a result, after multiple exchanges, all agents may end up worse off than they were initially, particularly when the $\beta_{ij}$ values are large\footnote{
Importantly, this phenomenon is natural in competitive settings and is not an artifact of our specific model.
In any exchange model (even with irreplicable goods), if two agents benefit from an exchange, their competitors are necessarily worse off.
}.
Yet, the initial allocation may itself be unstable, as any two agents could carry out mutually beneficial exchanges.

Therefore, we do not impose an efficiency desideratum. Conventional notions such as welfare maximization and Pareto efficiency (PE) are ill-suited, and diverge from stability under high competition.
For instance, the social welfare under an allocation $\data$ is 
$\sum_{i} \utilalloci(\data)
= \sum_{i} (1 - \sum_{j \neq i} \beta_{ij}) \onev^\top \datai$, 
which can \emph{decrease} as agents acquire more goods when the $\beta_{ij}$’s are large.  
Moreover, via Theorem~\ref{thm:pehardness}, we will show that in highly competitive environments, stability and PE are incompatible: an unstable initial allocation can be PE, while a stable final allocation may be Pareto-dominated.
However, as we will see, under low competition, we conjecture that our protocol may achieve PE.

\parahead{Departure from classical models}
While our model draws inspiration from classical work in mechanism design and exchange models~\citep{gale1962college,shapley1974cores,fisher1891mathematical,varian1973equity,walras1900elements}, it also departs sharply: we assume truthful revelation of the initial allocation, only require NIC rather than DSIC, 
do not compare to an agent's initial utility in IR,
and observe that stability need not imply efficiency. 
As our impossibility results and discussions indicate, these
adaptations are both necessary and natural.  Reformulating classical desiderata for competitive exchanges was one of the major conceptual and technical challenges of this work.
It is worth pointing out that the limited prior work on mechanism design in competitive environments either overlooks these challenges~\citep{chen2023equilibrium,agarwal2020towards,tsoy2023strategic,branzei2013externalities}, or considers regimes with small competition, where such difficulties do not arise~\citep{dorner2023incentivizing}.

\subsection{Summary of our results, methods, and proof techniques}
\label{sec:summaryresults}

Next, we summarize our results, including three impossibility results, our method and key result, and additional results on achieving Pareto-efficiency in this setting.

\parahead{Impossibility results}
We present three impossibility results that justify our modeling choices.

\subparahead{1) DSIC vs NIC protocols}
Our first result shows that no nontrivial DSIC protocols are possible in this problem.
This motivates our focus on NIC.

\begin{restatable}[Informal]{theorem}{nodsic}
    \label{thm:nodsic}
    There exists a problem instance $(\initdata, \beta)$ on which no protocol can
    simultaneously satisfy DSIC and stability.
\end{restatable}

A key challenge in this proof is the potentially unbounded space of protocols. We show that, for a carefully constructed instance $(\initdata, \beta)$, any stable protocol can be represented as a tree whose branches correspond to agent responses. This representation allows us to reduce the analysis to a finite, though still large, set of trees. By partitioning these trees into cases, we systematically construct counterexample strategy profiles, thereby establishing the impossibility of DSIC.

\subparahead{2) Truthful revelation of $\initdata$}
Our model assumes that the initial allocations $\{\initdatai\}_{i\in\partagents}$ of participating agents are always revealed truthfully.  
Beyond the practical motivations outlined in our examples, this assumption is further supported by the following theorem, which demonstrates that eliciting the true initial allocation \(\initdata\) is impossible in our setting.

\begin{restatable}[Informal]{theorem}{truthfulreporting}
\label{thm:truthfulreporting}
There exists a problem instance $(\initdata, \beta)$ such that for any stable exchange protocol, all agents truthfully reporting the initial allocation and accepting all exchange proposals does not constitute a Nash equilibrium.
\end{restatable}

Our proof shows that an agent $i$ who holds many goods can benefit by misreporting and claiming to hold only a few. This may cause the mechanism to schedule exchanges in which $i$ receives a good she already possesses but has not revealed. This can reduce potential exchanges for a competitor who genuinely needs the good, thereby increasing $i$’s utility.

\subparahead{3) Multi-round vs single-round protocols}
Finally, we focus on multi-round protocols, since a protocol which only proposes exchanges once cannot be simultaneously stable and NIC.
That said, our method terminates in a single round if all agents follow the accepting policy.

\begin{restatable}[Informal]{observation}{singleround}
\label{obs:singleround}
    There exists $\initdata$ such that, for all competition levels $\beta$, no single-round protocol can simultaneously satisfy stability and NIC.
\end{restatable}

\parahead{Our method and key result}
We start with an intuitive approach to designing a stable protocol.
On each round, we evaluate all possible pairwise exchanges and propose a set of exchanges
that satisfy two conditions:
\emph{(a)} no agent receives duplicate items, and
\emph{(b)} if all proposals are accepted, the resulting allocation is stable, \ie no further mutually beneficial
pairwise exchanges are possible.

While intuitive, this approach is not \emph{NIC}.
Agents may find it beneficial to decline proposals for
two reasons.
\emph{(i) Level of competition:}
If agent $k$ is less competitive with agent $i$ than agent $j$ is,
then $i$ might reject an exchange with $j$,
anticipating that the protocol might propose a similar exchange with $k$ in a
future round.
While $i$'s allocation does not change in either scenario, she benefits since a less competitive agent gains access to her good.
\emph{(ii) Rarity of goods:}
Suppose $j, k$ both have a common good and $k$ has a rare good. Then, $i$ might decline the common good from $k$, hoping to acquire both the rare good from $k$ and the common good from $j$ in a future round so that she receives two goods in total instead of one.

\subparahead{An improved method}
Our protocol \algname{} (\textbf{C}ompetitive-order \textbf{L}azy 
\textbf{E}xch\textbf{A}nges with 
\textbf{R}etrospection)
builds on this initial attempt but resolves the two shortcomings.
Addressing the first issue is relatively straightforward:
we can simply iterate over all pairs of agents $(i,j)$ in
increasing order of competition factor $\betaij$, scheduling as many exchanges as possible between them.
This ensures that each agent has the opportunity to exchange with less competitive agents first.

Addressing the second issue is much more challenging.
For this, on any given round,
we first lazily add exchanges to the set of proposals for that round. Then, for each agent pair $(i, j)$, we retrospectively
check if either agent could have obtained a different (rarer) good from previous exchanges, had these been scheduled differently, so that more exchanges are possible between $(i, j)$.
The key tool is a recursive subroutine, called \retrospect, which backtracks through previously scheduled exchanges in that round to determine whether a rarer good can be made accessible by modifying these exchanges, but without altering the \emph{number} of these exchanges.

We now state the main theoretical result of this paper. 

\begin{restatable}[Informal]{theorem}{main}
\label{thm:main}
The above protocol 
 satisfies NIC, IR, and stability.
When all agents follow the accepting policy, it terminates in one round in
$\bigO(|\partagents|^2 m^2 \max\{|\partagents|, m\})$ time.
\end{restatable}

\emph{Proof techniques.}
The proof of stability is by design.
For NIC,
we show that \retrospect{} always optimizes the number of exchanges an agent can get as the protocol proceeds from low to high competition, and that prioritizing exchanges with lower competition optimizes the total utility gain from these exchanges.
Thus, an agent cannot hope for higher utility from any set of alternative exchanges involving her in future rounds.
For IR, a key challenge is that an agent $i$'s participation will create knock-on effects among other exchanges \algname{} scheduled when $i$ does not participate. We construct a graph that captures the difference in exchanges when $i$ is present or absent.
Through careful analysis, we show that $i$'s utility is only affected by chains of such knock-on effects represented by certain connected components in the graph.
Furthermore, each chain contributes nonnegatively to her utility, so summing these results in nonnegative utility change for $i$ overall.

Intuitively, what \algname{} schedules after the first round is to ensure stability in case some agents may reject.
However, NIC should convince agents to follow the accepting policy.

\parahead{On achieving Pareto-efficiency}
Finally, we study when Pareto-efficiency  (\ie to increase the utility of one agent, we have to decrease the utility of another) is achievable in this problem.
For simplicity, let us assume that all agents are participating.
We first identify two (non-exhaustive) regimes characterizing the level of competition among agents.
We refer to the agents as a \emph{low externality coalition} (LEC) if, upon giving each agent one additional good they do not currently possess, all agents are made better off,
\ie $1 - \sum_{j \neq i} \beta_{ij} > 0$ for all $i \in \allagents$ (see~\eqref{eqn:utilsummary}).
Conversely, we call them a \emph{high externality coalition} (HEC) if, upon giving each agent one additional good they do not currently possess, no agent is made better off, \ie $1 - \sum_{j \neq i} \beta_{ij} \leq 0$ for all $i \in \allagents$.
Our first result below shows that, in an HEC, one cannot improve all agents' utilities via multiple exchanges from the initial allocation; hence in any stable protocol,
either at least one agent is strictly worse off, or none of the agents are strictly better off when compared to the initial allocation.

\begin{restatable}[Informal]{theorem}{pehardness}
\label{thm:pehardness}
Suppose $\allagents$ is an HEC and none of the agents possess all goods initially. Then, any other allocation $x$ does not Pareto-dominate the initial allocation $\initdata$.
\end{restatable}

Next, even when externalities are small, stability does not necessarily imply PE.  
However, we conjecture that our protocol achieves  PE in LECs.  
We partially support this claim by proving this under the following conjecture about the protocol’s behavior:
when all agents follow the accepting policy, at least one agent receives all goods.  
Extensive simulations in Appendix~\ref{app:simulations}  suggest that this conjecture is likely true.  
Under this conjecture, we obtain the following result.

\begin{restatable}[Informal]{theorem}{clearpelec}
    \label{thm:clearpelec}
    Suppose $\allagents$ is an LEC and all agents follow the accepting policy.
    Under the above conjecture, the final allocation produced by the protocol in Theorem~\ref{thm:main} is Pareto-efficient.
\end{restatable}

\subsection{Related work}

To the best of our knowledge, no prior work directly addresses the problem of sharing
freely replicable goods in the presence of negative externalities without money.

\parahead{Sharing irreplicable goods}
Classical economic theory has extensively studied the exchange of irreplicable (regular) goods, particularly in the contexts of barter economies and market
design~\citep{varian1973equity,walras1900elements,fisher1891mathematical}. Similar to our
model, here agents initially possess goods, and engage in exchanges to
achieve mutually beneficial outcomes.
Mechanism design for the exchange of irreplicable goods has been explored in various
settings,
including stable matching~\citep{gale1962college,roth1990two},
one-sided matching (a.k.a. house
allocation, trading agents)~\citep{shapley1974cores,roth1982incentive}, kidney
exchange~\citep{roth2004kidney}, and resource
allocation~\citep{parkes2015beyond,freeman2018dynamic,ghodsi2011dominant,gutman2012fair}.
These works aim to incentivize agents to truthfully reveal their preferences.
In both lines of work above,
an agent's possession of a good does not impose externalities on
others. 
As discussed above, the presence of negative externalities introduces challenges absent in classical models.

\parahead{Sharing irreplicable goods with externalities}
Some study fairness in exchanging irreplicable goods with
externalities~\citep{velez2016fairness,branzei2013externalities,seddighin2019externalities,treibich2019welfare},
but they do not extend to freely replicable goods or address incentives for participation and
pro-social behavior.

\parahead{Data sharing}
Recent work has explored mechanisms for sharing data among strategic, noncompetitive agents, focusing on preventing free-riding~\citep{blum2021one,karimireddy2022mechanisms,chen2023mechanism,cai2015optimum,clintoncollaborative}.
However, these models do not consider externalities. 
Finally, some study data sharing for mean estimation in competitive setting; they assume  that all agents have similar competing levels with each other, and that the level of competition is low, hence avoiding many of the challenges we encounter in this work~\citep{dorner2023incentivizing}.

\parahead{Auctions and data marketplaces with externalities}
A separate line of work in mechanism design studies auctions that incorporate externalities among participants.
Some examine positive externalities in a data auction setting and propose competitive auctions where a buyer's valuation increases with the set of other winning buyers \citep{GravinLu2013}. 
Others look at negative externalities where a bidder incurs a loss if other bidders win, or bidders form a conflict graph \citep{Brocas2013, BelloniDengPekec2017, CheungHenzingerHoeferStarnberger2015}.
Others study data marketplaces where buyers impose externalities on other buyers~\citep{agarwal2020towards,chen2023equilibrium,branzei2013externalities}.
These works differ fundamentally from our work, in both formulation and reliance on monetary incentives.

\section{Problem Setting}
\label{sec:setup}

We will now formally describe the problem set up.

\parahead{Problem instance}
There is a finite set of agents $\allagents$ and $m$ freely replicable goods. Agent $i$ is initially in possession of a subset of these goods. Let
$\initdatai \in \{0, 1\}^m$ represent $i$'s
\emph{initial allocation}, 
with $\initdatair = 1$ meaning that she initially has good $r$.
Multiple agents may initially possess the same good.
Denote $\initdata = \{\initdatai\}_{i\in \allagents}$.
We model competition between agents via symmetric parameters
$\beta= \{\betaij\}_{i,j \in \allagents}$,
where $\betaij = \betaji \in (0, 1)$ is the \emph{competition level} between $i$ and $j$; higher $\betaij$ means $i$ and $j$ are more competitive.
An instance of this problem is given by $(\initdata, \beta)$.

\parahead{Exchange protocol}
A central planner (an altruistic third party) designs and publishes an \emph{exchange protocol} (a.k.a. mechanism) to facilitate the exchange of goods.
A subset of agents $\agsubset\subset\allagents$ choose to participate in the protocol.
Then, the initial allocations $\{\initdatai\}_{i\in\partagents}$
and competition levels $\{\betaij\}_{i,j\in\partagents}$ of participating agents
are \emph{revealed truthfully} to the planner.
As we will show via Theorem~\ref{thm:truthfulreporting},  eliciting the initial allocation $\initdata$ truthfully from the agents is impossible in this problem.

The protocol
operates over a sequence of rounds $t=1,2,\dots$.
On each round, the planner proposes a set of pairwise exchanges between participating agents,
where two agents give copies of
a good that they have to each other, in exchange for one that they do not.
Agents can choose which proposals to accept and reject, with the exchange occurring only if both agents accept.
At the end of each round, the agents exchange goods for accepted exchanges.
In the next round, the planner will  propose a new set of exchanges based on the history of exchanges so far.
She will continue in this fashion, up until she chooses to terminate.
We will now define this procedure formally.

\subparahead{Pairwise exchanges}
A pairwise \emph{exchange proposal} is 
a tuple $p = ((i, r), (j, s))$ where $i, j \in \agsubset$ are two agents and $r, s\in[m]$ are goods.
This proposal is a recommendation by the planner that agent $i$ gives a \emph{copy}
of  good $r$ to agent $j$ and in return, agent $j$ gives a \emph{copy} of good $s$ to agent $i$.
When the planner proposes $p$, each agent can decide whether to accept or reject the proposal.
An exchange occurs only if both agents accept the proposal.
As agents share copies of the goods,
after an exchange, each agent has the good they gave as well as the good they received.

The proposals by the planner and the accept/reject decisions by the agents result in a sequence
of allocations
$\initdata, \datatt{1}, \datatt{2}, \dots$, where
$\initdata$ is the initial set of goods, and
$\datat = \{\datati\}_{i\in \allagents}$,
where $\datati \in \{0,1\}^{m}$ specifies which goods agent $i$
has at the end of round $t$.
In particular, if $\datatmoir=1$ we set $\datatir=1$;
moreover, for every proposal $((i,r),(j,s))$ that was accepted by both agents on
round $t$, we also set
$\datat_{i,s} = \datat_{j,r}=1$.
For every other $i,r\in \agsubset\times[m]$, we have $\datatir = 0$.
If an exchange $((i,r),(j,s))$ is executed on round $t$ when $i$ already possesses $s$, \ie $\datatmo_{i,s}=1$, or if $i$ receives $s$ through multiple exchanges in round $t$, then $i$ retains a (single) copy of the good, \ie $\datat_{i,s}=1$.
Non-participating agents do not gain new goods, \ie for all $i\in \allagents \backslash \agsubset$,
we have $\datati =\initdatai$ for all $t$.

\subparahead{Feasible exchange proposals}
In this work, we only consider exchanges where the agent giving a good originally had the good in $\initdata$, 
\ie she cannot give a good she received from one agent to another to obtain more goods (see discussions in~\S\ref{sec:summarymodel} and Appendix~\ref{sec:discussion}).
Hence, we can write the set of feasible proposals
$\feasprops(\agsubset)$ 
as shown below
in~\eqref{eqn:feasibleexchanges}.
For what follows, we also define
$\feaspropsp(\data,\agsubset)$ to be those exchanges in a given state $\data\in\{\datat\}_{t\geq 0}$,
where neither agent
receiving a good has previously received the good from someone else:
\begingroup
\allowdisplaybreaks
\begin{align*}
    \feasprops(\agsubset) &= \Big\{((i, r), (j, s)); \;\;
      i, j\in\agsubset,\;\; r, s\in[m], \;\;\,
    \initdata_{i, r} = \initdata_{j, s} =  1
    \Big\},
    \\
    \feaspropsp(\data, \agsubset) &= \Big\{((i, r), (j, s)) \in \feasprops(\agsubset); \;\;
    \data_{i, s} = \data_{j, r} =  0
    \Big\}.
    \numberthis
    \label{eqn:feasibleexchanges}
\end{align*}
\endgroup

\subparaheadwnogap{Exchange protocol}
An exchange protocol $\mech=(\mecht)_{t\geq 1}$
is a multi-round policy which outputs a set of feasible exchange proposals on each round.
Let $\Htmo$ denote the history of exchanges proposed by the planner and the associated
decisions by agents up to round $t-1$.
On round $t$, the protocol maps this history to a set of feasible exchanges $\Pt$ which the planner proposes to the agents,
\ie $\Pt\defeq M_t(\Htmo, \agsubset) \subset \feasprops(\agsubset)$.
We do not allow
a protocol to propose exchanges which have been previously rejected by agents (see footnote~\ref{fn:cannotproposerejected}).
The protocol terminates when the planner proposes $\Pt=\emptyset$.
Let $\datafinal$ denote the final allocation at termination.

\parahead{Agent strategy}
An agent $i$'s strategy $\strat_i$ describes whether she chooses to accept or reject proposals.
At the outset, she knows $\partagents$, which agents are participating.
On round $t$, she knows
  the history of previous exchange proposals and decisions $\Htmo$,
and all exchange proposals  $\Pt$ in the current round,
including any not involving her.
Then $\strat_i(p, \agsubset, \Pt, \Htmo)\in\{0,1\}$
specifies, for every exchange proposal $p\in\Pt$ \emph{involving} $i$, whether she accepts
($\strat_i(p, \agsubset, \Pt, \Htmo) = 1$) or
rejects ($\strat_i(p, \agsubset, \Pt, \Htmo)=0$) it.
When $\agsubset, \Pt, \Htmo$ are clear from context, we will write
$\strat_i(p) = \strat_i(p, \agsubset, \Pt, \Htmo)$.
Let $\stratpart = \{\strati\}_{i\in\agsubset}$ denote the strategy profile of all participating agents,
and $\stratpartmi = \{\stratj\}_{j\in\agsubset\backslash\{i\}}$ denote everyone's strategies 
except $i$.

\subparahead{Accepting policy}
The \emph{accepting policy} $\stratacci$ always accepts all proposals, 
\ie $\stratacci(p, \agsubset, \Pt, \Htmo) = 1$
for every proposal $p\in\Pt$ involving $i$, and for any $\agsubset$, $\Pt$ and $\Htmo$.
Denote $\strataccpart = \{\stratacci\}_{i\in\partagents}$.

\parahead{Agent utility}
For an allocation $\data = \{\datai\}_{i \in \allagents}$, agent $i$'s utility
$\utilalloci(\data)$ (defined in~\eqref{eqn:utilsummary}) is the number of goods she holds minus the number of goods held by
other agents scaled by competition factors.
The final allocation $\datafinal = \datafinal(M, \agsubset, \stratpart)$ depends on the
exchange protocol $M$, which agents are participating $\agsubset$, and the agents' strategies $\stratpart$.
Hence, we can write the utility $\utili(M, \agsubset, \stratpart)$ of an agent $i$  in a protocol $M$ under strategy profile $\stratpart$ as
$\utili(M, \agsubset, \stratpart) \defeq \utilalloci(\datafinal(M, \agsubset, \stratpart))$.
Observe that the utility 
accounts for the goods in possession of both participating and non-participating agents.
Moreover, $\utili$ is also defined for both participating 
and non-participating agents. These are necessary to formalize the IR and NIC
requirements in competitive settings.

\subsection{Desiderata for an exchange protocol}
\label{sec:desiderata}

We aim to design a protocol that satisfies  three key desiderata:
\emph{(i) Incentive compatibility:} agents are incentivized to share as much as possible, \ie follow the accepting policy;
\emph{(ii) Individual rationality:} no agent is worse off by participating;
\emph{(iii) Stability:} after the protocol terminates, agents have no incentive to engage in further exchanges among themselves.
While classical models for exchange~\citep{gale1962college,shapley1974cores,varian1973equity} also pursue these three goals, they must be reformulated to address new challenges arising from agent competition and the free replicability of goods.

\parahead{Incentive-compatibility}
In order to incentivize agents to share as much as possible, we would like all agents to follow the accepting policy $\stratacci$.
One common way to formalize this is to require a protocol be
\emph{dominant-strategy incentive compatible (DSIC)}, \ie $\stratacci$ is the best strategy
for each agent regardless of who is participating and the strategies followed by those
participating.
Formally, for all $\agsubset\subset \allagents$, $i\in\agsubset$, $\strati$
and $\stratpartmi$, we require
$\utili(M, \agsubset, (\stratacci, \stratpartmi)) \geq \utili(M, \agsubset,
(\strati,\stratpartmi))$.
Unfortunately, as we will demonstrate via Theorem~\ref{thm:nodsic} below, there is no nontrivial DSIC protocol for this problem.
Hence, we resort to designing a protocol where following the accepting strategy is the best
strategy for an agent when other agents also follow the accepting strategy.

\begin{desideratum}
[(Nash) incentive-compatibility (NIC)]
\label{des:nic}
All agents following the accepting strategy is a Nash equilibrium of $M$,
\ie for any set of participating agents $\agsubset\subset \allagents$,
for all $i\in\agsubset$, 
and all alternative strategies $\strati$ for $i$, we
have
$\utili(M, \agsubset, (\stratacci, \strataccpartmi)) \geq
\utili(M, \agsubset, (\strati, \strataccpartmi))$.
\end{desideratum}

\parahead{Individual rationality (IR)}
The standard IR definition in mechanism design is to require that each agent be no worse off after participation
when compared to the initial state, \ie $\utili(M, \agsubset, \strataccpart) \geq \utilalloci(\initdata)$.
However, the example in Figure~\ref{fig:irillus} illustrates why such a requirement is problematic in our setting.
A more sensible IR requirement, defined below, is that
an agent's utility when participating with the accepting policy is higher than when
she does not participate, but others do.

\begin{desideratum}
[Individual rationality (IR)]
\label{des:ir}
Suppose $\agsubset\subset \allagents \backslash\{i\}$ participate with the accepting strategy.
Then, $i$'s utility by participating with the accepting strategy is no worse than
when she does not participate.
That is, for all $\agsubset\subset\allagents\backslash\{i\}$, we have
$\utili(M, \agsubset\cup\{i\}, (\stratacci, \strataccpart))
\geq \utili(M, \agsubset, \strataccpart)$.
\end{desideratum}

A few remarks are in order:
\emph{1) IR is not trivial:} While each exchange benefits the agents involved, IR does not follow automatically. In~\S\ref{sec:IR-sketch}, we show 
that by not participating, an agent can alter the exchanges proposed by a planner, and thereby block more competitive agents from receiving additional goods, which could increase her utility.
\emph{2) NIC $\centernot\implies$ IR:} Non-participation differs from participating and rejecting all proposals, as the  exchanges proposed by a planner depends on which agents are present. Thus, NIC does not imply IR.
\emph{3) IR regardless of others' strategies:} Requiring that an agent benefits regardless of others’ strategies is also too strong here. For instance, in Figure~\ref{fig:irillus}, suppose $i$ follows a pathological strategy: accepting all proposals when $k$ participates, and rejecting all when $k$ does not. In that case, $k$ is better off not participating. Our definition ensures that others behave sensibly regardless of an agent's participation.

\parahead{Stability}
In classical models for exchanging goods~\citep{shapley1974cores,gale1962college,varian1973equity}, stability means that upon termination, no two agents can gain from further exchanges.
We formalize a similar notion here.
Note that $\feaspropsp(\data, \partagents)$ contains
precisely the set of exchanges that are mutually beneficial to both agents from $\data$:
any exchange $((i,r),(j,s))$ in $\feaspropsp(\data, \agsubset)$ will increase both $i$ and $j$'s utilities by $1-\betaij>0$  (see~\eqref{eqn:feasibleexchanges}).

\begin{desideratum}[Stability]
\label{des:stability}
An allocation $\data = \{\datai\}_{i\in \allagents}$ is stable for a subset $\agsubset\subset \allagents$
if $\feaspropsp(\data, \agsubset) = \emptyset$, \ie no two agents in $\agsubset$ can
exchange goods to increase their utilities. A protocol is stable if, upon termination, any such exchange $p\in\feaspropsp(\datafinal(M, \agsubset, \stratpart), \agsubset)$ has
already been rejected. Thus, if all agents follow the accepting policy, the terminal allocation 
$\datafinal(M, \agsubset, \strataccpart)$ is stable for $\agsubset$.
\end{desideratum}

\subsection{Impossibility Results}
\label{sec:impossibilityresults}

We conclude this section with three impossibility results
to support the above model.

\parahead{NIC vs DSIC protocols}
The following result shows that no nontrivial DSIC protocols are possible in this problem. A proof sketch is given in~\S\ref{sec:nodsic-proof}. This motivates our focus on NIC.

\begin{restatable*}{theorem}{nodsic}
    \label{thm:nodsic}
    There exists a problem instance $(\initdata, \beta)$ on which no protocol can
    simultaneously satisfy DSIC and stability.
\end{restatable*}

\parahead{Initial allocations are revealed truthfully}
Our model assumes that the initial allocations of participating agents are truthfully revealed.
As noted in the real-world examples in~\S\ref{sec:intro}, this reflects practical cases where such information is publicly known.
This assumption is further motivated by Theorem~\ref{thm:truthfulreporting} below, which establishes that truthfully eliciting $\initdata$ is impossible in our setting.

\subparahead{A modified setting}
To state this result formally,
let $\inityi \in \{0,1\}^m$ denote agent~$i$'s \emph{true} initial allocation, and denote $\inity = \{\inityi\}_{i\in\allagents}$.
When the planner must rely on agents to report their initial allocations, the interaction proceeds as described below.

First, the planner designs and publishes an exchange protocol.
Each agent~$i$ then reports $\initdatai$ representing $\inityi$ (not necessarily truthfully).
When reporting, we assume that an agent may choose to hide a good she has, but cannot claim to have a good she does not, since she would be unable to give it to another agent during an exchange.
Formally, if $\inityir=0$, then she must report $\initdatair=0$; however, if $\inityir=1$, she may choose to untruthfully report $\initdatair=0$.
Based on the reported $\initdata$, the planner then proposes exchanges over multiple rounds, as described in~\S\ref{sec:setup}.
If an agent chooses to hide a good she has and receives the good from another agent, she simply retains one copy.

In this setting, we say a protocol is stable if it
satisfies Desideratum~\ref{des:stability} with the reported allocations $\initdata$.
We have the following theorem,
whose proof sketch is given in~\S\ref{sec:truthfulreporting-proof}.

\begin{restatable*}{theorem}{truthfulreporting}
\label{thm:truthfulreporting}
There exists a problem instance $(\inity, \beta)$ such that for any stable exchange protocol, all agents truthfully reporting the initial allocation and accepting all exchange proposals does not constitute a Nash equilibrium.
\end{restatable*}

\parahead{Multi-round vs single-round protocols}
We focus on multi-round protocols since a protocol that only proposes exchanges once cannot be simultaneously stable and NIC.
To state this formally, we say that $\mech$ is a single round protocol if
it always chooses some $\Pone\subset\feasprops(\partagents)$
and then terminates on round 2 with $\Ptwo = \emptyset$.
We have the following observation,
whose proof is given in Appendix~\ref{sec:singleround}.

\begin{restatable*}{observation}{singleround}
\label{obs:singleround}
    Let there be $m=2$ goods and three agents.
    Then, there exists $\initdata$ such that, for all competition levels $\beta$, no single-round protocol can simultaneously satisfy stability and NIC.
\end{restatable*}

That said, it is worth pointing out here that our multi-round protocol will terminate in a single round when all agents follow the accepting policy.

\subsection{Proof sketch of Theorem~\ref{thm:nodsic}}
\label{sec:nodsic-proof}

Consider 
the following instance with 4 agents $i, j, k, l$ and 2 goods.
The initial allocations are $\initdatai = \initdatak = [1, 0]$, and $\initdataj = \initdatal = [0, 1]$.
Let $\beta$ be such that $\beta_{ij} < \beta_{jk} < \beta_{kl} < \beta_{il}$.
In this example, we can refer to exchanges by only agent names without any ambiguity. For example, the exchange 
where $i$ gives good $1$ to $j$ and $j$ gives good $2$ to $i$ can be written $ij$ instead of $((i, 1), (j, 2))$.

\parahead{Representing protocols as trees}
We can represent each proposal via the following tree:
the root node lists the exchanges proposed in the first round, and each child branch corresponds to the agents’ accept/reject decisions.
Subsequent nodes record the protocol’s next proposals given those decisions.
The leaves represent the final allocations, based on the proposals and decisions.
The trees have finite height since
any protocol must terminate in a finite number of rounds (as it cannot re-propose previously rejected exchanges).
These trees are illustrated in Figure~\ref{decision-tree} in Appendix~\ref{sec:impossibility}.

\parahead{Protocols with poorly sequenced exchanges}
The first key insight in this proof is that if a protocol proposes exchanges in an order that \emph{inverts} an agent's preferences---offering an exchange with a more competitive agent before a less competitive one---then this agent can benefit by strategically rejecting the first proposal if other agents follow certain strategy profiles.
Formally, we define an \emph{inversion pair} $(i',j',k')$ if $\beta_{i'j'} < \beta_{i'k'}$ (so $i'$ prefers to trade with $j'$ over $k'$) but $i'k'$ is an ancestor of $i'j'$ in the corresponding tree.
Note that $i'j'$ is only feasible if $i'k'$ is rejected.

We show that the existence of an inversion pair implies a beneficial deviation from the accepting strategy for $i'$.
In particular, suppose $j'$ and $k'$ accept all their proposals, while the remaining agent $l'$ who is not a part of the inversion pair rejects all proposals. 
Since $l'$ rejects all, we know that any exchanges involving $l'$ will not happen. If $i'$ knows that, she should reject the earlier $i'k'$ and wait to accept the later $i'j'$ so that the only accepted exchange is $i'j'$. If she follows the accepting policy instead, the only accepted exchange will be $i'k'$, which is less favorable for her.
This violates DSIC as accepting all proposals is not optimal for $i'$ \emph{regardless} of others' strategies.

\parahead{The remaining protocol}
We then prove that there is only one remaining stable protocol which does not admit an inversion pair on this instance, which proposes only one exchange on each round that has the lowest competition level among all feasible exchanges.
However, even in this protocol, we show that some agent can benefit by rejecting under a carefully chosen strategy profile.

\subsection{Proof sketch of Theorem~\ref{thm:truthfulreporting}}
\label{sec:truthfulreporting-proof}

The key intuition is that by under-reporting her goods, an agent can induce the protocol to schedule more exchanges involving her, thereby blocking others from obtaining goods.  
Our proof specifically considers 3 agents and 2 goods, and how a protocol would schedule exchanges under 3 possible \emph{reported} initial allocations:
\emph{Case 1: } $\initdatai = \initdataj = [1, 0]$ and $\initdatak = [0, 1]$;
\emph{Case 2: } $\initdatai = [1, 1]$ instead;
and \emph{Case 3: } $\initdataj = [1, 1]$ instead.
In case~1, any stable and NIC protocol should propose \emph{exactly one} of the two possible exchanges $((i,1),(k,2))$ or $((j,1),(k,2))$.
In cases~2 or~3, it should propose the unique feasible exchange $((j,1),(k,2))$ or $((i,1),(k,2))$, respectively.

Suppose the protocol proposes $((i,1),(k,2))$ in case~1.
In this case, if $i$ truly had both goods, truthfully reporting them would result in both $j$ and $k$ receiving goods (case 2); however, if she pretends to not have good 2, then the protocol is in case 1, and proposes an exchange between $i$ and $k$.
While $i$ does not actually gain a good (since she already has it), she prevents her competitor $j$ from gaining an additional good.
A symmetric argument holds if the protocol instead proposes $((j,1),(k,2))$ in case 1, as agent $j$ could benefit by pretending to not have good 2.
Thus, regardless of the protocol’s choice in case~1, some agent can under-report her initial goods and increase her utility by preventing a competitor from gaining additional goods.

\begin{remark}
One could consider the following alternative setting,
which mirrors traditional mechanism design 
settings: instead of proposing pairwise exchanges, we have the agents report their initial goods to the planner, who then reallocates access.
The goal of the planner is to design a mechanism to incentivize agents to truthfully report their initial goods, while satisfying IR and stability.
The intuition developed above suggests that incentivizing truthful reporting is impossible, even in this alternative formulation.
\end{remark}

\section{Method}
\label{sec:method}

We now describe our method. We start with an intuitive idea, identify its
shortcomings, and describe how it can be improved to obtain the desired properties.

\parahead{An initial attempt}
A natural approach to designing pairwise exchanges is as follows:
In each round, evaluate all possible pairwise exchanges and propose a set of exchanges
that satisfy two conditions:  
\emph{(a)} no agent receives duplicate items, and
\emph{(b)} if all proposals are accepted, the allocation is stable, \ie no further
pairwise exchanges are possible.
Formally,
let $\Pt$, initialized to $\emptyset$, be the set of proposed exchanges
for the current
round. Let $\interdatat$, initialized to $\datatmo$, represent the allocation if all
proposals in $\Pt$ were accepted.
Consider a protocol which iteratively updates $\Pt$ and
$\interdatat$ as follows:
\emph{(1)} Select an arbitrary exchange proposal $((i, r), (j, s))$ from $\feaspropsp(\interdatat, \partagents)$ (see~\eqref{eqn:feasibleexchanges}) that has not been rejected before, and add it to $\Pt$.  
\emph{(2)} Update the allocation: $\interdatat_{i,s} = \interdatat_{j,r} = 1$, pretending that
both agents will accept this proposal.  
\emph{(3)} If $\feaspropsp(\interdatat, \partagents)$ with the updated $\interdatat$ is non-empty,
go back to step 1.
\emph{(4)} Otherwise, \ie if
$\feaspropsp(\interdatat, \partagents) = \emptyset$,
stop and propose the exchanges in $\Pt$ to the agents.
By construction, if all proposals are accepted, we reach a stable allocation, and the
protocol terminates.
This idea is illustrated in Figure~\ref{fig:naivesub}.

\insertFigNaiveNotNic

\subparahead{Shortcomings}
This approach is not \emph{NIC}. Strategic agents may decline proposals for two reasons:

\begin{enumerate}[leftmargin=0.2in]
    \item \emph{Level of competition:}
From~\eqref{eqn:utilsummary}, we see that an exchange with agent $j$ will increase $i$'s utility by $1-\betaij$ as both agents gain one additional good.
Suppose agent $k$ is less competitive with agent $i$ than agent $j$ is, \ie $\betaik < \betaij$.
Then $i$ might reject an exchange with $j$,
anticipating that the planner might propose a similar exchange with $k$ in a
future round.
While $i$'s allocation is the same in either scenario, she benefits since a less competitive agent gains a good.

    \item \emph{Rarity of goods:}
    An agent $i$ may reject an exchange involving a commonly available good with agent $j$, hoping to receive the common good from some other agent $k$ and instead receive a rarer good from $j$ in a later round. This is illustrated in Figure~\ref{fig:notnicsub}. 
\end{enumerate}

\insertAlgoMain

\paraheadwnogap{Our method}
Our protocol, \algname{} (\textbf{C}ompetitive-order \textbf{L}azy
\textbf{E}xch\textbf{A}nges with
\textbf{R}etrospection), outlined in Algorithm~\ref{alg:main},
builds on the above idea but resolves the two shortcomings.

The first shortcoming is straightforward to fix: we iterate over all pairs of agents $(i,j)$ in
increasing order of competition factor $\betaij$,
 scheduling as many exchanges as possible between them.
 This ensures that each
agent has the opportunity to exchange with less competitive agents first.
We construct \emph{demand sets} $\dmdifromj$ and $\dmdjfromi$,
listing goods
that each agent can offer to the other based on $\interdatat$
(lines~\ref{lin:dmdjfromi} and~\ref{lin:dmdifromj}).
If both demand sets are non-empty (line~\ref{demand-nonempty}),
we select goods
$r \in\dmdjfromi$ and $s\in \dmdifromj$ respectively
such that the exchange $((i, r), (j, s))$ has not been rejected before. If multiple such goods are available, we pick the lexicographically smallest pair $(r,s)$.
We add this exchange to $\Pt$ and update $\interdatat$.
If all such exchanges have been rejected, we move to the next agent pair.

Addressing the second issue is much more challenging.
The lazy selection of exchanges in previous agent pairs might affect exchanges we can schedule in the current pair.
Hence,
if no more exchanges are possible between two agents $(i,j)$, 
we retrospectively check if more exchanges between $i$ and $j$ could have been possible,
if previous exchanges involving either $i$ or $j$ had been scheduled differently.
For this,
we develop the \retrospect{} subroutine which recursively checks for such adjustments. We will describe this subroutine shortly.
If either demand set $\dmdifromj, \dmdjfromi$ is empty, we invoke
\retrospect{} (lines~\ref{IC1},~\ref{IC2})
to test if any exchanges already scheduled in $\Pt$
 could be adjusted to accommodate more exchanges between $i$ and $j$. 
If all calls to \retrospect{}
succeed, then $\interdatat$ and $\Pt$ would have been modified to make room for one more
exchange between $i$ and $j$.
If either call returns \retrofail, it means further adjustments are not
possible and we stop and proceed to the next agent pair
(line~\ref{no-exchange2}).
In this case, we revert any changes made to $\interdatat$ and $\Pt$.

\subparahead{Termination}
Once the algorithm constructs the proposals $\Pt$, they are presented to the agents,
who decide whether to accept or reject them. If all proposals are accepted,
the protocol terminates (line~\ref{lin:termination}). Otherwise, rejected proposals
are added to \rejprops, the allocation $\interdatat$ is updated for the next round,
and the process repeats (lines~\ref{lin:agentphasestart}--\ref{lin:agentphaseend}).

\parahead{The \retrospect{} subroutine}
\retrospect{} \emph{retro}spectively \emph{trace}s 
scheduled exchanges to create room for additional ones.  
When $\dmdifromj = \emptyset$, we invoke $\retrospectmath(i, j, \{i, j\})$ to check whether  
any of $i$'s scheduled exchanges in $\Pt$ can be adjusted to allow $i$ to receive a good from $j$.  
On line~\ref{all-possible}, \retrospect{} examines each good $s$ that $j$ could offer $i$.  
We check $\datatmo_{i,s} = 0$ instead of $\interdatat_{i,s} = 0$
because $\interdatat_{i,s}$ may have been updated by exchanges in $\Pt$.
Line~\ref{base-case} handles the base case where a good $s$ is available for $i$
to receive from $j$ (this is never triggered in the initial recursive call).

In line \ref{query}, we know $\interdatat_{i, s} = 1$, \ie $i$ is scheduled to receive good $s$ in the current round via some exchange $((i, r), (j', s)) \in \Pt$.
If $j'\notin S$, we invoke  $\retrospectmath(i, j', S\cup\{j'\})$
to recursively check whether
$j'$ can offer $i$ an alternative good $s'$ so that $s$ is cleared for $i$ to exchange
with $j$ (line \ref{recursive-call}).
We then apply the adjusted exchange $((i, r), (j', s'))$, thus freeing up the slot $\interdatat_{i, s}$ on line \ref{alter-x}
and return $s$ on line \ref{recursive-case}. If the recursive call
returns \retrofail{} or $((i, r), (j', s'))$ has been rejected, we proceed to the next good.
If we
have examined all such goods without success,
we return \retrofail{} (line \ref{all-examined}).

\subparahead{An example of \emph{\retrospect}}
Let us revisit Figure~\ref{fig:notnicsub}.
Say the algorithm lazily scheduled $((i, 1), (j, 2))$ when processing $(i,j)$ and we are in case 1. 
When the algorithm finds that $X_{j \leftarrow k} = \emptyset$, it invokes \retrospect$(j, k, \{j, k\})$ . The only possible good $s$ on line \ref{all-possible} is good $1$,
which $k$ possesses but $j$ has received from $i$. Now $j'$ is agent $i$ and the algorithm recursively
invokes \retrospect$(j, i, \{i, j, k\})$ on line \ref{recursive-call}.
Here, we find $s =
3$ as an alternative good $i$ can give $j$. The initial \retrospect{} call takes this information and adjust $\interdata_j$ on line \ref{alter-x} to be like case 2, opening a slot for
$j$ and $k$ to exchange. Eventually we end up with the final allocation in case 2.
When there are many agents, we may need to change multiple previously schedules exchanges, which is handled by recursion.

\parahead{Additional remarks}
This concludes the description of our method. A few remarks are in order.

\subparahead{1) Measuring rarity requires retrospection}
We note that there is no simple way to measure the \emph{rarity} of a good without performing a full retrospection. For instance, counting the number of agents initially possessing a good (with fewer agents implying rarity) is a poor heuristic as it fails to capture the effect of exchange order. Suppose only agents $i$ and $j$ initially hold a good, which appears to be rare, but $i$ has lower competition factors with others than $j$ does. \algname{} will have scheduled exchanges widely with $i$'s good before $j$ has a chance to trade. As a result, the good may no longer be rare when we schedule exchanges for $j$. Due to such complex interactions, we believe that retrospection is the most reliable way to address incentive issues related to rarity.

\subparahead{2) Termination in one round}
\algname{} terminates in a single round when all agents follow the accepting strategy. This does not contradict Observation~\ref{obs:singleround}, since the role of potential exchanges in later rounds is to both incentivize acceptance of current proposals and ensure stability.

\subparahead{3) Runtime}
For clarity,
 Algorithm~\ref{alg:main}
presents an un-optimized version of \retrospect{}.
In Appendix~\ref{sec:runtime}, we outline a computationally efficient version of
\retrospect{} and analyze its runtime.
With this version,
each iteration of \algname{} runs  in $\bigO(|\partagents|^2 m^2 \max\{|\partagents|, m\})$ time.

\section{Analysis}
\label{sec:analysis_sketch}

The following theorem, our main theoretical result, summarizes the key properties of Algorithm~\ref{alg:main}.

\begin{restatable*}{theorem}{main}
\label{thm:main}
\emph{\algname} satisfies the NIC, IR, and stability desiderata outlined in~\S\ref{sec:desiderata}.
Moreover, when all agents follow the accepting policy, \emph{\algname} terminates in one iteration with a stable allocation.
\end{restatable*}

Next, in~\S\ref{sec:NIC-sketch} and~\S\ref{sec:IR-sketch}, we sketch the proofs for NIC and IR.
The proof of stability follows quite straightforwardly by the design of the algorithm. The full proof is given in Appendix~\ref{sec:analysis}.

\subsection{Proof sketch for NIC}
\label{sec:NIC-sketch}

We need to show that, when others are following the
accepting strategy, agent $i$'s utility is maximized by also following the accepting strategy.
Let $\stratacci$ denote the accepting strategy profile for $i$.
Let $\exchangesinvolvi(\strati), \exchangesnotinvolvi(\strati)$ respectively
        denote the \emph{accepted} exchanges involving $i$ and those not involving $i$, when $i$ follows strategy $\strati$ and the others follow the accepting strategy.
Every exchange $((i,r),(j,s)) \in\exchangesinvolvi(\strati)$ increases her utility by $1-\betaij$ and every exchange $((j,r),(k,s)) \in\exchangesnotinvolvi(\strati)$ decreases her utility by $\betaij+\betaik$ (see~\eqref{eqn:utilsummary}), 
we can decompose $i$'s utility as follows:
        \begin{align*}
       \underbrace{\onev^\top \initdatai - \sum_{j\neq i} \betaij \onev^\top \initdataj}_{\utiliI}
                + \underbrace{\sum_{((i,r),(j,s))\in\exchangesinvolvi(\strati)} (1-\betaij)}_{\utiliE(\strati)}
                \,-\, \underbrace{\sum_{((j,r),(k,s))\in\exchangesnotinvolvi(\strati)} (\betaij+\betaik)}_{\utiliL(\strati)}.
        \end{align*}
        Here, $\utiliI$ is $i$'s initial utility, which does not depend on $i$'s strategy.
        Next,
        $\utiliE(\strati)$ is the increase in her utility due to \emph{accepted} exchanges involving her.
        Finally, $\utiliL(\strati)$ is the reduction in her utility due to exchanges not involving her.
        To control the third term, note that since other agents are following the accepting policy,
        when agent $i$ also follows the accepting policy, $\exchangesnotinvolvi(\stratacci)$ will be precisely the set of exchanges \emph{proposed} by \algname{} in the first round not involving $i$; recall that \algname{} terminates in one round if all agents follow the accepting strategy.
        If $i$ follows any other policy, then $\exchangesnotinvolvi(\strati)$ will contain $\exchangesnotinvolvi(\stratacci)$ and any other exchanges between other agents scheduled in subsequent rounds.
        Hence, $\exchangesnotinvolvi(\stratacci) \subset \exchangesnotinvolvi(\strati)$ and consequently $\utiliL(\stratacci) \leq \utiliL(\strati)$.
        Therefore, it is sufficient to show that $\utiliE(\stratacci) \geq \utiliE(\strati)$ for every other strategy $\strati$.

\parahead{Maximizing the number of less competitive exchanges}
Suppose $i$ gets $\numexchacc$ exchanges when following $\stratacci$ and $\numexch$ exchanges when following $\strati$.
Let $\beta_{ij^\star_1},\beta_{ij^\star_2},\dots,\beta_{ij^\star_{\numexchacc}}$ be the competition factors between $i$ and who she has exchanged with, sorted in ascending order, when $i$ follows $\stratacci$.
Let $\beta_{ij_1},\beta_{ij_2},\dots,\beta_{ij_{\numexch}}$ be the same when $i$ follows $\strati$.
To show that $\utiliE(\stratacci) \geq \utiliE(\strati)$, we need to show 
        $\sum_{p = 1}^{\numexchacc} (1 - \beta_{ij^\star_p})
        \geq \sum_{p = 1}^{\numexch} (1 - \beta_{ij_p})$.
To that end, we establish a key technical result: 

    \begin{restatable}[Informal]{lemma}{betagood}
        \label{lem:betagood}
      For any $\beta' \in (0, 1)$, to maximize the total number of exchanges with all agents $j$ such that $\beta_{ij} \leq \beta'$ in $\exchangesinvolvi(\strati)$,  agent~$i$ must follow the accepting policy $\stratacci$.
    \end{restatable}

    By choosing $\beta' = 1$ in lemma \ref{lem:betagood},
        we have $\numexchacc \geq \numexch$.
        Therefore, it is sufficient to show
        $\beta_{ij^\star_p} \leq \beta_{ij_p}$
        for any $p \leq \numexch$.
        By way of contradiction, assume that there exists $p$ such that $\beta_{ij^\star_p} > \beta_{ij_p}$.
        Then, under the accepting strategy $\stratacci$, there are at most $p-1$ exchanges 
        with agents $j$ such that $\betaij \leq \beta_{ij_p}$; and
        under $\strati$, there are
        at least $p$ such exchanges.
        This contradicts Lemma~\ref{lem:betagood}.

\parahead{Proof of Lemma~\ref{lem:betagood}}
This result is the most technically challenging part of the proof. As it relies on the detailed mechanics of our method, we provide only a high-level overview here.
We first show that although \retrospect{} may revise previously scheduled exchanges to accommodate additional exchanges between the current agent pair, it does not alter the \emph{number} of exchanges between earlier agent pairs in a round.
Building on this, we argue that if agent $i$ wishes for the planner to schedule an additional exchange with some agent $j$ in a future round, the only way to do so is by rejecting an exchange with some agent~$k$ such that $\betaik \le \betaij$.
Intuitively, this is undesirable, as it requires giving up an exchange with a less competitive agent in favor of a more competitive one.

\insertFigIRProofIllus
\subsection{Proof sketch for IR}
\label{sec:IR-sketch}

We first describe the main challenge in proving IR.
While each exchange involving an agent~\( i \) directly increases her utility, 
by not participating, \( i \) could change \algname{}’s scheduling of exchanges, 
potentially blocking more competitive agents from receiving goods and thereby improving her utility. 
To illustrate,
consider an instance with $\initdata_i = \initdata_k = [1, 0]$, $\initdata_j = \initdata_l = [0, 1]$, and $\beta_{ij} < \beta_{il} < \beta_{jk} < \beta_{kl}$, as in Figure~\ref{fig:irproofillus}. Suppose all participating agents follow the accepting strategy.
When \( i \) participates, \algname{} schedules two exchanges 
$\{((i, 1), (j, 2)), ((k, 1), (l, 2))\}$, giving all four agents both goods.
If \( i \) does not participate, \algname{} instead proposes 
$\{((k, 1), (j, 2))\}$, and agent \( l \), the one more competitive with \( i \), receives nothing.
Thus, by not participating, $i$ could potentially block highly competitive agents from receiving goods.
We need to show that, despite such effects, participation remains beneficial, as $i$ herself receives new goods by participating. Indeed, in this example, if~\( i \) participates, her utility is \( 2 - 2(\beta_{ij} + \beta_{ik} + \beta_{il}) \), whereas if she does not, it is \( 1 - 2(\beta_{ij} + \beta_{ik}) - \beta_{il} \) (see~\eqref{eqn:utilsummary}). As \( \beta_{il} < 1 \), participation is better.

The key technical hurdle in proving this for a general instance is that any potential exchange involving agent~\( i \) can---when she does not participate---trigger a cascade of knock-on effects that alter many subsequent exchanges, particularly when the number of agents and goods is large.
We now describe how our proof tracks these knock-on effects to establish IR.

\parahead{Utility decomposition}
Let $\Pcal \subset \allagents \setminus \{i\}$ be a set of agents excluding $i$ who are already participating.
Let us denote  $i$'s utility when joining $\Pcal$ by $U_i$ and when not participating by  $U'_i$.
Let $\pairnum_{jk}, \pairnumalt_{jk}$ denote the number of exchanges between agent pair $(j, k)$ when $i$ participates and does not participate respectively.
Recall that every exchange $((i, r), (j, s))$ involving $i$ improves her utility by $(1-\betaij)$ as both $i$ and $j$ gain a good, and that every exchange $((j, r), (k,s))$ not involving $i$ decreases her utility by $-\betaij-\betaik$ (see~\eqref{eqn:utilsummary}).
Hence, we can write $U_i$ and $U'_i$ as follows:
\begingroup
\allowdisplaybreaks
\begin{align*}
    U_i &= \Big(\onev^\top \initdatai -
        \sum_{j \in \allagents \setminus \{i\}} \betaij \onev^\top \initdataj \Big)
        \;+\; \sum_{j, k \in \partagents} \pairnum_{jk} (-\betaij - \betaik)
        \;+\; \sum_{j \in \partagents } \pairnum_{ij} (1 - \betaij), \\
    U'_i &= \Big(\onev^\top \initdatai -
        \sum_{j \in \allagents \setminus \{i\}} \betaij \onev^\top \initdataj \Big)
        \;+\;
        \sum_{j, k \in \partagents } \pairnumalt_{jk} (- \betaij - \betaik).
\end{align*}
\endgroup
Above, the first term in $U_i$ is $i$'s initial utility, the second term is her utility loss due exchanges involving others when she participates, and the third term is her utility gain due to exchanges involving herself.
Similarly, the second term in $U'_i$ is her loss due to others' exchanges when she does not participate.
We therefore have the following expression for $U_i - U'_i$, the increase in utility when $i$ participates, which depends only on the number of exchanges between each pair of agents.
We will show that the RHS of the expression below is nonnegative, which implies IR. We have:
\begingroup
\allowdisplaybreaks
    \begin{align*}
    U_i - U'_i
    &= 
    \sum_{j \in \partagents} \pairnum_{ij} (1 - \betaij) \;+\; 
    \sum_{j, k \in \partagents} (\pairnum_{jk} - \pairnumalt_{jk}) (- \betaij - \betaik)
    \numberthis
    \label{eqn:irutilsummary}
\end{align*}
\endgroup

\paraheadwnogap{A graph construction}
To track the knock-on effects due to $i$'s (non-)participation, we construct a lattice graph $G = (V, E \cup E')$ where each agent-good pair is a node in $V$. The edges $E$ and $E'$ correspond to the set of exchanges in \algname{} when agent $i$ participates and does not participate, respectively. For each exchange $((j, r), (k, s))$, we add an edge connecting $(j, s)$ and $(k, r)$. Note that edges connect goods \emph{acquired}.
This construction is illustrated in Figure~\ref{IR-proof-graph-intro}.

\insertFigIRIntro

\emph{Tracked paths.}
We will use this graph to track the changes in the scheduled exchanges when $i$ participates versus when she does not.
The key tool is a construct called a \emph{tracked path}: formally, it is defined as a connected component in the graph which has at least one vertex representing $i$'s goods.
We prove that such connected components are in fact, paths, whose edges alternate between $E$ and $E'$;
intuitively, since an agent never receives the same good via multiple exchanges in \algname, no two edges in $E$ share the same vertex, and likewise for $E'$.
As we will see, this construction allows us to compute $i$'s utility changes due to her (non-)participation
by summing utility differences along all tracked paths.
We illustrate 2 tracked paths in Figure \ref{fig:IRsub1-intro} and 1 tracked path in Figure~\ref{fig:IRsub2-intro}. 

Let us call exchanges in a tracked path \emph{tracked exchanges} and call other exchanges \emph{untracked}.
Let $\pairnumtr_{jk}, \pairnumtralt_{jk}$ be the number of tracked exchanges between $j$ and $k$ in $E$ and $E'$ respectively.
Note that $\pairnumtr_{ij} = \pairnum_{ij}$ and $\pairnumtralt_{ij} = \pairnumalt_{ij} = 0$ for exchanges involving $i$ as every exchange involving $i$ is tracked.
The following lemma establishes an important relationship between $\pairnumtr_{jk}, \pairnumtralt_{jk}$ and $\pairnum_{jk}, \pairnumalt_{jk}$
from~\eqref{eqn:irutilsummary}.

\begin{restatable}[Informal]{lemma}{irtrackedpath}
\label{lem:IR-tracked-path}
    The number of untracked exchanges between each agent pair in $E$ and $E'$ is the same.
    That is, for all $j, k \in \partagents \cup \{i\}$,
    $\pairnum_{jk} - \pairnumtr_{jk} = \pairnumalt_{jk} - \pairnumtralt_{jk}$.
\end{restatable}

Combining~\eqref{eqn:irutilsummary} and Lemma~\ref{lem:IR-tracked-path}, we have
$U_i - U'_i
    = 
    \sum_{j \in \partagents} \pairnumtr_{ij} (1 - \betaij) \;+\; 
    \sum_{j, k \in \partagents} (\pairnumtr_{jk} - \pairnumtralt_{jk}) (- \betaij - \betaik)$.
As this expression only depends on the number of tracked exchanges $\{\pairnumtr_{jk}, \pairnumtralt_{jk}\}_{j,k}$,
we can express it as a sum of utility changes for $i$ on each tracked exchange.
As tracked exchanges arise from tracked paths, 
it can therefore be expressed as a sum of utility changes along each tracked path.

\parahead{Computing utility changes along tracked paths}
An intuitive way to think about tracked paths is as follows: while the exchanges scheduled by \algname{} may differ based on $i$'s (non-)participation, the goods acquired by agents only differ at the start and end of the tracked paths.
For instance, in the tracked path of length 3 from $i$ to $l$ in Figure \ref{fig:IRsub1-intro}, 
 agents $j$ and $k$ receive goods $1$ and $2$ respectively, regardless of $i$'s participation;
 they just receive them through different exchanges.
This intuition allows us to argue that many utility differences along tracked edges ``cancel out''.

Formally, we show that there are three types of tracked paths.
While each affect $i$'s utility differently,
they all  contribute non-negative utility when $i$ participates.
Note that we can view each tracked path as ``starting'' with an exchange in $E$ involving $i$ (see Figure~\ref{IR-proof-graph-intro}).
Then, a tracked path might: (1) have even length and end on some other agent $j$ in $E'$, (2) have odd length and end on some other agent $k$ in $E$, (3) have odd length and end on $i$ herself in $E$. Type (1) and (2) are illustrated in Figure \ref{fig:IRsub1-intro} and type (3) is illustrated in Figure \ref{fig:IRsub2-intro}. We show that tracked paths of type (1) results in a utility increase of $1 + \beta_{ij} \geq 0$ when $i$ participates, type (2) results in a utility increase of $1-\beta_{ik} \geq 0$ (recall $\beta$ values are bounded in $(0,1)$), and type (3) results in a utility increase of 2.
IR follows since the utility change along a tracked path is always non-negative.

\parahead{Proof sketch of Lemma~\ref{lem:IR-tracked-path}}
This result is quite technical, so we only provide a high-level overview. We must use induction because the number of untracked exchanges for the current pair depends on the number of untracked exchanges for previous pairs. The key intuition involves two observations. First, \retrospect{} always maximizes the number of total exchanges between the current pair of exchanging agents given that it does not alter the number of exchanges between any previous pairs. Second, we may replace the set of untracked exchanges in $E$ by those in $E'$, and vice versa, because untracked exchanges are disjoint from any tracked paths. Thus, if a pair of agent has different numbers of untracked exchanges in $E$ and $E'$, we may replace one set of untracked exchanges by the other to increase their exchanges, contradicting the optimality of \retrospect{}.

\section{Conclusion}
\label{sec:conclusion}

To the best of our knowledge, we are the first to study the exchange of freely replicable goods among competitive agents. 
We formalized this problem through three key desiderata---Nash incentive-compatibility, individual rationality, and stability---and proposed a protocol that satisfies all three. Our model and analysis reveals sharp departure from classical, non-competitive models:
we showed that weaker notions of incentive compatibility and IR are necessary here, truthful revelation of initial goods cannot be enforced, and that stability and Pareto-efficiency  can diverge when the externalities are high.
Some open questions resulting from our work include proving Conjecture~\ref{conj:clearperfagent}, and better understanding when Pareto-efficient outcomes are possible.

Our work also opens several promising directions for future work, including settings where agents can share goods they receive from others, divisible goods, different valuations for different goods, and nonlinear utilities over goods possessed by oneself or others.

\bibliographystyle{ACM-Reference-Format}
\bibliography{bib_coopshare}

\appendix

\section{Proof of Theorem \ref{thm:main}}
\label{sec:analysis}

We now present the proof of Theorem~\ref{thm:main}.
We will begin with the following lemma, which states some
properties about the \retrospect{} subroutine.

\begin{lemma}[Correctness of \retrospect] \label{RT-correctness}
    A call to \emph{\retrospect{}}
    \emph{(i)} always terminates,
    \emph{(ii)} does not alter the number of goods any agent possesses (\ie $\sum_{s=1}^{\numgoods}\interdatat_{i,s}$ does not change for any $i$ after any \emph{\retrospect} call),
    \emph{(iii)} does not alter the number of exchanges scheduled between any pair of agents (\ie suppose $\pairnum_{ij}$ exchanges are scheduled between agent $i, j$ before a \retrospect{} call, then we still have $\pairnum_{ij}$ exchanges scheduled between them after the call), and
    \emph{(iv)} returns a good which the giving agent $j$ could give to the receiving agent $i$, or returns \emph{\retrofail{}} if no such good exists.
\end{lemma}

\begin{proof}
First, for statement (i), \retrospect{} always terminates because each recursive call is made on an expanding set of agents $S$. On lines~\ref{condition} and~\ref{recursive-call}, the agent $j'$ selected for the recursive call is guaranteed to be not in $S$, and is added to $S$ in the subsequent call. Since $S$ strictly grows with each call, the recursion must eventually terminate when all agents have been examined.

Next, for statements (ii) and (iii), note that line~\ref{alter-x} is the only point where \retrospect{} modifies the allocation. On this line, agent~$i$ receives a new good from agent~$j'$ and simultaneously gives up the original good. As a result, the total number of goods held by any agent and the number of exchanges scheduled between any pair of agents both remain unchanged.

Finally, statement (iv) follows from the fact that the original good $s$ is returned on either line~\ref{base-case} or line~\ref{recursive-case}. Line~\ref{base-case} corresponds to the base case, where we know from the condition $\interdatat_{i,s} = 0$ that agent~$j$ can directly give $s$ to agent~$i$. Line~\ref{recursive-case} handles the recursive case: if the recursive call on agent~$j'$ successfully finds an alternative good $s'$ that can be given to $i$, then line~\ref{alter-x} uses that exchange to free up $s$, making it available to return. Conversely, if no such empty slot $s$ exists and the allocation cannot be adjusted to create one, the procedure returns \retrofail{} on line~\ref{all-examined}.
\end{proof}

\subsection{Nash Incentive Compatibility}
\label{sec:NIC}

First, we will show that the proposed exchange protocol \algname{} is Nash incentive-compatible. Suppose all agents except agent $i$ are following the accepting policy. We need to demonstrate that adhering to the accepting strategy $\stratacci$ is at least as good as any other strategy for agent $i$. 
To establish this, we will proceed through the following three steps.
    \begin{enumerate}
    \item We decompose agent $i$'s utility into components from exchanges involving $i$ and those that do not. Exchanges not involving $i$ will only increase when $i$ deviates from the accepting policy. Hence, in the remainder of the proof we focus on exchanges involving $i$.
    
    \item Using the structure of \retrospect{} and the increasing competition order in \algname{}, we prove a key lemma: for any $\beta' > 0$, to maximize the number of exchanges with agents $j$ such that $\beta_{ij} \leq \beta'$,  agent~$i$ must follow the accepting policy.

    \item We show that the accepting policy maximizes the utility gain from exchanges involving $i$---both in terms of the number of exchanges and the utility gain from each exchange.
\end{enumerate}

        \begin{proof}
        \textbf{Step 1.\;}
        We will begin by decomposing $i$'s utility.
        For brevity, let $\datafinal(\strati) = \datafinal(M, \agsubset, $ $(\strati,\strataccpartmi))$ denote the final allocation when agent $i$ follows $\strati$ and the others follow the accepting policy.
        Let $\exchangesinvolvi(\strati), \exchangesnotinvolvi(\strati)$ respectively
        denote the exchanges involving $i$ and those not involving $i$ when $i$ follows $\strati$ and the others follow the accepting policy.
        Observing that every exchange $((i,r),(j,s)) \in\exchangesinvolvi(\strati)$ increases her utility by $1-\betaij$ and every exchange $((j,r),(k,s)) \in\exchangesnotinvolvi(\strati)$ decreases her utility by $\betaij+\betaik$, 
        we can decompose $i$'s utility as follows:
        \begin{align*}
        \utili(\datafinal(\strati))
             &= \onev^\top \datafinali(\strati) - \sum_{j\neq i} \betaij \onev^\top \datafinalj(\strati),                 \numberthis
                \label{eqn:utildecomp}
                \\
            &= \underbrace{\onev^\top \initdatai - \sum_{j\neq i} \betaij \onev^\top \initdataj}_{\utiliI}
                + \underbrace{\sum_{((i,r),(j,s))\in\exchangesinvolvi(\strati)} (1-\betaij)}_{\utiliE(\strati)}
                \,-\, \underbrace{\sum_{((j,r),(k,s))\in\exchangesnotinvolvi(\strati)} (\betaij+\betaik)}_{\utiliL(\strati)}.
        \end{align*}
        We need to show $\utili(\datafinal(\stratacci)) \geq \utili(\datafinal(\strati))$ for any other strategy $\strati$.
        In~\eqref{eqn:utildecomp}, $\utiliI$ is $i$'s initial utilities, which does not depend on $i$'s strategy.
        Next, $\utiliL(\strati)$ is the reduction in her utility due to exchanges not involving her. 
        Since other agents are following the accepting policy,
        when agent $i$ also follows the accepting policy, $\exchangesnotinvolvi(\stratacci)$ will be precisely the set of exchanges \emph{proposed} by \algname{} in the first round (recall that \algname{} terminates in the first round under $\strataccpart = \{\stratacci\}_{i\in\partagents}$).
        If $i$ follows any other policy, then $\exchangesnotinvolvi(\strati)$ will contain $\exchangesnotinvolvi(\stratacci)$ and any other exchanges added in subsequent rounds.
        Therefore, $\exchangesnotinvolvi(\stratacci) \subset \exchangesnotinvolvi(\strati)$ and consequently $\utiliL(\stratacci) \leq \utiliL(\strati)$.
        Finally,
        $\utiliE(\strati)$, is the increase in her utility due to (accepted) exchanges involving her. It is sufficient to show that $\utiliE(\stratacci) \geq \utiliE(\strati)$ for every other strategy $\strati$.

        \parahead{Step 2}
        To do so, we will use the following lemma, which is the
        key technical ingredient of this proof.
        Its proof is deferred 
        to~\S\ref{sec:proofbetagood}.

        \begin{restatable*}{lemma}{betagood}
        \label{lem:betagood}
            An exchange $((i,r), (j,s))$ is said to be $(1-\beta')$--good if $\betaij \leq \beta'$, \ie agents $i, j$ gain at least utility
            $1-\beta'$ from this exchange.
            Let $\beta' > 0$ be given and
            suppose all other agents follow the accepting strategy. Then, agent $i$ gets the most number of $(1-\beta')$--good exchanges in $\exchangesinvolvi(\strati)$ if she also uses the accepting policy.
        \end{restatable*}

        \parahead{Step 3}
        Suppose $i$ gets $\numexchacc$ exchanges when following $\stratacci$ and $\numexch$ exchanges when following $\strati$.
        Let $\beta_{ij^\star_1},\beta_{ij^\star_2},\dots,\beta_{ij^\star_{\numexchacc}}$ be the competition factors between $i$ and who she has exchanged with, sorted in ascending order, when $i$ follows $\stratacci$.
        Let $\beta_{ij_1},\beta_{ij_2},\dots,\beta_{ij_{\numexch}}$ be the same when $i$ follows $\strati$.
        To show that $\utiliE(\stratacci) \geq \utiliE(\strati)$, we need to show 
        $\sum_{p = 1}^{\numexchacc} (1 - \beta_{ij^\star_p})
        \geq \sum_{p = 1}^{\numexch} (1 - \beta_{ij_p})$.
        
        By choosing $\beta' = 1$ in Lemma~\ref{lem:betagood},
        we have $\numexchacc \geq \numexch$.
        Therefore, it is sufficient to show
        $\beta_{ij^\star_p} \leq \beta_{ij_p}$
        for any $p \leq \numexch$.
        By way of contradiction, assume that there exists $p$ such that $\beta_{ij^\star_p} > \beta_{ij_p}$.
        Then, under $\stratacci$, there are at most $p-1$ exchanges that are
        $(1-\beta_{ij_p})$--good; and
        under $\strati$, there are
        at least $p$ exchanges which are $(1-\beta_{ij_p})$--good.
        This contradicts Lemma~\ref{lem:betagood}. 
        \end{proof}

    \subsubsection{Proof of Lemma~\ref{lem:betagood}}
    \label{sec:proofbetagood}

        Say that \algname{} proposes  $\numexchacc$ exchanges that are $(1-\beta')$--good involving $i$ in round $1$.
        If following the accepting policy $\stratacci$, then $i$ gets $\numexchacc$ exchanges that are
$(1-\beta')$--good.
        If $i$ were to get more than $\numexchacc$ exchanges that are $(1-\beta')$--good (via exchanges in subsequent rounds) by following some strategy $\strati$, it must be the case that $i$ gets strictly more exchanges from at least one other agent $j$ where $\beta_{ij} \leq \beta'$.

        Since $\rejprops = \emptyset$ in round $1$, \algname{} stops scheduling exchanges between $i, j$ either because $\dmdjfromi = \emptyset$ ($j$ needs no more goods from $i$) or $\dmdifromj = \emptyset$ ($i$ needs no more goods from $j$) after respective \retrospect{} calls (see line $\ref{no-exchange2}$ in Algorithm~\ref{alg:main}).
    First, consider the case $\dmdjfromi = \emptyset$.
    To handle this, we will use Claim~\ref{expanding-possession} below, which asserts that although \retrospect{} may revise proposed exchanges to accommodate for more goods, any agent $k \in \partagents$ who is scheduled to receive a good $s$ at some point in round $t$ (\ie $\interdatat_{k,s} = 1$) will still be scheduled to receive that good by the end of the round---possibly from a different agent.

        \begin{claim} \label{expanding-possession}
            Let $k\in\partagents$. Once agent $k$ is scheduled to receive good $s$ in round $t$, \ie $\interdatat_{k, s} = 1$, 
            then $\interdatat_{k, s} = 1$ is guaranteed in the proposed allocation of this round.
        \end{claim}

       \begin{proof}[Proof of Claim~\ref{expanding-possession}]
When $\interdatat_{k, s} = 1$, the only point where it is (temporarily) set to $0$ is on line~\ref{alter-x} of \retrospect{}, where $s$ is cleared for an exchange with some agent $j$. However, the subsequent exchange with $j$ immediately sets $\interdatat_{k, s} = 1$ again. Thus, once $k$ gains good $s$, she never loses it.
        \end{proof}
        
        Applying Claim~\ref{expanding-possession} on agent $j$, we find that if $\dmdjfromi=\emptyset$, then, at the end of this round, it must still be the case that $j$ needs no more goods from $i$ as per $\interdata_j$.
        Since $j$ accepts all exchanges from other agents, $i$ cannot make $\dmdjfromi \neq \emptyset$ in a subsequent round unless she rejects her own exchanges with $j$ to make openings in $j$'s goods. However, each such rejected proposal by $i$ opens exactly one slot for a possible future exchange with $j$.
        That is, if
        $i$ rejects $\numexchj$ exchanges with $j$, the best she could hope for is that \algname{} schedules $\numexchj$ different exchanges between her and $j$ in subsequent rounds. Hence, rejecting proposals with $j$ in the first round has no benefit to $i$.
        
        Next consider the case $\dmdifromj = \emptyset$, \ie $i$ needs no more goods from $j$ as per $\interdata_i$. At the moment \algname{} finishes scheduling exchanges between $i$ and $j$, each good $s$ where $\initdatajs = 1$ (what $j$ can offer from her initial set of goods) has status $\interdata_{i,s} = 1$ due to one of the following three reasons:
        \emph{(i)} it could be that $i$ already has this as an initial good ($\initdatais = 1$),
        \emph{(ii)} it could be offered by $j$ in their exchange, or \emph{(iii)} it could be offered by another agent $j'$ who exchanged with $i$ before $j$.
        
        There is nothing $i$ can do in case \emph{(i)}. By using a similar reasoning as above (when $\dmdjfromi=\emptyset$), $i$ cannot increase the number of exchanges with $j$ in case \emph{(ii)}.
        In case \emph{(iii)}, recall that \algname{} must have invoked \retrospect$(i, j)$ on behalf of $i$, which recursively invokes \retrospect$(i, j')$ with the other agent $j'$ who were scheduled to exchange with $i$ before $j$, and gets \retrofail{}, before concluding that no further exchanges are possible between $i$ and $j$. By Lemma~\ref{RT-correctness}, the failure of \retrospect{} indicates that $i$ cannot find an alternative set of previously scheduled exchanges to make room for additional exchanges with $j$.
        Note that the exchange between $i$ and $j'$ is also $(1-\beta')$--good because $j'$ exchanges with $i$ before $j$, implying that $\beta_{ij'} \leq \betaij \leq \beta'$. Therefore, for every $(1-\beta')$--good exchange $i$ wants to get after round $1$, she must give up a $(1-\beta')$--good exchange in round 1.
        \qed

\subsection{Individual rationality}
\label{sec:irproof}

We now show that \algname{} is IR. Each exchange involving an agent \( i \) can only increase her utility. However, by not participating, an agent may alter the exchanges proposed by \algname{} and block more competitive agents from receiving goods, potentially improving her utility. The following example illustrates this effect.

\begin{example}
Say we have four agents \( \{i, j, k, l\} \) such that \( \beta_{ij} < \beta_{ik} < \beta_{il} < \beta_{jk} < \beta_{jl} < \beta_{kl} \) and two goods with \( \initdata_i = \initdata_k = [1, 0] \) and \( \initdata_j = \initdata_l = [0, 1] \). When $i$ participates, \emph{\algname{}} proposes the exchanges \( P_1 = \{((i, 1), (j, 2)), ((k, 1), (l, 2))\} \), which—if accepted—results in all four agents receiving both goods. In contrast, if \( i \) does not participate,
\emph{\algname{}} proposes \( P_1 = \{((k, 1), (j, 2))\} \), and agent~\( l \), who is most competitive  with \( i \), receives nothing. This illustrates that an agent’s (non-)participation can affect the goods others receive, potentially blocking competitors from receiving goods.
\end{example}

The key challenge is to show that, despite such knock-on effects, participation remains beneficial. Indeed, in the example above, if~\( i \) participates, her utility is \( 2 - 2(\beta_{ij} + \beta_{ik} + \beta_{il}) \), whereas if she does not, it is \( 1 - 2(\beta_{ij} + \beta_{ik}) - \beta_{il} \). Since \( \beta_{il} < 1 \), participation yields higher utility.

Formally proving this is more involved, as each missed exchange may trigger a chain of knock-on effects, particularly when the number of agents and goods is large.
We proceed in three steps:
\begin{enumerate}[leftmargin=0.2in]
    \item Construct a graph \( G = (V, E \cup E') \), where the vertex set is \( V = (\partagents \cup \{i\}) \times [\numgoods] \). Edges in \( E \) represent exchanges when agent~\( i \) participates, and edges in \( E' \) represent exchanges when she does not. We also show that any agent's utility can be computed given the number of exchanges between each pair of agents.

    \item Define a structure in $E \cup E'$ called \emph{tracked path}, which is a connected component including any of agent $i$'s goods, and give several properties of such structures. Then we present a technical result stating that the number of \emph{untracked} exchanges (exchanges that are not a part of any tracked path) between each pair of agents is always the same in $E$ and $E'$.

    \item Track the change in \( i \)'s utility by summing utility changes along tracked paths. We will show that each path yields nonnegative utility gain when \( i \) participates, from which IR follows.
\end{enumerate}

\begin{proof} 
Let $i$ be a given agent and let $\partagents \subset \allagents\backslash\{i\}$ be the agents participating with the accepting policy.
We will show that by joining $\partagents$ with the accepting policy, $i$ is no worse off.
Recall that as all agents are following the accepting policy, \algname{} will terminate in one iteration.

\parahead{Step 1}
We begin by describing the construction of the graph \( G = (V, E \cup E') \). The vertex set is \( V = (\partagents \cup \{i\}) \times [m] \), where each vertex represents an agent-good pair, including agent~\( i \). Each edge belongs to either \( E \) or \( E' \), with some possibly appearing in both.

The edge set \( E \) is obtained by running \algname{} with agents \( \partagents \cup \{i\} \): for every exchange \( ((j, r), (k, s)) \) in this run, we add an edge between \( (j, s) \) and \( (k, r) \) in \( E \). Note that edges connect the goods \emph{acquired}, not the goods contributed. The set \( E' \) is constructed analogously by running \algname{} without agent~\( i \).

We have illustrated the construction in Figure~\ref{IR-proof-graph}. Importantly, no edge in \( E \) corresponding to an exchange involving \( i \) appears in \( E' \).
We will overload notations by referring to an edge and its corresponding exchange interchangeably, and referring to the instance where $i$ participates with the accepting policy (or does not participate) by $E$ (or $E'$), respectively.

\insertFigIR

We give an alternative computation of any agent's utility using the number of exchanges between each pair of agents. Specifically, let $U_i, U'_i$ respectively denote $i$'s utility in $E$ and in $E'$, and let $\pairnum_{jk}, \pairnumalt_{jk}$ respectively denote the number of exchanges between agent pair $(j, k)$ in $E$ and $E'$.

We will first decompose $U_i$ as follows.
Recall that every exchange $((i, r), (j, s))$ involving $i$ improves her utility by $(1-\betaij)$, and that every exchange $((j, r), (k,s))$ not involving $i$ decreases her utility by $-\betaij-\betaik$.
Therefore:
\begingroup
\allowdisplaybreaks
\begin{align*}
    U_i &= \onev^\top \datafinali(M, \agsubset \cup \{i\}, (\strataccpart, \stratacci)) -
        \sum_{j \in \allagents \setminus \{i\}} \betaij \onev^\top \datafinalj(M, \agsubset \cup \{i\}, (\strataccpart, \stratacci)) \\
        &= \onev^\top \initdatai -
        \sum_{j \in \allagents \setminus \{i\}} \betaij \onev^\top \initdataj
        \;+\; \sum_{j, k \in \partagents} \pairnum_{jk} (-\betaij - \betaik)
        \;+\; \sum_{j \in \partagents} \pairnum_{ij} (1 - \betaij)
\end{align*}
\endgroup

Similarly, we can write $U'_i$ as follows:
\begingroup
\allowdisplaybreaks
\begin{align*}
    U'_i &= \onev^\top \datafinali(M, \agsubset, \strataccpart) -
        \sum_{j \in \allagents \setminus \{i\}} \betaij \onev^\top \datafinalj(M, \agsubset, \strataccpart) \\
        &= \onev^\top \initdatai -
        \sum_{j \in \allagents \setminus \{i\}} \betaij \onev^\top \initdataj
        \;+\;
        \sum_{j, k \in \partagents} \pairnumalt_{jk} (- \betaij - \betaik)
\end{align*}
\endgroup

We therefore have the following expression for $U_i - U'_i$, the increase in utility when $i$ participates, which depends only on the number of exchanges between each pair of agents:
\begingroup
\allowdisplaybreaks
    \begin{align*}
    \numberthis\label{eqn:UiminusUip}
    U_i - U'_i
    &= 
    \sum_{j \in \partagents} \pairnum_{ij} (1 - \betaij) \;+\; 
    \sum_{j, k \in \partagents} (\pairnum_{jk} - \pairnumalt_{jk}) (- \betaij - \betaik)
\end{align*}
\endgroup
In the remainder of this proof we will show that the RHS of~\eqref{eqn:UiminusUip} is nonnegative, which implies IR.

\parahead{Step 2}
We use the graph constructed above to track differences in number of exchanges. To this end, we extract a special structure called \emph{tracked path} from the edge set $E \cup E'$. A tracked path is a connected component in the graph $G = (V, E\cup E')$ including a node corresponding to one of $i$'s goods (shortly, in Lemma~\ref{lem:tracked-path-property}, we will show that these connected components are, in fact, paths). We will use these paths to track the differences in number of exchanges arising from $i$'s (non-)participation. Figure~\ref{IR-proof-graph} illustrates several tracked paths. We refer to exchanges on a tracked path as \emph{tracked exchanges} and correspondingly, exchanges not on any tracked path as \emph{untracked}.

\begin{remark}[Interpreting a tracked path]
Each tracked path represents a sequence of knock-on effects in the exchanges proposed by \emph{\algname{}} due to agent~\( i \)'s (non-)participation. For example, in Figure~\ref{fig:IRsub1}, the tracked path \(\{(i, 2), (j, 1)\},\, \{(j, 1), (k, 2)\},\, \{(k, 2), (l, 1)\}\) can be interpreted as follows. The edges \(\{(i, 2), (j, 1)\}\) and \(\{(k, 2), (l, 1)\}\) belong to \( E \) and represent the exchanges \(((i, 1), (j, 2))\) and \(((k, 1), (l, 2))\) that would have occurred if \( i \) had participated. Without \( i \), the first exchange would not be scheduled, and instead \emph{\algname{}} would propose \(((j, 2), (k, 1))\), represented by the edge \(\{(j, 1), (k, 2)\} \in E'\). As a result, the final exchange \(((k, 1), (l, 2))\) would also not occur, since agent~\( k \) would have already received good \( 2 \) from \( j \).

Also note that, while the exchanges differ, the actual goods acquired by the agents only differ at the start and end of the tracked paths.
For instance, above, agents $j$ and $k$ receive goods $1$ and $2$ respectively, regardless of $i$'s participation.
However, they receive them through different exchanges depending on $i$'s (non-)participation.
\end{remark}

The following lemma characterizes the structure of tracked paths. This result is quite intuitive from the observation that each edge in $E$ (or $E'$) are disjoint from other edges in $E$ (or $E'$) because \algname{} never schedules any agent to receive the same good more than once. The proof is in~\S\ref{sec:proof-tracked-path-property}.

\begin{lemma}\label{lem:tracked-path-property}
Tracked paths satisfy the following properties:
\begin{enumerate}
    \item Each tracked path is indeed a path (even though it is defined as a connected component).
    \item Each tracked path starts with an exchange in $E$ involving agent $i$ and alternates between exchanges in $E'$ and $E$ afterwards.
\end{enumerate}
\end{lemma}

Next, we formalize the intuition that tracked paths, which capture the knock-on effects due to $i$'s participation, should account for all differences in the number of exchanges in $E$ and $E'$. The following lemma is the most technical part of this proof, which requires a deep understanding of the functionality of \retrospect{}. Its proof is in \S\ref{sec:proof-IR-tracked-path}.

\begin{restatable*}{lemma}{irtrackedpath}
\label{lem:IR-tracked-path}
    The number of untracked exchanges between each agent pair is always the same in $E$ and $E'$.
    That is, for all $j, k \in \partagents \cup \{i\}$, $\pairnum_{jk} - \pairnumtr_{jk} = \pairnumalt_{jk} - \pairnumtralt_{jk}$.
\end{restatable*}
    
\parahead{Step 3}
We now examine the effect of agent~\( i \)'s (non-)participation on her utility by tracking changes along tracked paths. By definition, each tracked path is a connected component, which implies that they are not connected to other tracked paths or other exchanges in $E \cup E'$. This ensures that we can compute the utility changes along each tracked path separately. From Lemma~\ref{lem:tracked-path-property}, each tracked path originates from an edge in $E$ involving \( i \) and alternates between edges in $E$ and edges in $E'$. There are three distinct cases based on where the path terminates, each affecting \( i \)'s utility differently. We classify them as follows:
(1) even-length paths that end in \( E' \), denoted by \( \Tcal_1 \);  
(2) odd-length paths that end in some agent \( j \neq i \) in $E$, denoted by \( \Tcal_2 \);  
(3) odd-length paths that end in \( i \) herself, denoted by \( \Tcal_3 \).  
All three types of tracked paths are illustrated in Figure~\ref{IR-proof-graph}.

From Lemma~\ref{lem:IR-tracked-path}, we can derive $\pairnum_{jk} - \pairnumalt_{jk} = \pairnumtr_{jk} - \pairnumtralt_{jk}$ for all pairs $(j, k)$. Also note that the lemma is trivially true for any pair involving $i$, because $\pairnumtr_{ij} = \pairnum_{ij}$ and $\pairnumtralt_{ij} = \pairnumalt_{ij} = 0$ for all $j$. Combining these with~\eqref{eqn:UiminusUip}, we have $U_i - U'_i = \sum_{j \in \partagents} \pairnumtr_{ij} (1 - \betaij) \;+\; \sum_{j, k \in \partagents} (\pairnumtr_{jk} - \pairnumtralt_{jk}) (- \betaij - \betaik)$.
As this expression only depends on the number of tracked exchanges $\{\pairnumtr_{jk}, \pairnumtralt_{jk}\}_{j,k \in \partagents \cup \{i\}}$,
we can compute $U_i - U'_i$ by summing only tracked exchanges:
\begin{align*}
    \numberthis\label{eqn:UiminusUiptrackedpaths}
    U_i - U'_i =
        \sum_{p\in\Tcal_1}\sum_{e\in p} \Delta(e) +
        \sum_{p\in\Tcal_2}\sum_{e\in p} \Delta(e) +
        \sum_{p\in\Tcal_3}\sum_{e\in p} \Delta(e).
\end{align*}
where $\Delta(e) = 1-\betaij$ for $e \in E$ involving $i$,
$\Delta(e) = -\betaij-\betaik$ for $e \in E$ does not involve $i$,
and $\Delta(e) = \betaij+\betaik$ for $e \in E'$ denotes the effect of a single tracked exchange $e$ on $U_i - U'_i$.

We will separately handle $\Tcal_1, \Tcal_2, \Tcal_3$ to show that for every tracked
path $\sum_{e\in p}\Delta(e)$ is non-negative.

First consider paths in $\Tcal_1$.
The path has even length and ends with an edge in \( E' \). If \( i \) participates, all exchanges in \( E \) occur, and the final agent on the path, say agent~\( j \), receives one fewer good. If \( i \) does not participate, all exchanges in \( E' \) occur, and \( i \) receives one fewer good. Comparing the gain from participating and the loss from not participating, agent~\( i \) gains \(\sum_{e\in p}\Delta(e) = 1 + \beta_{ij} \geq 0 \) by participating.
This can be seen explicitly via the following calculations. Recall the definition of $\Delta(e)$, and note that as the path has even length $\ell$, the last exchange is in $E'$.
Letting the agents involved in this path be $i, k_1, k_2, k_3 \dots, k_{\ell-1}, j$, we have,
\begin{align*}
    \sum_{e\in p} \Delta(e)
        = (1-\beta_{ik_1}) + (\beta_{ik_1} + \beta_{ik_2}) + (-\beta_{ik_2} - \beta_{ik_3}) + \dots +
        (\beta_{i k_{\ell-1}} + \betaij)
        = 1+\betaij \geq 0.
\end{align*}
Next consider paths in $\Tcal_2$.
The path has odd length and ends in some agent \( j \neq i \), with the final edge in \( E \). If \( i \) participates, all exchanges in \( E \) occur and all agents along the path receive their corresponding goods. If \( i \) does not participate, all exchanges in \( E' \) occur, and both \( i \) and \( j \) receive one fewer good. Since \( \beta_{ij} < 1 \), the utility change for \( i \) is \( 1 - \beta_{ij} \geq 0 \).
This can be seen explicitly via the following calculations.
Let the agents involved in this path be $i, k_1, k_2, k_3 \dots, k_{\ell-1}, j$ and recall that the last exchange is in $E$ as the path has odd length $\ell$. We have,
\begin{align*}
    \sum_{e\in p} \Delta(e)
        = (1-\beta_{ik_1}) + (\beta_{ik_1} + \beta_{ik_2}) + (-\beta_{ik_2} - \beta_{ik_3}) + \dots +
        (-\beta_{i k_{\ell-1}} - \betaij)
        = 1-\betaij \geq 0.
\end{align*}
Finally, consider paths in $\Tcal_3$.
The path has odd length and ends in \( i \).
In this case, the goods received by the others do not change.
If $i$ does not participate, she misses out on two goods, while if she does participate, she receives both. Hence, her utility gain from participation is \( 2 \).
This can be seen explicitly via the following calculations.
Let the agents involved in this path be $i, k_1, k_2, k_3 \dots, k_{\ell-1}, i$, and recall that the last exchange, involving $i$, is in $E$. We have,
\begin{align*}
    \sum_{e\in p} \Delta(e)
        = (1-\beta_{ik_1}) + (\beta_{ik_1} + \beta_{ik_2}) + (-\beta_{ik_2} - \beta_{ik_3}) + \dots +
        (1 -\beta_{i k_{\ell-1}})
        = 2.
\end{align*}
Combining these results with~\eqref{eqn:UiminusUiptrackedpaths}, we obtain $U_i-U'_i\geq 0$, which completes the proof.
\end{proof}

\subsubsection{Proof of Lemma \ref{lem:tracked-path-property}}
\label{sec:proof-tracked-path-property}

Recall the observation that each edge in $E$ (or $E'$) are disjoint from other edges in $E$ (or $E'$). This is simply because \algname{} never schedules any agent to receive the same good more than once, so no two exchanges will share an endpoint under one allocation. Intuitively, connectedness in $E \cup E'$ only happens when we alternate between edges in $E$ and edges in $E'$.

Given that a tracked path must contain a node in agent $i$'s goods, it must also contain an exchange involving $i$. Suppose this exchange is given by edge $\{(i, s), (j, r)\}$ in $E$. Note that no other edge in $E$ will use node $(i, s)$ by the observation above, and no other edge in $E'$ will use node $(i, s)$ because $i$ does not participate. Thus, the connect component can only extend itself on the other node $(j, r)$.

However, $(j, r)$ cannot be used by other exchanges in $E$ for the same reason, so the only possibility is that some exchange $\{(j, r), (k, p)\}$ in $E'$ extends the connected component. Note that $(k, p)$ must be different from $(i, s)$ because $i$ does not participate in $E'$. It is also easy to see that no other edges in $E'$ can use $(j, r)$, so the only opening is at $(k, p)$.
Consider one more step $\{(k, p), (l, q)\}$ extending the connected component. This exchange must be in $E$ if it exists, and must be the only exchange in $E$ that uses $(k, p)$ by similar logic as above. Also, $(l, q)$ must be a different node from $(i, s)$ or $(j, r)$, so $(l, q)$ is the only opening for the entire connected component.

We can follow similar logic and keep extending the connected component. However, if the previous step was in $E$ (or $E'$), we can find one and exactly one exchange in $E'$ (or $E$) that connects to it. This new exchange, whether in $E$ or $E'$, will not revisit any previous nodes because the starting node at $i$ is used in $E$ and unavailable in $E'$, and any intermediate node is used in both $E$ and $E'$. Thus, we may conclude that this connected component is a path that starts with an exchange involving $i$ and alternates between exchanges in $E$ and $E'$ afterwards. \qed

We make a note that a tracked path might revisit the same agent but not the same good. It is also possible for a tracked path to end itself back at some good of $i$, as discussed in Step 3 of the main IR proof. These do not affect the utility computation along any tracked path.

\subsubsection{Proof of Lemma \ref{lem:IR-tracked-path}}
\label{sec:proof-IR-tracked-path}

We first give an intuitive interpretation about the functionality of \retrospect{} based on Lemma~\ref{RT-correctness}. By statement (iii) of Lemma~\ref{RT-correctness}, \retrospect{} does not alter the number of exchanges scheduled between any pair of agents who exchanged previously. Statement (iv) states that \retrospect{} \emph{always} finds an alternative set of previous exchanges that enables one more exchange between the current pair if any such alternative exists, and only fails if no such alternative exists. In other words, \retrospect{} may adjust \emph{any} relevant previous exchanges, as long as it does not change the number of exchanges between previous pairs of agents, in order to maximize the number of exchanges between the current pair (note that Step 2 of the NIC proof in \S\ref{sec:NIC} uses a similar idea).

In this proof, we will align the execution of \algname{} in $E$ (\ie when $i$ participates) and $E'$ (\ie when $i$ is absent) by stages, where each \emph{stage} corresponds to exchanges between a pair of agents. As \algname{} proceeds in order of competition, all stages in $E$ except those involving $i$ will have a counterpart in $E'$.

We will compare the allocations $\interdatat$ in $E$ and $E'$ at the end of each stage. For this, it is convenient to naturally extend the definition of tracked paths to allocations at the end of each stage, where we look at edge sets $E \cup E'$ induced by $\interdatat$ in the two instances and find connected components involving agent $i$'s goods. It is straightforward to verify that the properties of tracked paths stated in Lemma \ref{lem:tracked-path-property} still hold.

The number of exchanges between a pair of agents becomes fixed once \algname{} finishes scheduling exchanges between them and moves on to the next pair, because \retrospect{} does not change the number of exchanges scheduled between previous agent pairs. Thus, comparing the number of exchanges between a pair of agents in the final allocation in $E$ and $E'$ is the same as comparing the number of exchanges when they finish exchanging.

Next, given an allocation $x$, we say that a set of exchanges $S$ is a \emph{feasible arrangement} from $x$, if we can execute a set of exchanges from $\feaspropsp(x)$ in~\eqref{eqn:feasibleexchanges} so that no agent receives the same good more than once.
Precisely, $S\subset \feaspropsp(x)$ and for all $((i, r), (j, s)), ((i', r'), (j', s')) \in S$
we have $(i, s) \neq (i', s') \neq (j, r) \neq (j', r')$.
Furthermore, we say $\mathcal{S}(x)$ is the set of all possible feasible arrangements under allocation $x$. Note that the number of total, tracked, and untracked exchanges between a pair of agents $(i, j)$ ($\pairnum_{ij}, \pairnumtr_{ij}, \pairnum_{ij} - \pairnumtr_{ij}$ respectively) can be defined accordingly over any $S \in \mathcal{S}$.
Based on the above, we observe the following:
\begin{observation}[Optimality of \retrospect{}] \label{optimality-of-RT}
    Consider the feasible arrangement $S^\star$ given by \algname{} at the stage when $(i, j)$ finishes exchanging, and any other $S \in \mathcal{S}$ satisfying that $\pairnum_{i'j'} = \pairnum^\star_{i'j'}$ for all previous pairs $(i', j')$. Then we have $\pairnum^\star_{ij} \geq \pairnum_{ij}$ for all such $S \in \mathcal{S}$.
\end{observation}
Intuitively, given that we do not change the number of exchanges between any previous pair, the maximum possible number of exchanges between $i$ and $j$ (whether in $E$ or $E'$) are scheduled when they just finish exchanging with each other in \algname{}.
This is because \retrospect{} has checked all previous exchanges, and future exchanges only take away available good slots from $i$ or $j$, which cannot make new exchanges possible between them.

Now we look at the untracked exchanges. By definition, they are disjoint from any tracked path. Thus, we may define an \emph{untracked region} for any $E$ and $E'$, which is the set of nodes that are not involved in any tracked exchanges. Clearly all untracked exchanges in $E \cup E'$ only use nodes in the untracked region. Given a feasible arrangement $S$, let $\suntrack$ be the set of untracked exchanges in $S$. We observe the following about $\suntrack$:
\begin{observation}[Flexibility of Untracked Exchanges] \label{untracked-flexibility}
     Consider two feasible arrangements $S$ and $S'$. Then $S \setminus \suntrack \cup \suntrack'$ (respectively $S' \setminus \suntrack' \cup \suntrack$) is still a feasible arrangement.
\end{observation} 
Intuitively, replacing the set of untracked exchanges in one arrangement by those in another does not break feasibility.
This is because all untracked exchanges are disjoint from tracked paths.

With Observation \ref{optimality-of-RT} and \ref{untracked-flexibility}, we proceed to establish Lemma \ref{lem:IR-tracked-path} by induction over the stages. Induction is needed here because the number of untracked exchanges for the current pair depends on the number of untracked exchanges for previous pairs. If any previous pair has different numbers of untracked exchanges, that difference will cascade to the current pair and leave no guarantees about the number of untracked exchanges for the current pair.

Lemma \ref{lem:IR-tracked-path} is trivially true when no pairs of agents have ever exchanged. Then, we consider a currently exchanging pair given that all previous pairs have the same number of untracked exchanges scheduled between them. We complete the induction with two contradictions.

First, we will argue that any \retrospect{} for the current pair does not break the pattern that all previous pairs have the same number of untracked exchanges. By way of contradiction, suppose it does. Then there exists an earliest pair $(i', j')$ with $k$ untracked exchanges in $E$ and $k'$ untracked exchanges in $E'$ (WLOG $k > k'$) such that all pairs who exchanged before $(i', j')$ still have the same number of untracked exchanges. Say the arrangements in $E$ and $E'$ are $S$ and $S'$ respectively. Then by Observation \ref{untracked-flexibility}, we may replace $\suntrack'$ by $\suntrack$ in $E'$ to enable $k > k'$ untracked exchanges between $(i', j')$. This will increase the total number of exchanges between $(i', j')$ by $k - k' > 0$ as well because the number of tracked exchanges between them remains the same. Given that the replacement above does not change the number of untracked or total exchanges between any pairs before $(i', j')$, the sub-optimality of $S'$ for the pair $(i', j')$ contradicts Observation \ref{optimality-of-RT}.

Second, the current pair $(i, j)$ must also have the same number of untracked exchanges. By way of contradiction suppose they have different numbers of untracked exchanges. Then say there are $k$ untracked exchanges between them in $E$ and $k'$ untracked exchanges between them in $E'$ (WLOG $k > k'$). Say the arrangements in $E$ and $E'$ are $S$ and $S'$ respectively. Then by Observation \ref{untracked-flexibility}, we may replace $\suntrack'$ by $\suntrack$ in $E'$ to enable $k > k'$ untracked exchanges between $(i, j)$. Since we have already established that all previous pairs have the same number of untracked exchanges by the first contradiction, the sub-optimality of $S'$ for the pair $(i, j)$ contradicts Observation \ref{optimality-of-RT}. Thus, we proved Lemma \ref{lem:IR-tracked-path} by induction. \qed

\subsection{Stability}
\label{sec:nief}

Stability follows straightforwardly from the design of \algname{}.
First recall that although \retrospect{} may alter the agents involved in specific exchanges, Claim~\ref{expanding-possession} guarantees that once an agent \( i \) is scheduled to receive a good from some agent \( j \) in a given round, she is guaranteed to receive that good, possibly from a different agent \( k \) in that round.

\begin{proof}
Consider any pair of agents \( (i, j) \). We show that when \algname{} terminates, either there are no remaining mutually beneficial exchanges between \( i \) and \( j \), or any such exchanges have already been rejected. Applying this argument to all agent pairs yields stability.

Consider the final round of \algname{}. It exits $(i,j)$ and proceeds to the next agent pair at line~\ref{no-exchange1} or line~\ref{no-exchange2}.
If it exits $(i,j)$ at line~\ref{no-exchange1}, then both \( \dmdifromj \neq \emptyset \) and \( \dmdjfromi \neq \emptyset \), but every potential exchange \( ((i, r), (j, s)) \in \rejprops \) for all \( (r, s) \in \dmdjfromi \times \dmdifromj \) (see line~\ref{find-exchange}). That is, all remaining mutually beneficial exchanges between \( i \) and \( j \) have already been rejected by one of the agents.

If \algname{} exits $(i,j)$ at line~\ref{no-exchange2} in the final round, then \( \dmdifromj = \emptyset \) or \( \dmdjfromi = \emptyset \); that is, at least one agent has nothing to gain from the other at the time \algname{} exits to the next pair. By Claim~\ref{expanding-possession}, we know that by the end of the round, these goods will still be received by the agents. Since \algname{} only terminates once all proposals have been accepted, there are no remaining mutually beneficial exchanges between \( i \) and \( j \).
\end{proof}

\subsection{Proof of the second statement of Theorem~\ref{thm:main}}

\begin{proof}
If all agents follow the accepting policy, then no proposals are rejected, and the protocol terminates after the first round.  
As in the argument in~\S\ref{sec:nief}, there are no remaining mutually beneficial exchanges between any pair of agents, and hence the resulting allocation is stable.
\end{proof}

\section{Proofs of Impossibility Results}
\label{sec:impossibility}

In this section, we provide proofs of
the impossibility results in~\S\ref{sec:impossibilityresults}.

\subsection{Proof of Theorem~\ref{thm:nodsic}}
\label{sec:nodsic}

First, we will prove Theorem~\ref{thm:nodsic}, which states that no protocol can be simultaneously stable and DSIC on all problem instances.
We proceed in the following three steps:
    \begin{enumerate}[leftmargin=0.2in]
        \item 
        We will first construct a problem instance $(\initdata, \beta)$.
        Any stable protocol on this instance can be represented by a tree, where the internal nodes represent exchanges proposed by this protocol on each round, the edges represent if each proposed exchange was accepted/rejected, and the leaf nodes represent final allocations obtained.

        \item 
        We then rule out protocols which have the following structure: for some agent, the protocol schedules a higher-competition exchange before a conflicting, lower-competition one. Crucially, the lower-competition exchange is only feasible if the higher-competition exchange is rejected. We refer to this pattern as an \emph{inversion pair}, and we show that for any protocol exhibiting such a structure, one can systematically construct a strategy profile that serves as a counterexample to DSIC. This rules out almost all stable protocols from being DSIC.

        \item
        We argue that the only protocol that avoids creating inversion pairs is the Sequential Ascending Competition (SAC) protocol, which schedules exchanges one at a time in increasing order of competition. However, even SAC is not DSIC: we can construct a counterexample strategy profile under which an agent benefits by deviating from the accepting policy. 

    \end{enumerate}

Before we proceed, let us recall the definition of DSIC.
A protocol $M$ is DSIC if, for all participating agents $i\in\partagents$, all strategies $\strati$ of agent $i$, and all strategy profiles $\stratpartmi$ of the other participating agents, following the accepting strategy $\stratacci$ is no worse for agent $i$, \ie
$\utili(M, \partagents, (\stratpartmi, \stratacci)) \geq \utili(M, \partagents, (\stratpartmi, \strati))$.
Hence, to show that no protocol is DSIC,
we need to show that for any protocol $M$, there exists some agent $i$, some strategy $\strati\neq\stratacci$ for agent $i$,
and some strategy profile $\stratpartmi$ for the others such that $\utili(M, \partagents, (\stratpartmi, \stratacci)) < \utili(M, \partagents, (\stratpartmi, \strati))$.

    \begin{proof} 
    \parahead{Step 1}
    Let us consider the following instance with 4 agents  $i, j, k, l$ and 2 goods. Agents $i, k$ have only good $1$. Agents $j, l$ have only good $2$.
    Assume that $\beta_{ij} < \beta_{jk} < \beta_{kl} < \beta_{il}$, and let $\beta_{ik}, \beta_{jl}$ be arbitrary (since exchanges are not possible between $i,k$ and $j,l$, their values will not affect this proof).
    In this example, we can refer to exchanges by only agent names without any ambiguity. For example, the exchange 
    where $i$ gives good $1$ to $j$ and $j$ gives good $2$ to $i$ can be concisely referred to as $ij$,
    instead of the more verbose $((i, 1), (j, 2))$.

    We will represent a protocol’s execution on a given instance via the following tree. The root node encodes the exchanges proposed in round 1. Each edge from the root represents which exchanges were accepted or rejected by the agents; so if $N$ exchanges are proposed, the root has $2^N$ child nodes. Each child node then represents the exchanges proposed by the protocol in the second round, conditioned on those agent decisions. This process continues until the protocol terminates, with each leaf node corresponding to a final allocation.
    Since the protocol cannot propose a previously proposed exchange (whether it is accepted or rejected), this tree has finite height.
    Figure~\ref{fig:tree-sub1} illustrates this tree structure for \algname{} on the instance described above.

    Before proceeding, let us rule out some poorly designed protocols.
    Our formalism in~\S\ref{sec:setup} allows a protocol to schedule two exchanges with the same agent in the same round (\eg $ij$ and $jk$ together), or schedule an exchange where an agent is receiving a good she has previously received (\eg $jk$ is scheduled in a future round after $ij$ was accepted),
    However, such protocols are trivially not DSIC.
    In the first case, agent $j$ will only accept $ij$ if both $i$ and $k$ are following the accepting policy,
    and in the second case, $j$ will always reject the second exchange.
    Therefore, going forward, we will not consider such protocols.

    \insertFigTree
  
    \parahead{Step 2}
    To build some intuition for this step,
    we will use the tree for \algname{} in Figure~\ref{fig:tree-sub1}.
    Observe that agent $k$ prefers to exchange with $j$ than $l$, but \algname{} proposes $kl$ first. If $k$ knows that any exchanges involving $i$ will not happen (for instance, say $i$ always rejects all proposals), then she should reject $kl$ and then accept $jk$ in the subsequent round.
    This leads to a counter example strategy profile: $i$ rejects all, $j$ and $l$ accept all. Under this profile, if $k$ uses the accepting policy, the only accepted exchange is $kl$. Instead, if $k$ rejects $kl$ and accepts $jk$, the only accepted exchange is $jk$. Since $\beta_{jk} < \beta_{kl}$, $k$ gets strictly higher utility by deviating from the accepting policy.

    Here, the protocol proposes $kl$ before $jk$, but $k$ prefers $jk$ due to lower competition. This \emph{inversion} of $k$'s preferences leads to the above strategy profile where $k$ can benefit by deviating from the accepting policy.
    We can generalize this intuition to any protocol with a similar structure in its corresponding tree. We call this structure an \emph{inversion pair} (referring to the pair of exchanges), defined formally as follows: for any agents $i', j', k' \in \{i, j, k, l\}$, we say $(i', j', k')$ is an \emph{inversion pair} for protocol $M$ if $\beta_{i'j'} < \beta_{i'k'}$ but $i'k'$ is an ancestor of $i'j'$ in the corresponding tree.
    In other words,
    $M$ proposes exchange $i'k'$ in a round before  $i'j'$, so  $i'j'$ is only possible if $i'k'$ is rejected.
    The following lemma states that if there exists an inversion pair in the tree representing a protocol, then the protocol is not DSIC.

    \begin{lemma}\label{inversion-pair}
        If protocol $M$ creates an inversion pair $(i', j', k')$, then there exists a strategy profile
        for the other agents such that deviating from the accepting policy is better for agent $i'$.
    \end{lemma}
    
    \begin{proof}[Proof of Lemma~\ref{inversion-pair}]
        Consider the following strategy profile: $j'$ and $k'$ accept all their proposals,
        while the remaining agent $l'$ who is not a part of the inversion pair rejects all proposals involving her.
        Since $l'$ rejects all, we know that any exchanges involving $l'$ will not happen. By definition of an inversion pair, we know that $M$ proposes exchange $i'k'$ before exchange $i'j'$. If $i'$ follows the accepting policy, the only accepted exchange is $i'k'$. Instead, if $i'$ uses a strategy that rejects $i'k'$ and accepts $i'j'$, the only accepted exchange is $i'j'$. Since $\beta_{i'j'} < \beta_{i'k'}$, $i'$ gets strictly higher utility by following the latter policy instead of the accepting policy.
    \end{proof}

    \parahead{Step 3}
    To complete the proof, it is sufficient to show that any stable protocol without an inversion pair is also not DSIC.
    We argue that there is only one possible tree which does not create an inversion pair.
    For this, note that there are 4 possible exchanges
    $ij, jk, kl, il$ on this instance.
    First note that every exchange must appear at least once to ensure stability, because in the case that all agents reject all exchanges, every available exchange must be proposed before the protocol realizes that the agents will reject it.
    Second, 
    in any path from the root to a leaf, an exchange can appear at most once: a protocol cannot propose an exchange more than once because the exchange is either accepted, in which case it has happened, or rejected, in which case it cannot be proposed again.
    Moreover, an exchange $i'j'$ should not be a descendant of $i'k'$ if $i'k'$ has already been accepted; as explained in Step 1, such protocols are trivially not DSIC.

    Now, let us try to arrange the 4 proposals while following the above rules and avoiding inversion pairs. In round 1 we must propose only $ij$, because proposing anything else will put $ij$ in an inverted position and proposing $kl$ with $ij$ will put $kl$ in an inverted position against $jk$.
    If $ij$ is accepted on round 1, then we have to propose $kl$ in the next round to ensure stability.
    If $ij$ is rejected, we must propose only $jk$ in round 2 because proposing anything else will put $jk$ in an inverted position and proposing $il$ with $jk$ will put $il$ in an inverted position against $kl$.
    Continuing in this fashion, we obtain
    the tree shown in Figure~\ref{fig:tree-sub2}.

    Now, let us consider the following strategy profile
    in this protocol.
    Suppose $i$ accepts $ij$ and rejects $il$, and $k$ and $l$ accepts all.
    If $j$ follows the accepting policy, exchanges $ij$ and $kl$ will happen, resulting in all agents possessing all goods. If $j$ rejects $ij$ and accepts $jk$, then $jk$ is the only exchange that happens because $i$ will reject $il$. We see that $j$ gets strictly higher utility by deviating from the accepting policy since two other agents do not receive goods.
    Thus, this protocol is not DSIC.
    \end{proof}

    Note that by definition, we allow the protocol to schedule differently based on \emph{who} rejects an exchange, instead of merely whether an exchange is accepted or rejected. For example, if the protocol knows that exchange $ij$ is rejected by $i$ but accepted by $j$, it might punish $i$ by not scheduling any other exchanges involving her.

    However, having this flexibility does not break the proof above. First, any protocol with an inversion pair is still not DSIC because the systematic counter example we constructed in Step 2 works out regardlessly. Second, the only stable protocol not creating an inversion pair is still the SAC protocol we discussed in Step 3. Thus, there is still no stable DSIC protocol on this instance.

\subsection{Proof of Theorem~\ref{thm:truthfulreporting}}
\label{sec:truthfulreporting}

Next, we prove Theorem~\ref{thm:truthfulreporting}, which
states that truthfully eliciting initial allocations is impossible in this problem setting.
The key intuition is that by pretending to have fewer goods than she does, an agent can cause a protocol to schedule more exchanges involving her,
and hence reduce the number of exchanges for other agents. 
While our proof below clearly relies on the specifics of our setting, we believe the underlying intuition carries over to other sharing models with freely replicable goods and competitive agents.

\begin{proof}
    Say that all three agents are participating so that  $\allagents = \partagents = \{i,j,k\}$.
    Suppose, by way of contradiction, that $\mech$ is a stable exchange protocol 
    in which for all \emph{true} initial allocations $\inity$, each agent truthfully reporting $\initdatai = \inityi$ and accepting all proposals is a Nash equilibrium. We will show the contradiction by
    studying the protocol's behavior under three different values for the \emph{reported} data allocations:
    \emph{Case (a): } $\initdatai = \initdataj = [1, 0]$ and $\initdatak = [0, 1]$;
    \emph{Case (b): } $\initdatai = [1, 1]$, $\initdataj = [1, 0]$ and $\initdatak = [0, 1]$;
    \emph{Case (c): } $\initdatai = [1, 0]$, $\initdataj = [1, 1]$ and $\initdatak = [0, 1]$.

    \subparahead{Case (a)}
    First consider $\initdatai = \initdataj = [1, 0]$ and $\initdatak = [0, 1]$.
    In this case, there are only two possible exchanges that $\mech$ can schedule, \ie
    $\feaspropsp(\initdata, \partagents) = \{((i, 1), (k, 2)), \, ((j, 1), (k, 2))\}$.
    
    If $\mech$ does not propose any exchanges on round 1, \ie it terminates immediately, then $\datafinal = \initdata$, and the resulting final allocation would not be stable in the event that all agents had truthfully reported, as $\feaspropsp(\datafinal, \partagents)$ is non-empty, and the proposals in $\feaspropsp(\datafinal, \partagents)$ were never rejected by the agents.
    On the other hand, if $\mech$ proposes both exchanges in $\feaspropsp(\initdata, \partagents)$, then agent $k$ would not accept both proposals: precisely, if she accepts both,
    the final allocation would be
    $\datafinali = \datafinalj =  \datafinalk = [1, 1]$, and $k$'s utility will be $2 - 2\betaik - 2\betajk$; however, had she accepted only one of those proposals (say the exchange with $i$), 
    the final allocation would be
    $\datafinali = \datafinalk = [1, 1]$ and $\datafinalj = [1, 0]$, and $k$'s utility will be $2 - 2\betaik - \betajk$ which is larger.
    Hence, $\mech$ should only propose one of the two proposals in $\feaspropsp(\initdata, \partagents)$ in round $1$. Suppose that $\mech$ proposes
    $((i, 1), (k, 2))$ with some probability $\alpha$
    and $((j, 1), (k, 2))$ with probability $1-\alpha$.

    \subparahead{Case (b) and (c)}
    Next, consider Case (b) where $\initdatai = [1, 1]$, $\initdataj = [1, 0]$ and $\initdatak = [0, 1]$.
    In this case, there is only one feasible exchange $\feaspropsp(\initdata, \partagents) = \{((j, 1), (k, 2))\}$. To achieve stability, $\mech$ has to propose this exchange.
    Similarly, in Case (c) where $\initdatai = [1, 0]$, $\initdataj = [1, 1]$ and $\initdatak = [0, 1]$.
    $\mech$ has to propose the single feasible
    exchange in $\feaspropsp(\initdata, \partagents) = \{((i, 1), (k, 2))\}$.

    \subparahead{Contradiction}
    First assume $\alpha > 0$ (from Case (a)) 
    and that \emph{true} initial allocations are
    $\inityi = [1, 1]$, $\inityj=[1, 0]$, $\inityk=[0, 1]$.
    Suppose that $j, k$ are reporting truthfully and following the accepting  policy.
    In this case, if $i$ truthfully report $\initdatai = [1, 1]$, via our analysis of Case (b), agents $j$ and $k$ receive an additional good, and her utility 
    is $U_i \defeq 2 - 2\betaij - 2\betaik$.
    On the other hand, had she untruthfully reported
    $\initdatai = [1, 0]$ and followed the accepting policy, she will prevent $j$ from receiving an additional good with probability $\alpha$;
    hence, her expected utility will be
    $(1-\alpha)(2 - 2\betaij - 2\betaik) +
    \alpha(2 - \betaij - 2\betaik) = 2 - (2-\alpha)\betaij - 2\betaik > U_i$.

    Next, if $\alpha=0$, a similar argument reveals that $j$ may be able to hide her goods and achieve a higher utility.
    Specifically, suppose the \emph{true} initial allocations are
    $\inityi = [1, 0]$, $\inityj=[1, 1]$, $\inityk=[0, 1]$,
    and that $i, k$ are reporting truthfully and accepting all proposals.
    If $j$ reports truthfully, her utility is
    $U_j \defeq 2 - 2\betaij - 2\betajk$.
    However, had she untruthfully reported
    $\initdataj = [1, 0]$ and followed the accepting policy, she will prevent $i$ from gaining additional goods, and her utility will be
    $2 - \betaij - 2\betaik > U_j$.
\end{proof}

\subsection{Proof of Observation~\ref{obs:singleround}}
\label{sec:singleround}

Next, we prove Observation~\ref{obs:singleround}, which shows that a single-round protocol cannot be both stable and NIC.
Recall from Desideratum~\ref{des:stability} that stability requires for all possible agent strategies, 
every remaining mutually beneficial exchange be rejected by the agents upon termination.
Hence,
if a protocol were to terminate after a single round, to achieve stability, it will need to propose all feasible exchanges as agents may reject some of them.
However, such a protocol will not be NIC as a rational agent will not accept multiple proposals in which she is receiving the same good, as she may be giving away multiple goods in exchange for just one.
The following proof formalizes this intuition (also see Remark~\ref{rem:singleroundoutside}).

\begin{proof}[Proof]
Say all three agents participate, so that $\allagents = \partagents = \{i,j,k\}$, with initial data $\initdatai = \initdataj = [1,0]$ and $\initdatak = [0,1]$. There are exactly two feasible exchanges, $\feasprops(\partagents) = \feaspropsp(\initdata, \partagents) = \{((i,1),(k,2)),\,((j,1),(k,2))\}$.
If the planner proposes both exchanges in round~1, the protocol is not NIC: agent~$k$ would not accept both. If she did, the final allocation would be $\datafinali = \datafinalj = \datafinalk = [1,1]$, giving her utility $2 - 2\beta_{ik} - 2\beta_{jk}$; by instead accepting only one proposal (say, with~$i$), the final allocation becomes $\datafinali = \datafinalk = [1,1]$ and $\datafinalj = [1,0]$, yielding utility $2 - 2\beta_{ik} - \beta_{jk} > 2 - 2\beta_{ik} - 2\beta_{jk}$.  

If the planner proposes only one of the two exchanges, the protocol produces a stable allocation if agents accept this proposal. However, rejection of this proposal by any agent leads to termination with $\datafinal = \initdata$. Since $\feaspropsp(\datafinal, \partagents)$ has two elements and only one of them has been rejected, the protocol is not stable in this case, and thus fails to satisfy stability for all agent strategies.
Likewise, if the planner were to propose no exchanges at all, the protocol would also be unstable.
\end{proof}

\begin{remark}
\label{rem:singleroundoutside}
It is worth noting that the above proof relies on the requirement in Desideratum~\ref{des:stability} that, upon termination, any mutually beneficial exchange $p \in \feaspropsp(\datafinal(M,\agsubset,\stratpart), \agsubset)$ has already been rejected. 
From a practical perspective, this ensures that strategic agents cannot reject proposals from the protocol in anticipation of more favorable exchanges they could execute outside the protocol.

To illustrate, suppose we relax Desideratum~\ref{des:stability} to require stability of the terminal allocation only when all agents follow the accepting policy, 
\ie only require $\feaspropsp(\datafinal(M,\agsubset,\stratpart), \agsubset) = \emptyset$ under the accepting strategies.
Consider again the example in the proof of Observation~\ref{obs:singleround} and assume $\betajk < \betaik$.
A single-round protocol that proposes only $((i,1),(k,2))$ satisfies this weaker notion of stability.
It is also IR and NIC as both $i$ and $k$ benefit.
However, $k$ could reject this proposal and privately arrange $((j,1),(k,2))$ with the less competitive agent $j$.
Hence, if agents can act outside the planner's protocol, this would not be a good exchange protocol.

One could alternatively model agent behavior outside the protocol (where agents directly propose and accept exchanges), but this would be cumbersome and significantly complicate the model.
Our stronger stability condition---while meaningful on its own right---also rules out such outside exchanges.
\end{remark}

\section{On Pareto-efficiency}
\label{sec:pediscussion}

In noncompetitive models for the exchange of goods, \emph{stability}---a state in which agents cannot benefit from further exchanges---implies Pareto efficiency (PE). In contrast, we show that in highly competitive settings, these two properties can diverge: an unstable allocation may still be PE, while a stable allocation may fail to be PE.
To illustrate this tension, consider an example with three agents $i, j, k$ and three goods, where $\initdata_i = [1, 0, 0]$, $\initdata_j = [0, 1, 0]$, $\initdata_k = [0, 0, 1]$, and $\beta_{ij} = \beta_{ik} = \beta_{jk} = 0.9$. The initial allocation is clearly unstable, as any two agents are incentivized to exchange goods (\eg $i$ will trade good 1 with $j$ for good 2). The only stable allocation through pairwise exchanges is $x_i = x_j = x_k = [1, 1, 1]$, where all agents possess all goods. However, in this allocation, each agent's utility is $3 - 3 \times 0.9 - 3 \times 0.9 = -2.4$, strictly worse than the utility of $-0.8$ in the initial (unstable) state. Thus, the stable allocation is Pareto-dominated by the initial one.

The above example illustrates a highly competitive setting (large $\betaij$'s) where stability and PE diverge.
We will show that even in low-competition regimes (small $\betaij$'s), stability does not imply PE.
We conjecture, however, that \algname{} achieves PE in such low-competition regimes. We provide partial support by proving this under an intermediate conjecture on \algname{}'s behavior.

\subparahead{Definitions}
We will now formalize these intuitions. To simplify the exposition, we assume all agents participate (\ie $\partagents = \allagents$) and that each good is initially held by at least one agent (\ie $\forall r \in [m], \exists\, i \in \allagents$ such that $\initdatair = 1$). The latter is without loss of generality since if no agent initially holds a good, it cannot be part of any exchange.

Let us also define Pareto efficiency.
An allocation $x \in \{0,1\}^{|\allagents| \times \numgoods}$ \emph{Pareto-dominates} another allocation $x'$ if every agent has equal or higher utility in $x$ than in $x'$, and at least one agent has strictly higher utility; that is, $\utilalloci(x) \geq \utilalloci(x')$ for all $i\in\allagents$ and $\utilalloci(x) > \utilalloci(x')$ for some $i\in\allagents$. An allocation is \emph{Pareto-efficient} if it is not Pareto-dominated by any other allocation.

Next, we identify two (non-exhaustive) regimes characterizing the level of competition among agents. We refer to a set of agents as a \emph{low externality coalition} (LEC) if, upon giving each agent one additional good they do not currently possess, all agents are made better off.
Intuitively, while agents still wish to have more goods for themselves, they all benefit if each receive an equal amount of additional goods.
Conversely, we call them a \emph{high externality coalition} (HEC) if, upon giving each agent one additional good they do not currently possess, no agent is made better off.
Using~\eqref{eqn:utilsummary}, the following definition provides an equivalent characterization of these conditions.

\begin{definition}[LEC/HEC]
A set of agents $\allagents$ is an LEC if $1 - \sum_{j \neq i} \beta_{ij} > 0$ for all $i \in \allagents$. Similarly, $\allagents$ is an HEC if $1 - \sum_{j \neq i} \beta_{ij} \leq 0$ for all $i \in \allagents$.
\end{definition}

\subsection{Pareto-efficiency in HECs}

The following theorem illustrates that Pareto efficiency (PE) is an ill-suited objective in high externality coalitions. In particular, the initial allocation is already Pareto-efficient, and it is impossible to make all agents better off by simply giving them additional goods. In fact, as the example above demonstrates, agents can be strictly worse off in a stable state compared to the beginning.

\begin{restatable*}{theorem}{pehardness}
\label{thm:pehardness}
Suppose $\allagents$ is an HEC and that none of the agents possess all goods initially. Then, from any initial allocation $\initdata$, any allocation $x$ in which agents receive additional goods does not Pareto-dominate $\initdata$.
\end{restatable*}

\begin{proof}
First, consider any other allocation $\data$ where agents have received additional goods; since goods cannot be taken away from the initial allocation, each agent has at least as many goods as in $\initdata$.
Let us count how many more goods each agent has in $\data$ compared to $\initdata$, and identify the agent $i$ who gains the fewest additional goods---denote this number by $q$ (with $q \geq 0$).

Next, consider an allocation $\data'$ obtained by starting from $\initdata$ and giving each agent exactly $q$ additional goods. By the HEC conditions, it follows that each agent in $\data'$ is not better off than in $\initdata$, since each agent loses utility
$
q \left(1 - \sum_{j \neq i} \beta_{ij} \right) \leq 0
$.
Now we compare $\data'$ and $\data$. If $\data = \data'$, we have established that $\data'$ does not Pareto-dominate $\initdata$. If $\data \neq \data'$, agent $i$’s utility is worse in $\data$ than in $\data'$ as other agents have gained strictly more goods.
Hence, $\data$ does not Pareto-dominate $\initdata$.
\end{proof}

\subsection{Pareto-efficiency in LECs}

Next, we consider the case where the agents form a low externality coalition.
We believe that \algname{} achieves PE allocations in this case.
We support this claim with a partial proof, under the following conjecture about \algname's behavior.

\begin{conjecture}
    \label{conj:clearperfagent}
    Suppose all agents follow the accepting policy. Then, for any initial allocation $\initdata$, in the final allocation $\datafinal$ produced by \emph{\algname}, there is at least one agent who is allocated all the goods, \ie there exists $i\in\allagents$ such that $\datafinalir = 1$ for all $r\in[m]$.
\end{conjecture}

We have verified this conjecture extensively via simulations in Appendix~\ref{app:simulations}.
The following theorem shows that \algname{} finds PE allocations when this conjecture is true.

\begin{restatable*}{theorem}{clearpelec}
    \label{thm:clearpelec}
    Suppose the agents are an LEC and they all participate with the accepting policy.
    If Conjecture~\ref{conj:clearperfagent} is true, then the final allocation produced by
    \emph{\algname{}} is Pareto-efficient.
\end{restatable*}

The following lemma characterizes necessary and sufficient conditions for Pareto-efficiency when the agents form an LEC.
Theorem~\ref{thm:clearpelec} is a direct consequence of the sufficient condition, while we will use the necessary condition for the ensuing discussion.

\begin{lemma} \label{perfect-agent-PE}
    In an LEC, an allocation $\data \in\{0,1\}^{|\allagents| \times  \numgoods}$ is Pareto-efficient if and only if at least one agent has received all goods.
\end{lemma}

\begin{proof}
The necessary condition follows directly from the definition of an LEC: if none of the agents have all the goods, we can make a Pareto-dominating allocation by giving each agent one more good they do not have.
Next, we make the following claim:
``if an allocation $x$ Pareto dominates allocation $x'$, then every agent must have strictly more goods in $x$ than in $x'$''.
The sufficient condition follows from this claim.
In particular, if an agent is allocated all the goods in some allocation $x$,
in any other allocation $x'$, this agent cannot receive strictly more goods; as this allocation cannot be Pareto-dominated, it is Pareto-efficient.

All that is left to do is to prove the claim above, which we will do by way of contradiction.
Assume that not all agents in $x$ have strictly more goods than in $x'$. 
First, let us compare the number of goods each agent holds in $x$ versus $x'$, and identify an agent $k$ who loses the greatest number of goods in $x$ compared to $x'$, \ie $k = \argmax_{i\in\allagents}(\onev^\top \datai' - \onev^\top \datai)$. Let $q \geq 0$ denote the number of goods that $k$ loses. Consider an intermediate allocation $x''$, obtained by starting from $x'$ and removing $q$ goods from each agent.

By the LEC condition, each agent’s utility in $x''$ decreases by $-q(1 - \sum_{j \neq i} \beta_{ij}) \leq 0$ compared to $x'$, so no agent is better off in $x''$ than in $x'$. If $x = x''$, then clearly $x$ does not Pareto dominate $x'$, as no agent strictly improves. 
Otherwise, if $x \neq x''$, then some agent $j \neq k$ must have strictly more goods in $x$ than in $x''$, while agent $k$ has exactly the same number of goods in $x$ and in $x''$. This implies that $k$’s utility 
is strictly lower in $x$ than in $x''$, as other agents have more goods.
Therefore, $x$ does not Pareto-dominate $x'$ as well.
\end{proof}

\parahead{Stability does not imply Pareto-efficiency}
We conclude this section with an example which illustrates that not all stable exchange protocols are Pareto-efficient, even in LECs.
It also illustrates that \retrospect{} may play an important role in satisfying Conjecture~\ref{conj:clearperfagent} (in addition to NIC).

\insertFigPO

\begin{example}
Suppose we have three agents $\{i, j, k\}$ and $m = 9$ goods, as illustrated in Figure~\ref{fig:POexample}. Let $\initdata_i = [0, 0, 1, 1, 1, 1, 0, 0, 0]$, $\initdata_j = [0, 0, 0, 0, 0, 0, 1, 1, 1]$, and $\initdata_k = [1, 1, 1, 1, 1, 0, 0, 0, 0]$, with competition factors satisfying $\beta_{ij} < \beta_{ik} < \beta_{jk}$. Assume all agents follow the accepting policy.
Suppose we run \emph{\algname{}} without \emph{\retrospect{}};
this exchange protocol is provably stable since the stability proof does not involve \emph{\retrospect{}}.
This protocol schedules the following exchanges: between $i$ and $j$, it schedules $((i, 3), (j, 7))$, $((i, 4), (j, 8))$, and $((i, 5), (j, 9))$; between $i$ and $k$, it schedules $((i, 6), (k, 1))$; and between $j$ and $k$, it schedules $((j, 7), (k, 1))$ and $((j, 8), (k, 2))$. The final allocation is stable, as no two agents can beneficially trade goods they initially owned (see Figure~\ref{fig:PObad}).
However, it is not PE since all agents still lack a good (see Lemma~\ref{perfect-agent-PE}).

The failure to produce an agent with all goods stems from suboptimal exchange choices. For instance, if $j$ had received good 6 from $i$ instead of good 3, she could have later offered good 9 to $k$ in exchange for good 3. Interestingly, \emph{\retrospect{}} rearranges the exchanges to enable such improvements (see Figure~\ref{fig:POgood}). While its primary goal is to ensure NIC through this rearranging, we find that it also helps \emph{\algname{}} to produce a Pareto-efficient final allocation as a consequence.
\end{example}

\subsection{Pareto-efficiency outside LECs and HECs}

Finally, we consider cases where agents form neither HECs nor LECs. In these settings, as the following example demonstrates, the potential to improve everyone's utility from the initial state through additional exchanges may depend on the initial allocation~$\initdata$.

\begin{example} \label{neither-LEC-nor-HEC}
    Suppose we have $m=3$ goods and five agents $\{i,j,k, l, h\}$.
    Suppose $\beta_{ij} = \beta_{ik} = \beta_{il} = \beta_{ih} = 0.26$ and
    $\beta_{jk} = \beta_{jl} = \beta_{jh} = \beta_{kl} = \beta_{kh} = \beta_{lh} = 0.01$.
    That is, agent $i$ is highly competitive with the others while the rest have low competition with each other.
    These agents are neither an LEC nor an HEC because if we give all of them an additional good, 
    $i$'s utility decreases by $1-4\times0.26 = -0.04$,
    while other agents' utilities increase by $1-0.26-3\times0.01 = 0.71$.

    First assume that $\initdata_i = [0, 0, 1]$, $\initdata_j = \initdata_k = \initdata_l = [1, 1, 0]$, and $\initdata_h = [1, 0, 1]$.
    If we run \emph{\algname{}} on this instance (with tie-breaking rules about $\beta_{ij}$'s on line \ref{increasing-beta}), it will schedule the following exchanges:
    $((j, 2), (h, 3))$, $((i, 3), (k, 1))$, and $((i, 3), (l, 2))$. This results in a final allocation where all agents receive all goods.
    The utility change is $2-4\times0.26 = 0.96$ for $i$ and $1-2\times0.26-3\times0.01 = 0.45$ for the others.
    Hence, this final stable allocation Pareto-dominates the initial allocation, and is in fact Pareto-efficient.
    In this example, by gaining 2 goods, $i$ offsets her high competition with others.

    Next, let us consider a different initial allocation where $i$ also holds good 2, \ie $\initdatai = [0, 1, 1]$, while others have the same initial goods.
    Clearly the initial allocation is not stable.
    However, no stable allocation can Pareto-dominate the initial allocation: if all agents receive an additional good, $i$ is worse off, whereas if only a subset of the agents receive an additional good, those who did not receive goods are worse off.
    While some stable allocations may be Pareto-efficient, so is the unstable initial allocation.

\end{example}

\section{Discussion}
\label{sec:discussion}

We briefly discuss some alternative modeling choices.

\parahead{Core-like stability}
One could define stability using a core-based notion from coalitional games~\citep{shapley1974cores, shapley1971core}, but
this too is unsuitable under large negative externalities. A subset of agents might break away from the protocol (grand coalition) to obtain higher utility, but this is only because it reduces overall sharing.
Moreover, its members need not remain committed to the break-away coalition since
they can benefit by trading outside the coalition. 
Hence, the core fails to capture stability in highly competitive settings.
We illustrate this via the following example.

\begin{example}
Consider an instance with 4 agents $\{i, j, k, l\}$ and 4 goods, with each agent possessing exactly one unique good, \ie $\initdatai = [1, 0, 0, 0]$, $\initdataj = [0, 1, 0, 0]$, etc.
Moreover, say that the competition factors ($\beta$'s) are $0.9$ for all pairs.
In this setting, under the accepting strategy profile, any stable protocol must produce the allocation where all agents possess all goods and each agent has utility $4 - 12 \times 0.9 = -6.8$.
The core would require that no subset of agents be better off if they choose to leave the protocol and conduct pairwise exchanges within themselves.
However, in this example, any 3 agents could benefit by leaving the protocol and exchanging privately among themselves; for example, if $\{i, j, k\}$ leaves and exchanges privately, each of them has 3 goods in the final allocation, while $l$ has only 1, resulting in a higher $3 - 7 \times 0.9 = -3.3$ utility for each of $i,j$, and $k$. However, the $\{i, j, k\}$ coalition is also unstable, because any 2 of them, say $i$ and $j$, also have the incentive to leave and exchange privately. Doing so will result in $i$ and $j$ each having 2 goods, while $k$ and $l$ each have 1, yielding an even higher $2 - 4 \times 0.9 = -1.6$ utility for $i$ and $j$.
More importantly, agents who break away from the protocol are not necessarily incentivized to remain loyal to the break-away coalition; for instance, in the scenario where $i,j,k$ break away from the protocol, agent $i$ could simply go back and exchange with $l$ to increase her utility.
\end{example}

\parahead{Allowing onward sharing of goods}  
In our model, agents cannot share goods they have received from others. Allowing onward sharing would fundamentally alter the dynamics. We outline some key aspects of this setting, leaving a detailed analysis to future work.

At first glance, onward sharing might appear to simplify the problem. For an allocation to be stable, the goods held by any pair of agents must form a subset relation; otherwise, agents could continue exchanging goods they lack. This subset structure seemingly makes it easier for a planner to characterize and construct stable allocations.

However, such a setting complicates agents' incentives. Each exchange now has greater potential value, as received goods can be shared further to increase utility. Agents thus become selective in both sharing and receiving goods. Specifically, they are more willing to share widely-held  good, but reluctant to receive them, since these provide few future exchange opportunities. Conversely, they prefer to receive scarce  goods but hesitate to share them, as doing so grants more opportunities to others than to themselves. These incentive considerations significantly complicate the problem and pose challenges for designing protocols that ensure incentive compatibility.

\section{Run time of \algname{}}
\label{sec:runtime}

First, in~\S\ref{app:retrospectoptimization}, we describe three optimizations to the description of \algname{} in Algorithm~\ref{alg:main} to improve the runtime \retrospect{}.
Then, in~\S\ref{app:runtimeafteropt}, we will analyze the runtime after these optimizations.

\subsection{Optimizations to \retrospect{}}
\label{app:retrospectoptimization}

\subparahead{(1) Representation of allocations}
To efficiently perform line~\ref{query} of \retrospect{}, where we must identify the original owner of a good shared with an agent, we enhance the representation of the allocation. Instead of using $\data_{i, s} \in \{0, 1\}$ to indicate whether agent~$i$ possesses good~$s$, we use $x_{i, s} \in \{0, 1, 2, \ldots, n\}$ to also encode the source of the good (where $n = |\allagents|$). Specifically, $x_{i, s} = 0$ means that $i$ does not hold good~$s$, while $x_{i, s} = j$ for $j \in \{1, \ldots, n\}$ indicates that $i$ received good~$s$ from agent~$j$, including the case where $j = i$ (\ie $i$ originally owned it). We apply the same convention to $\interdata$. This allows us to identify the original sharer of any good in constant time.

\subparahead{(2) Skipping \emph{\retrospect{}} for agents who have already received all goods}
If an agent is already scheduled to receive all goods while \algname{} is running, \ie if $\interdata_{i} = \mathbf{1}$, then we can safely skip any \retrospect{} calls where agent~$i$ is the receiving agent. To see why, note that for \retrospect{} to succeed (\ie not return \retrofail{}), it must eventually reach line~\ref{base-case} in some recursive call, which requires the condition $\interdata_{i,s} = 0$ for some good~$s$. However, if $\interdata_{i,s} = 1$ for all $s$, this condition can never be satisfied. Thus, \retrospect$(i, j, S)$ will necessarily return \retrofail{} for any $j$ and $S$, and the call can be skipped.

\subparahead{(3) Skipping redundant \emph{\retrospect{}} calls}
Consider an initial call \retrospect($i$, $j$, $\{i, j\}$) from lines~\ref{IC1} or~\ref{IC2} of the main algorithm. Suppose that within this call, a recursive call \retrospect($i$, $k$, $S \cup \{k\}$) returns \retrofail{} for some agent $k \neq i$ and some set $S \subset \partagents \setminus \{k\}$ (a \retrospect{} call where $k$ is the giving agent). Then, for the remainder of the initial call, we skip all further \retrospect{} calls where $k$ is the giving agent.
Lemma~\ref{repeating-target} guarantees that this optimization does not affect the behavior of \retrospect{}. As we will see shortly, it significantly improves runtime by reducing the number of recursive calls made by an initial \retrospect{} call. This optimization can be implemented by maintaining a global variable tracking all failed agents $k$ outside the \retrospect{} procedure.

    \begin{lemma} \label{repeating-target}
        The optimization outlined in \emph{(3)} above does not change the behavior of
        \emph{\retrospect{}}.
    \end{lemma}
    \begin{proof}
        Consider the following graphical view when we call \retrospect($i$, $j$, $\{i,j\}$). Our goal is to make the recursion resemble a DFS on the constructed graph.
        We will make each agent except $i$ a separate node. Agent $j$ will be the source node. The set of agents $T$ who can offer $i$ goods on its empty slots would be the sinks, \ie we can find a good $s$ s.t. $\interdata_{i,s} = 0, \initdata_{t, s} = 1$ for any $t \in T$. We make a directed edge from agent $l$ to $l'$ if we can find good $s$ s.t. $\initdata_{l, s} = 1, \interdata_{i,s} = 1$ and $l'$ is the agent who gives $s$ to $i$ in this round.
        
        Now we explain why this construction is valid. First, the giving agent of the initial call is $j$, so we start our search at $j$. When a recursive \retrospect{} call has $t \in T$ as the giving agent, we reach the base case for success on line \ref{base-case}, so each $t$ is a valid sink representing success.
        Each edge from $l$ to $l'$ means that we can potentially make a recursive \retrospect{} call where $l'$ is the giving agent in a \retrospect{} call where $l$ is the giving agent. The condition for making these edges is on line \ref{query}.
        We want to find a path from the source $j$ to any sink $t \in T$. If we find such a path, the initial \retrospect{} call succeeds. Otherwise, it fails. The \retrospect{} subroutine, then, can be viewed as a DFS on this graph.
          
        Say we started searching from $j$ and get to $k$, somewhere in the middle. \retrospect$(i, k, S \cup \{k\})$ returning \retrofail{} is the same as saying that we cannot get to a sink through unexplored nodes ($\notin S$) starting from $k$. Then, either there is absolutely no path from $k$ to any $t\in T$, or there is such a path but it uses explored nodes ($\in S$). We examine each case separately.
        In the first case, we know that future searches reaching $k$ would return \retrofail{} as well, so there is no need to visit $k$ again. In the second case, say the path goes through some already explored node $j' \in S$, \ie it is of the form $k \rightarrow j' \rightarrow t$. Note that we must have gone through $j \rightarrow j' \rightarrow k$ before reaching $k$ because $j' \in S$ is an explored node. Then, we can cut out the cycle $j' \rightarrow k \rightarrow j'$ and directly follow $j\rightarrow j' \rightarrow t$, \ie the search starting from $j'$ will also find this path and succeed.
        Therefore, it is redundant to keep looking at $k$. In both cases, we can safely skip subsequent \retrospect{} calls where $k$ is the giving agent.
    \end{proof}

    \subsection{Runtime after optimization}
    \label{app:runtimeafteropt}

In each round, \algname{} considers $O(|\partagents|^2)$ agent pairs. For each pair, the exchange loop on line~\ref{exchange-loop} may run up to $O(m)$ times in the worst case, as each agent can acquire up to $m$ goods. This results in an overall $O(|\partagents|^2 m)$ multiplier on the work performed within each iteration of the loop.

Each iteration consists of two components. First, we construct demand sets, which takes $O(m)$ time. Second, we either find a direct exchange on line~\ref{find-exchange} (with worst-case cost $O(m^2)$), or invoke a constant number of \retrospect{} calls on lines~\ref{IC1} and~\ref{IC2}. By Lemma~\ref{repeating-target}, a single initial \retrospect{} call generates at most $O(|\partagents|)$ recursive calls, and each call performs a linear scan over at most $m$ goods on line~\ref{all-possible}, for a total cost of $O(|\partagents|m)$. 

Therefore, the total work per iteration is $O(m + \max\{m^2, |\partagents|m\}) = O(m \cdot \max\{m, |\partagents|\})$. Multiplying by the outer loop factor, the worst-case runtime of a single round of \algname{} is $O(|\partagents|^2 m^2 \cdot \max\{m, |\partagents|\})$.
Note that optimization (1) is used to make each \retrospect{} call run in linear time, and optimization (3) is used to bound the number of recursive \retrospect{} calls.
Optimization (2) does not provide a tighter theoretical bound, but it significantly improves runtime in practice.

\section{Simulations}
\label{app:simulations}

\insertSimulation

We generated random instances $(\initdata, \beta)$ using three parameters: the number of agents $n = |\allagents|$, the number of goods $m$, and the probability $p$ that an agent initially holds a given good (each agent has each good with probability $p$). We ran \algname{} on various combinations of $n, m, p$ with up to $10^6$ trials per setting as shown in Table \ref{tab:simulation}. In all cases, \algname{} produced a final allocation in which at least one agent received all goods, thus empirically supported Conjecture \ref{conj:clearperfagent}.

We make a few other remarks about the simulation setting and simulation results.
First, the simulation code is shared
\href{https://osf.io/gwcrd/?view_only=d7456d050dd94ed39cbf28631f9fd26f}{\textcolor{blue}{here}} to facilitate reproducibility. There is only one file containing all the simulation code and instructions to run the simulation. Since the purpose of this code is to test against Conjecture \ref{conj:clearperfagent}, we implemented a single-round version of \algname{} (produces a final allocation assuming all agents accept) and enforced \algname{} to make a particular choice when multiple exchanges are available (``tie-breaking'').

We pick parameters $n, m, p$ to cover a diverse set of instances. The choice of $p \in \{0.1, 0.5, 0.9\}$ resembles different levels of data density in the initial allocation.

In order to keep the code concise, we implemented optimization (2) but not optimization (3). We noticed that the runtime increases significantly for large instances with $p = 0.1$ and suspect that \algname{} could suffer from deep recursion due to the lack of optimization (3). Because of this, we put N/A for these settings in Table \ref{tab:simulation}. On the other hand, for small instances or $p=0.9$, \algname{} typically terminates much faster than the worst case.

\end{document}